\PassOptionsToPackage{svgnames}{xcolor}

\documentclass[acmsmall,screen,nonacm]{acmart}

\title{Quantum Orchestras: a Concrete Semantics for Recursive Hybrid Programs}

\author{Alex Rice}
\orcid{0000-0002-2698-5122}
\affiliation{
  \institution{The University of Edinburgh}
  \country{United Kingdom}
}
\email{alex.rice@ed.ac.uk}

\author{Dominik Leichtle}
\orcid{0000-0002-4882-889X}
\affiliation{
  \institution{The University of Edinburgh}
  \country{United Kingdom}
}
\email{dominik.leichtle@ed.ac.uk}

\author{Kim Worrall}
\orcid{0009-0008-0666-9203}
\affiliation{
  \institution{The University of Edinburgh}
  \country{United Kingdom}
}
\email{kim.worrall@ed.ac.uk}

\author{Robert I. Booth}
\orcid{0000-0002-1146-3380}
\affiliation{
  \institution{University of Oxford}
  \country{United Kingdom}
}
\email{firstname.lastname@cs.ox.ac.uk}

\usepackage{quiver}
\tikzset{far/.style={outer sep=4pt}}
\usepackage{mathtools}
\usepackage{bbm}
\usepackage{bm}
\usepackage{enumitem}
\usepackage{braket}
\usepackage{mathpartir}
\usepackage{quantikz}
\usepackage{subcaption}
\let\truewire\wire

\tikzcdset{%
	row sep/normal=1.8em,
	column sep/normal=2.4em,
  nodes in empty cells/.style={},
  every matrix/.style={},
  every cell/.style={},
}
\let\wire\relax

\tikzcdset{%
	quantikz every matrix/.style={},
	quantikz every cell/.style={
        anchor=center,minimum size=0pt,inner sep=0pt,outer sep=0pt,thick,
        minimum height=\ifcsundef{cell@height@\the\pgfmatrixcurrentrow-\the\pgfmatrixcurrentcolumn}{0}{\csname cell@height@\the\pgfmatrixcurrentrow-\the\pgfmatrixcurrentcolumn\endcsname},
        minimum width=\ifcsundef{cell@width@\the\pgfmatrixcurrentrow-\the\pgfmatrixcurrentcolumn}{0}{\csname cell@width@\the\pgfmatrixcurrentrow-\the\pgfmatrixcurrentcolumn\endcsname},
        execute at end cell={\wire{a}
        \csgundef{cell@width@\the\pgfmatrixcurrentrow-\the\pgfmatrixcurrentcolumn}
      \csgundef{cell@height@\the\pgfmatrixcurrentrow-\the\pgfmatrixcurrentcolumn}
      \csgundef{wire@type@override@\the\pgfmatrixcurrentrow-\the\pgfmatrixcurrentcolumn}}
    },
  quantikz arrows/.style={dash,thick,line cap=round},
  quantikz 1-row diagram/.style={/tikz/baseline={([yshift=-\MathAxis]\tikzcdmatrixname-1-1.center)}},
  quantikz grade/.style={description,circle,draw,fill=white,inner sep=0pt, outer sep=0pt,font=\tiny}
}

\pretocmd{\quantikz}{%
  \let\wire\truewire
  \tikzcdset{
    row sep/normal=0.5cm,
    column sep/normal=0.5cm,
    every cell/.style={quantikz every cell},
    every matrix/.style={quantikz every matrix}
  }%
}{}{}

\apptocmd{\endquantikz}{%
  \tikzcdset{
    row sep/normal=1.8em,
    column sep/normal=2.4em,
    nodes in empty cells/.style={},
    every matrix/.style={},
    every cell/.style={}
  }%
  \let\wire\relax
}{}{}

\DeclareMathOperator{\id}{id}
\DeclareMathOperator{\Supp}{Supp}
\DeclareMathOperator{\ev}{ev}

\newcommand{\csaa}{{\mathscr{A}}}
\newcommand{\csab}{{\mathscr{B}}}
\newcommand{\csac}{{\mathscr{C}}}
\newcommand{\csad}{{\mathscr{D}}}
\newcommand{\csae}{{\mathscr{E}}}

\newcommand{\csak}{{\mathscr{K}}}

\newcommand{\ketbra}[1]{\ket{#1}\!\!\bra{#1}}
\newcommand{\orch}{\mathcal{Q}}
\newcommand{\valu}{\mathcal{V}}

\newcommand{\semctx}{\rho}

\newcommand{\sem}[1]{\llbracket #1 \rrbracket}
\renewcommand{\tt}{\textsf{True}}
\newcommand{\ff}{\textsf{False}}
\newcommand{\bool}{\textsf{bool}}
\newcommand{\fix}{\mathop{\textsf{fix}}}
\newcommand{\fst}{\mathop{\textsf{fst}}}
\newcommand{\snd}{\mathop{\textsf{snd}}}
\newcommand{\seq}{;}

\newcommand{\termif}[3]{\mathop{\textsf{if}}#1\mathop{\textsf{then}}#2\mathop{\textsf{else}}#3}
\newcommand{\letin}[2]{\mathop{\textsf{let}}#1=#2\mathop{\textsf{in}}}
\newcommand{\Had}{\textsc{H}}
\newcommand{\CNOT}{\textsc{CX}}
\newcommand{\SG}{\textsc{S}}
\newcommand{\T}{\textsc{T}}
\newcommand{\X}{\textsc{X}}
\newcommand{\meas}{\textsc{meas}}
\newcommand{\alloc}{\textsc{alloc}}

\renewcommand{\Set}{\mathbf{Set}}
\newcommand{\cstar}{C\textsuperscript{\(*\)}}
\newcommand{\wstar}{W\textsuperscript{\(*\)}}
\newcommand{\wstarcat}{\mathbf{W^{\bm{*}}}}
\newcommand{\dcpo}{\mathbf{DCPO}}
\newcommand{\measurable}{\mathbf{Meas}}

\newcommand{\kleisli}[1]{#1^{\sharp}}

\newcommand{\qid}{\textsf{qid}}
\newcommand{\qram}{\mathbf{Q}}
\newcommand{\unit}{\textsc{1}}

\newcommand{\return}{\mathop{\textsf{return}}}
\newcommand{\judgementvalue}{\vdash^v}
\newcommand{\judgementproducer}{\vdash^e}

\usepackage[svgnames]{xcolor}
\makeatletter
    \newcommand{\colorboxed}[3][white]{\fcolorbox{#2}{#1}{\m@th$\displaystyle#3$}}
\makeatother

\usepackage{ifdraft}
\ifdraft{
\paperwidth=\dimexpr \paperwidth + 6cm\relax
\oddsidemargin=\dimexpr\oddsidemargin + 3cm\relax
\evensidemargin=\dimexpr\evensidemargin + 3cm\relax
\marginparwidth=\dimexpr \marginparwidth + 3cm\relax
}{}

\usepackage[obeyDraft,colorinlistoftodos]{todonotes}
\todostyle{alex}{color=cyan}
\todostyle{dom}{color=orange!35!yellow}
\todostyle{kim}{color=green!35!yellow}
\todostyle{rob}{color=purple!60}
\todostyle{last}{color=pink!60}

\usepackage{amsthm}
\usepackage{hyperref}
\usepackage[capitalize]{cleveref}
\usepackage{thmtools}

\hypersetup{final}

\theoremstyle{plain}
\newtheorem{theorem}{Theorem}

\newtheorem{lemma}[theorem]{Lemma}
\newtheorem{proposition}[theorem]{Proposition}
\theoremstyle{definition}
\newtheorem{definition}[theorem]{Definition}
\newtheorem{example}[theorem]{Example}
\newtheorem{remark}[theorem]{Remark}

\begin{abstract}
  Many production quantum programming languages represent hybrid quantum computations by extending a classical base language with a quantum effect, where qubits are addressed by reference, and quantum operations are understood to mutate some external quantum state. However, the semantics of this view of quantum computation remains underdeveloped, especially when the language allows mid-circuit measurements and non-termination.

  In this work, we provide a general method for building denotational semantics for such languages, by defining the \emph{quantum orchestra} monad, which precisely captures this style of quantum effect. The monad has a concrete presentation, being based on the formalism of quantum instruments, a common tool in quantum information theory for capturing the action of a quantum process along with its classical outcomes. It acts on the category \(\dcpo\), and so enables the interpretation of divergent hybrid programs.

  The quantum orchestra monad serves as a natural extension of both the classical state monad and the probabilistic powerdomain monad. We investigate some of the subtleties present when trying to na\"ively extend these definitions to the quantum non-commutative case.
\end{abstract}

\begin{document}

\maketitle

\ifdraft{\tableofcontents}

\listoftodos{}

\section{Introduction}

The field of quantum computing is quickly progressing from small scale experiments in physics labs to increasingly capable industrial quantum computers. As these systems grow in scale, the software used to control and program them is becoming correspondingly more complex. This motivates the development of mathematical models that are both expressive and conceptually clear, supporting the design and analysis of modern quantum programming languages.

Traditionally, quantum programs have been expressed using the circuit model~\cite{deutschQuantumComputationalNetworks1989}: a static list of quantum gates is applied sequentially to an input quantum state, which is then measured at the end of the computation. While the semantics of circuits are well understood and their static nature makes it possible for them to be heavily optimised, programming directly in terms of circuits provides little support for programs that adapt their quantum behaviour in response to intermediate measurement outcomes.

More recently, there has been increased interest in hybrid quantum programs, which combine quantum operations with classical computation and control flow (see for example IBM introducing classical registers in OpenQASM 3~\cite{crossOpenQASM3Broader2022}, Quantinuum's HUGR~\cite{kochhugr}, QIRO~\cite{ittahQIROStaticSingle2022} which is used in the Pennylane toolkit~\cite{bergholmPennyLaneAutomaticDifferentiation2022}, and QIR~\cite{QIRSpec2021}). This hybrid functionality is essential for describing many aspects of quantum computing, including variational algorithms~\cite{cerezoVariationalQuantumAlgorithms2021}, measurement-based quantum computing~\cite{raussendorfOneWayQuantumComputer2001}, repeat-until-success circuits~\cite{adamRepeatUntilSuccessNondeterministicDecomposition2014}, error mitigation~\cite{wallmanNoiseTailoringScalable2016}, qubit teleportation~\cite{devulapalliQuantumRoutingTeleportation2024}, iterative phase estimation~\cite{granadeUsingRandomWalks2022}, circuit knitting~\cite{piveteausutter:knitting}, and quantum error correction (e.g. \cite{roffeQuantumErrorCorrection2019}).

At a high level, many such languages follow the quantum random-access machine (QRAM) model of quantum computation~\cite{knillConventionsQuantumPseudocode2022}: a classical host program dynamically controls a quantum processor, issuing quantum commands and observing the outcomes of measurements. This view is particularly natural for languages that extend a classical base language with operations for manipulating quantum state, such as Guppy~\cite{kochGUPPYPythonicQuantumClassical2025} and Silq~\cite{bichselSilqHighlevelQuantum2020}. In this setting, the classical program manipulates references to quantum resources whose underlying state evolves as the computation proceeds.

Despite its practical importance, the denotational semantics of this QRAM-style programming model remains comparatively underdeveloped. Unlike static circuits, hybrid quantum programs combine ordinary classical computation with effectful interaction with an evolving quantum state. Each interaction with the quantum processor may both update the quantum state and return classical information to the host program, which can then influence the rest of the computation. The challenge is to model this interaction compositionally while accommodating richer features of the classical host such as recursion.

A simple example is the reset operation on a single qubit, which maps every input state to \(\ket{0}\). Operationally, reset can be implemented by measuring the qubit and then conditionally applying an \(\X\) gate if the outcome was \(\ket{1}\):
\begin{align*}
  \textsc{reset} \coloneq &\letin b {\meas(q)} \ \{ \; \mathop{\mathsf{if}} b \mathop{\mathsf{then}} \X(q) \; \}
\end{align*}
This program exhibits measurement feedforward: a classical measurement outcome determines which quantum operation is performed next. Although the quantum state transformation induced by measurement can be described by a single channel acting on the quantum system, the program itself cannot be modelled by simply composing such channels, because the choice of the next channel depends on the classical outcome of the previous interaction.

The standard mathematical formalism for this kind of classical-quantum interaction is the \emph{quantum instrument}. Intuitively, an instrument assigns to each possible classical outcome a corresponding quantum state transformation. It therefore records both parts of an interaction with a quantum processor: the classical information returned to the host program, and the effect of that interaction on the quantum state.
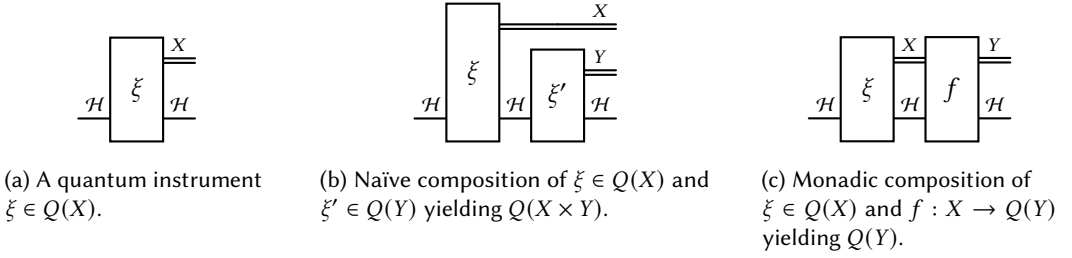
\begin{figure}
  \centering
  \begin{subfigure}[b]{0.25\textwidth}
    \centering
    \begin{quantikz}[row sep=tiny, column sep=small]
      \setwiretype{n} &\gate[2][0.7cm]{\xi} & \wire[l][1]["X"{above,pos=0.5}]{c} \\
      & \setwiretype{n} \wire[l][1]["\mathcal{H}"{above,pos=0.5}]{q} & \wire[l][1]["\mathcal{H}"{above,pos=0.5}]{q}
    \end{quantikz}
    \caption{A quantum instrument\\\(\xi \in Q(X)\).\\}
    \label{fig:quant-inst}
  \end{subfigure}%
  \hfill
  \begin{subfigure}[b]{0.37\textwidth}
    \centering
    \begin{quantikz}[row sep={0cm}, column sep=small]
      \setwiretype{n} &\gate[3][0.7cm]{\xi} & \setwiretype{c} & \setwiretype{n} \wire[l][1]["X"{above,pos=0.26}]{c} \\
      \setwiretype{n} &&\gate[2][0.7cm]{\xi'} & \setwiretype{n} \wire[l][1]["Y"{above,pos=0.5}]{c} \\
      & \setwiretype{n} \wire[l][1]["\mathcal{H}"{above,pos=0.5}]{q} & \wire[l][1]["\mathcal{H}"{above,pos=0.5}]{q} & \wire[l][1]["\mathcal{H}"{above,pos=0.5}]{q}
    \end{quantikz}
    \caption{Na\"ive composition of \(\xi \in Q(X)\) and \(\xi' \in Q(Y)\) yielding \(Q(X \times Y)\).\\}
    \label{fig:inst-comp-naive}
  \end{subfigure}%
  \hfill
  \begin{subfigure}[b]{0.28\textwidth}
    \centering
    \begin{quantikz}[row sep=tiny, column sep=small]
      \setwiretype{n} & \gate[2][0.7cm]{\xi} & \gate[2][0.7cm]{f} \wire[l][1]["X"{above,pos=0.5}]{c} & \wire[l][1]["Y"{above,pos=0.5}]{c} \\
      & \setwiretype{n} \wire[l][1]["\mathcal{H}"{above,pos=0.5}]{q} & \wire[l][1]["\mathcal{H}"{above,pos=0.5}]{q} & \wire[l][1]["\mathcal{H}"{above,pos=0.5}]{q}
    \end{quantikz}
    \caption{Monadic composition of\\\(\xi \in Q(X)\) and \(f : X \to Q(Y)\) yielding \(Q(Y)\).}
    \label{fig:inst-comp-monad}
  \end{subfigure}
  \caption{Sequential composition of quantum instruments. A fixed composition of \(\xi \in Q(X)\) with \(\xi' \in Q(Y)\) cannot express measurement-dependent control. The relevant operation is Kleisli composition, composing \(\xi\) with a continuation \(f : X \to Q(Y)\).}
  \label{fig:quant-inst-compose}
  \Description[Different ways of composing quantum instruments.]%
  {The figure depicts different ways of composing quantum instruments.
  A quantum instrument is depicted as a box with one input quantum wire, one output quantum wire, and one output classical wire.
  One na\"ive way of composing two quantum instruments is by sequentially composing their action on quantum states while forming the product of their classical outputs.
  The monadic composition as considered in this work sequentially composes both the classical and the quantum wires.}
\end{figure}

This also explains why we require a monadic composition for instruments---see \cref{fig:quant-inst-compose}. If \(\xi \in Q(X)\) is an instrument with classical output \(X\), then the next step of the program need not be a fixed instrument \(\xi' \in Q(Y)\). Rather, it may depend on the observed outcome \(x \in X\), and is therefore described by a function \(f : X \to Q(Y)\). Sequential composition of hybrid quantum computations is thus Kleisli composition: first run \(\xi\), observe its classical output, and then continue with the instrument selected by that output. The reset example above is the simplest instance of this pattern: the result of measurement determines whether the subsequent operation is the identity or the \(\X\) gate.

We argue that quantum instruments are the fundamental primitive of QRAM-style quantum programming, and give rise to a natural computational effect: pure computations are classical, while effectful computations interact with the quantum processor through instruments. In fact, we take a further step: we consider not only programs which condition on the results of measurement, but also those with general recursion, whose termination is dependent on the results of measurement. We capture this effect by generalising quantum instruments to our \emph{quantum orchestra} monad over directed complete partial orders, obtaining domain-theoretic semantics of hybrid quantum--classical programs.

\noindent Concretely, the quantum orchestra monad supports:
\begin{itemize}
\item arbitrary unitary evolution of a quantum state;
\item mid-circuit measurement and measurement-dependent control flow;
\item general recursion, via a fixpoint combinator obtained by defining the monad over the category \(\dcpo\) of directed-complete partial orders;
\item qubit allocation, by presenting the monad as a parameterised strong monad~\cite{atkeyParameterisedNotionsComputation2009} over the category \(\wstarcat\) of von Neumann algebras and quantum channels.
\end{itemize}

To demonstrate the utility of our monad, we give denotational semantics for a prototype hybrid language.
Rather than tailoring the semantics to an existing language, we develop the language throughout the paper, incrementally adding common hybrid features, and demonstrating in turn how each of these features is modelled by the monad.
This presents the quantum orchestra monad as a flexible tool for giving semantics to a variety of hybrid languages, rather than a model for one specific language.

\paragraph{Remark}

\begin{figure}
  \centering
  \begin{subfigure}[b]{0.45\textwidth}
    \centering
    \begin{quantikz}[row sep=tiny,column sep=small]
      & \setwiretype{n} \wire[r][1]["X"{above,pos=0.5}]{c} & \gate[2][0.7cm]{k} & \setwiretype{n} \midstick[2, brackets=none]{\(\mapsto\)} & \gate[2][0.7cm]{\xi} \wire[r][1]["X"{above,pos=0.5}]{c} & \gate[2][0.7cm]{k} & \setwiretype{n} & \\
      & \setwiretype{n} \wire[r][1]["\mathcal{H}"{above,pos=0.5}]{q} && \wire[l][1]["\mathcal{H}"{above,pos=0.5}]{q}
      \wire[r][1]["\mathcal{H}"{above,pos=0.5}]{q} & \wire[r][1]["\mathcal{H}"{above,pos=0.5}]{q} && \wire[l][1]["\mathcal{H}"{above,pos=0.5}]{q}
    \end{quantikz}
    \vspace*{0.523cm}
    \caption{Quantum orchestras capture precomposition with a quantum instrument.}
    \label{fig:continuation-precomposition}
  \end{subfigure}%
  \hfill
  \begin{subfigure}[b]{0.5\textwidth}
    \centering
    \hspace*{-0.5cm}
	\begin{quantikz}[row sep=tiny,column sep=small]
      & \setwiretype{n} \wire[r][1]["X"{above,pos=0.5}]{c} & \gate[2][0.7cm]{k} & \setwiretype{n} \midstick[2, brackets=none]{\(\mapsto\)} &&[4pt+5pt] \wire[r][1]["X"{above,pos=0.5}]{c} & \gate[2][0.7cm]{k} \gategroup[2,steps=1,style={inner xsep=8.5pt, inner ysep=3.5pt}]{} \gategroup[2,steps=1,style={inner xsep=24pt+5pt, inner ysep=15pt}]{} & \setwiretype{n} &[4pt+5pt] \\
      & \setwiretype{n} \wire[r][1]["\mathcal{H}"{above,pos=0.5}]{q} && \wire[l][1]["\mathcal{H}"{above,pos=0.5}]{q}
      \wire[r][1]["\mathcal{H}"{above,pos=0.5}]{q} & \wire[r][1]["\xi"{above,pos=0.5,font=\normalsize,yshift=2pt},draw=none]{q} & \wire[r][1]["\mathcal{H}"{above,pos=0.5}]{q} && \wire[l][1]["\mathcal{H}"{above,pos=0.5}]{q} & \wire[l][1][dotted, dash phase=1pt]{q} & \wire[l][1]["\mathcal{H}"{above,pos=0.5}]{q}
    \end{quantikz}
    \caption{The orchestra can access all interfaces of \(k\), but channels slide along the \(\mathcal{H}\)-wire on the right.}
    \label{fig:continuation-big-box}
  \end{subfigure}%
  \caption{Two graphical depictions of the action of a quantum orchestra \(\xi\) on a continuation \(k\).}
  \label{fig:continuation}
  \Description[Two graphical depictions of the action of a quantum orchestra on a continuation.]%
  {Figure (a) shows a mapping from a continuation, which is a classically controlled quantum channel, to the sequential composition of a quantum instrument with the continuation where the classical and the quantum output wires of the instrument are connected to the classical and the quantum input wires of the continuation.
  Figure (b) depicts a mapping from a continuation to its composition with the quantum orchestra, where the quantum orchestra is a large box with a hole. The hole contains the continuation, whose classical and quantum inputs are connected to the interfaces of the quantum orchestra.}
\end{figure}
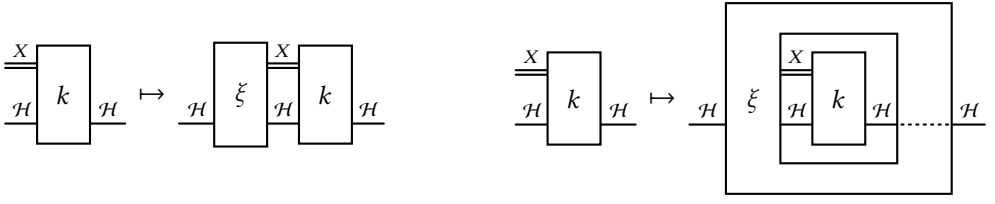

While quantum evolution is frequently thought of as a forward-progression of quantum states in the Schrödinger picture, it is oftentimes algebraically elegant to adopt the Heisenberg perspective in which quantum channels pull measurements backwards.
The design of the quantum orchestra monad heavily draws inspiration from this paradigm and represents computations in \emph{continuation passing style}, pulling continuations backwards through quantum instruments.
For finite classical outcomes, quantum orchestras act on continuations, as expected, through precomposition with a quantum instrument (\cref{fig:continuation-precomposition}).
Clearly, the treatment of unbounded loops however requires the generalisation to infinite classical outcomes.
In this case, the action of quantum orchestras becomes slightly more intricate (\cref{fig:continuation-big-box}), while retaining desirable features such as compositionality with subsequent computations.

\paragraph{Related work}

Existing approaches to the semantics of hybrid quantum computation differ in how classical control and quantum effects are integrated. Measurement has been treated as an algebraic effect over combined quantum-classical programming language~\cite{staton:measurementeffects}, added as a classical effect to a purely quantum programming language~\cite{kaarsgaardheunen:informationeffects}, and monadically without a denotational semantics~\cite{altenkirch:iomonad}. Quantum processes have been given categorical models, such as Selinger's CPM construction \cite{selinger:cpm}, which models mixed-state quantum evolution and measurement via completely positive maps in dagger compact categories, or Huot and Staton's result showing that the category of completely positive trace-preserving maps is a canonical completion of the category of isometries \cite{huotstaton:cptpcategoricalcompletion}. The quantum lambda calculus ~\cite{valiron:qlc} uses linear logic to model hybrid quantum computing. Monads have been used to describe circuit-generation in Quipper~\cite{green:quipper}, and extended to allow mid-circuit measurement~\cite{riosselinger:protoquipperm,sakayori:protoquipperc}. In contrast, we model quantum computation itself as a monadic effect over a purely classical language, using quantum instruments directly as the denotational structure. Closest to our work, Jia \emph{et al.} \cite{jia_semantics_2022} describe an extension of the probabilistic fixpoint calculus by quantum effects, using the theory of \emph{Kegelspitzen} \cite{keimel_mixed_2017} to link the classical part of the language modelled using a valuations monad, and the quantum part modelled with von Neumann algebras.

Similarly to the monad proposed in this paper, there exist other recent works \cite{booth2026composingquantuminstruments,fritz2026quantuminstrumentmonad} that construct monads based on quantum instruments.
While our monad acts on \(\dcpo\) and can thus be seen as a non-commutative quantum generalisation of the probabilistic powerdomain monad \cite{jones:probabilisticpowerdomain}, their monads act on \(\measurable\), the category of measurable spaces, and are therefore closer to generalisations of the classical Giry monad \cite{giry:monad}. The previously proposed monads could hence not directly serve to give semantics to recursive programs.
Interestingly, however, all of the constructions, including ours, coincide in the special case of finite classical outcomes.

\paragraph{Outline of this paper}
In \cref{sec:background}, we recall the necessary background from the literature regarding operator algebra, domain theory, and the study of quantum instruments.
\Cref{sec:toy-language} then introduces our toy language by adding quantum effects successively to a pure classical base language, while introducing quantum instruments as the natural tool to model its semantics.
At the core of our contributions, \cref{sec:continuous-orchestras} eventually presents and discusses the \emph{Quantum Orchestra Monad} on \(\dcpo\) as a way to model diverging quantum programs, with proofs of its well-definedness and of the monad laws given in \cref{sec:orch-monad}. We conclude with \cref{sec:discussion-and-future-work} which clarifies connections to other known constructions, discusses future extensions to the current work, and discusses some open problems.

\section{Background}%
\label{sec:background}

Operator algebras, and in particular von Neumann algebras, are a well-established and standard tool in the mathematical description of quantum information and its physical evolution. While it is largely equivalent to the usual Hilbert space formulation, its power, especially in infinite-dimensional settings, will become apparent later on. The strong order-structure of von Neumann algebras, as witnessed by \cref{thm:w_star_is_dcpo}, will be one of the main pillars of our results.
In the following, we give the most important definitions and known results about \wstar-algebras and their order-theoretic structure for our purposes.

\subsection{\texorpdfstring{\wstar-algebras}{W*-algebras}}

A \emph{\cstar-algebra} is a Banach \(*\)-algebra \(\csaa\) that satisfies the \cstar-identity: \[ \forall a \in \csaa.\quad \| a a^* \| = \| a \| \| a^* \| \]
A \emph{\wstar-algebra} (also known as von Neumann algebra) is a \cstar-algebra \(\csaa\) which admits a topological predual \(\csaa_*\). The weak-\textsuperscript{\(*\)} topology induced on \(\csaa\) by the functionals in \(\csaa_* \subseteq (\csaa_*)^{**} = \csaa^*\) is also called the ultraweak topology.
A \cstar-algebra is called \emph{unital} if it contains an identity; in this work, \emph{all} \cstar-algebras (including \wstar-algebras) are assumed to be unital.
\begin{example}
  Given any Hilbert space \(\mathcal{H}\), the \wstar-algebra \( \mathcal{B(H)} \) is the algebra of bounded linear operators on \(\mathcal{H}\).
  It plays a particularly important role as the algebra generated by all measurement effects on \(\mathcal{H}\).
  Many of the \wstar-algebras used in this work will be of this form.
\end{example}
The set of self-adjoint elements of \(\csaa\), \emph{i.e.}, all \(a \in \csaa\) that satisfy \(a = a^*\), forms a partially ordered vector space where \(a \geq 0\) if and only if \(a = b^* b\) for some \(b \in \csaa\) (this order is referred to as the \emph{L\"owner order}).
The set \(\csaa^+ = \{ a \in \csaa \mid a \geq 0 \}\) of positive elements of \(\csaa\) forms a convex cone.

\begin{definition}
	Let \(\csaa, \csab\) be \wstar-algebras, \(\Phi : \csab \to \csaa\) be a linear map, and \(\Phi^{(n)}\) its \(n\)-th matrix amplification. Then \(\Phi\) is called
	\begin{enumerate}
		\item \emph{normal} (n) if it is ultraweakly continuous;
		\item \emph{positive} (P) if \(\Phi(b) \geq 0\) for all \(b \in \csab^+\);
		\item \emph{completely positive} (CP) if \(\Phi^{(n)}\) is positive for all \(n \in \mathbb{N}\);
		\item \emph{unital} (U) if \(\Phi(1_\csab) = 1_\csaa\);
		\item \emph{subunital} (SU) if \(\Phi(1_\csab) \leq 1_\csaa\).
	\end{enumerate}
\end{definition}
\begin{proposition}[\cite{cho:wstarisdcpo}]
  The exists a symmetric monoidal category \(\mathbf{W^{\bm{*}}_{nCPSU}}\) where the objects are \wstar-algebras, the morphisms are normal completely positive subunital (nCPSU) maps, and the monoidal structure is given by the spatial von Neumann tensor product (see \cite[Definition~1.22.10]{sakai1971c} for a full definition).
  In the following we will also denote this category simply as \(\wstarcat\).
\end{proposition}

Whereas the reader might associate quantum channels in the Schrödinger picture with completely positive trace-preserving (CPTP) or completely positive trace-non-increasing (CPTNI) maps, it will be convenient to switch to the (dual) Heisenberg picture in which trace-non-increasingness is replaced by subunitality. In this sense, one might think of the category \(\wstarcat^{\text{op}}\) as the category of quantum channels.
\begin{example}
	Two particularly important channels are the identity and the zero channel.
	Given a \wstar-algebra \(\csaa\), the \emph{identity channel} on \(\csaa\) is given by the identity map on \(\csaa\), and denoted \( \id_{\csaa} \). It is normal, completely positive, and unital.
	For \wstar-algebras \(\csaa, \csab\), the \emph{zero channel} from \(\csaa\) to \(\csab\) is given by the nCPSU map \( 0_{\csab\csaa} : \csab \to \csaa ,\, b \mapsto 0 \).
\end{example}
\begin{example}[Isometric quantum channels]
	Let \(V : \mathcal{H} \to \mathcal{K}\) be an isometry. Then the channel
	\[
		\overline V : \mathcal{B(K)} \to \mathcal{B(H)}, \quad M \mapsto V^\dagger M V,
	\]
	is normal, completely positive, and unital.
	It is often convenient to obtain quantum channels in this way.
	Common examples include unitary channels (when \(V\) is a unitary), see \cref{sec:toy-cp}, and qubit allocation, as used in \cref{sec:toy-alloc}.
\end{example}

For a detailed introduction to \wstar-algebras beyond these basics, we refer to \cite{sakai1971c}.

\subsection{Directed complete partial orders}

We refer the reader to \cite{cho:wstarisdcpo} for a more detailed discussion of the order structure of \wstar-algebras.

\begin{definition}[Directed complete partial order]
	Let \(X\) be a set and \(\leq\) be a partial order on \(X\).
	\begin{enumerate}
		\item A subset \(Y \subseteq X\) is called \emph{directed} if for all \(x,y \in Y\) there exists \(z \in Y\) such that \(x,y \leq z\).
		\item \(X\) is called \emph{directed complete} if every directed subset has a supremum.
		\item \(X\) is called \emph{pointed} if it has a least element \(\bot\).
	\end{enumerate}
\end{definition}

In the following, directed complete partially ordered sets will be abbreviated as \emph{dcpo}.

\begin{definition}[Scott-continuity]
	Let \(X,Y\) be dcpos and \(f : X \to Y\). Then \(f\) is called \emph{Scott-continuous} if it preserves the suprema of all directed subsets. All Scott-continuous functions are \emph{monotone}, \emph{i.e.}, \(f(x) \leq f(y)\) whenever \(x \leq y\).
\end{definition}

\begin{definition}[Category of dcpos]
	There is a category \(\dcpo\) whose objects are dcpos and whose morphisms are Scott-continuous maps.
\end{definition}

\begin{theorem}[Kleene fixed-point theorem \cite{tarski1955lattice,stoltenberg1994mathematical}]%
\label{thm:kleene-fixed-point}
Let \(f : X \to X\) be a Scott-continuous function on a pointed dcpo \(X\). Then \(f\) has a least fixed point \(\fix f\) given by \[ \fix f = \sup\nolimits_n f^n (\bot). \]
Moreover, the function \(\fix : (X \to X) \to X\) is itself continuous.
\end{theorem}

\Cref{thm:kleene-fixed-point} is an essential and well-established tool to give meaning to diverging programs, and one of the reasons why domain theory became an important field of study for denotational semantics.

For \wstar-algebras \(\csaa, \csab\), we define a partial order \(\leq\) on the maps in \(\wstarcat(\csab,\csaa)\) by letting \(\Phi \geq \Psi \) if and only if \(\Phi - \Psi\) is completely positive.

\begin{theorem}[\cite{cho:wstarisdcpo}]\label{thm:w_star_is_dcpo}
	Let \(\csaa, \csab\) be \wstar-algebras. Then, \(\wstarcat(\csab,\csaa)\) with the partial order \(\leq\) is a pointed dcpo with the zero map as its least element.
	In fact, the category \(\wstarcat\) is \(\dcpo\)-enriched.
\end{theorem}

\subsection{Quantum instruments}
\label{sec:background-inst}

Measurements differ from ordinary quantum evolutions in that they produce both a quantum output and a classical outcome. In the Heisenberg picture, this means that, rather than a single normal completely positive subunital map describing the transformation of observables, one has a family of such maps indexed by the possible outcomes. Each map describes the transformation of observables conditional on a particular outcome, while their sum recovers the unconditional evolution obtained by disregarding the measurement result. This leads to the notion of a quantum instrument.
\begin{definition}[Finite quantum instrument]
	Let \(X\) be a set and \(\csaa, \csab\) be \wstar-algebras. Then, a (finite) quantum instrument from \(\csaa\) to \(\csab\) is a collection of normal completely positive subunital maps \( \{ \Psi_x : \csab \to \csaa \}_{x \in X} \) such that only finitely many maps \( \Psi_x \) are non-zero and such that \( \sum_{x\in X} \Psi_x \) is subunital.
\end{definition}

We interpret the set \(X\) as the set of possible measurement outcomes. In the concrete case where \(\csaa = \mathcal{B}(\mathcal{H})\) and \(\csab = \mathcal{B}(\mathcal{H}')\), each map \(\Psi_x : \mathcal{B}(\mathcal{H}') \to \mathcal{B}(\mathcal{H})\) pulls back observables on the output system \(\mathcal{H}'\) to the input system \(\mathcal{H}\). Given an observable \(A \in \mathcal{B}(\mathcal{H}')\) and quantum state \(\rho \in \mathcal{B}(\mathcal{H})_*\), \(\rho(\Psi_x(A))\) describes the conditional contribution of outcome \(x\) to the unconditional expectation value \(\sum_{x \in X} \rho(\Psi_x(A))\) of \(A\) when the measurement result is discarded. In particular, \(\Psi_x(1)\) is the effect associated to outcome \(x\), since \(\rho(\Psi_x(1))\) is the probability of observing \(x\).

\begin{definition}[Dirac quantum instrument]
  For any \(x \in X\) and \(\Psi \in \wstarcat(\csab, \csaa)\), we define the \emph{Dirac quantum instrument} \( \delta(x, \Psi) \) as
  \[\delta(x, \Psi)_{x'} =
    \begin{cases}
      \Psi, & \text{if } x = x', \\
      0,    & \text{otherwise}.
    \end{cases}
  \]
\end{definition}

We record the following straightforward fact about finite quantum instruments.

\begin{lemma}
	All finite quantum instruments are finite sums of Dirac instruments.
\end{lemma}

Going beyond the finite-outcome case, quantum instruments with continuous classical outcomes have been formally defined and extensively studied in a measure-theoretic context \cite{davies:instruments,booth2026composingquantuminstruments,fritz2026quantuminstrumentmonad}.
In the context of domain theory however, while a form of quantum instrument has been considered in the finite case \cite{cho:wstarisdcpo}, this work is the first to give a description in the case of infinite classical outcomes, to the best of the authors' knowledge.

\subsection{Monads}%
\label{sec:background-monads}

Monads are a central tool for understanding \emph{effectful} computation, functions which have side effects capturing behaviour other than returning a value. They have been used both as a semantic tool and internally to programming languages to embed effects in pure languages. We recall their definition.

\begin{definition}[Monad]
  A monad is a triple \((T, \eta, \mu)\), where \(T\) is an endofunctor on a category \(\mathcal{C}\), the unit \(\eta\) is a natural transformation from the identity functor on \(\mathcal{C}\) to \(T\), and the multiplication \(\mu\) is a natural transformation from \(T^2\) to \(T\) subject to the usual laws.
\end{definition}

Instead of giving a multiplication operation, monads can be defined in terms of \emph{Kleisli extension}, which maps a morphism \(f \in \mathcal{C}(X, T(Y))\) to a function \(\kleisli f : \mathcal{C}(T(X), T(Y))\). Given the extension, the multiplication is given by \(\mu_X = \kleisli{\id_X}\), and given the multiplication, the extension can be recovered by \(\kleisli f(x) = \mu_X(T(f)(x))\). When \(\mathcal{C}\) is monoidal, monads can also be equipped with a strength \(\tau\), a natural transformation such that for \(X, Y \in \mathcal{C}\), \( \tau_{X, Y} : X \otimes T(Y) \to T(X \otimes Y)\). Such monads are called \emph{strong}.

Strong monads capture many common computational effects; throughout the paper we will refer to the writer monad \(W(X) = X \times M\) for some fixed monoid \(M\), the state monad \(S(X) = S \to X \times S\) for a fixed state set \(S\), the continuation passing style (CPS) monad \(C(X) = (X \to \mathsf{Ans}) \to \mathsf{Ans}\) for a fixed set \(\mathsf{Ans}\), and the probabilistic powerdomain monad (see \cite{jones:probabilisticpowerdomain}). For further discussion on monads we refer the reader to \citet{perroneStartingCategoryTheory2024}.

\section{Toy Language}%
\label{sec:toy-language}

Throughout the following section, we will develop a toy language for hybrid quantum-classical computation (introduced in \cref{sec:base-class-lang}), slowly adding various quantum features as effects. In parallel, we describe a denotational semantics for each language, which we slowly adapt to incorporate each feature, motivating the final definition of the quantum orchestra monad.

In \cref{sec:toy-cp}, we add the ability to trigger the execution of unitary quantum gates on a fixed quantum space, creating a circuit description language whose terms can be understood as classical functions which print out a circuit. In this language, qubits are references by a fixed finite set of constants, which have predefined meaning in the quantum space we are manipulating.

In \cref{sec:toy-inst}, we add measurement. As the program can utilise the (classical) outcomes of measurement, it no longer produces static circuits for each input, instead exhibiting non-trivial hybrid interactions. In this section we will show how each program can be modelled by a quantum instrument, crucially describing how these instruments can be composed.

Next, we add the ability to allocate fresh qubits in \cref{sec:toy-alloc}. This allows the size of our quantum state to vary throughout the computation, and requires the addition of non-constant qubit references. We semantically model this by moving to a parameterised monad in the style of \citet{atkeyParameterisedNotionsComputation2009}, where computations are parameterised by their input and output quantum spaces.

Lastly, in \cref{sec:continuous-orchestras} we move our model from sets to \emph{directed-complete partial orders}, allowing the addition of fixpoints combinators to our toy language. In this section we finally give a full definition of our quantum orchestra monad, and the associated operations that can be performed on it.

\subsection{A Pure Classical Base Language}
\label{sec:base-class-lang}

All variants of our toy language are built on the same classical foundation. Our presentation closely follows \citet{matache:concurrency} (following \cite{moggi:monads,levy:fgcbpv}), who use a similar language to study a computational effect. In particular we have types given by the following syntax:
\[ X, Y ~::=~ \unit \mid X \times Y \mid \bool \mid X \rightarrow Y \]
Although we foresee no difficulties with their inclusion, we omit more general sum types, including just a boolean type which is necessary to describe the classical outcomes of measurement.

The language is a form of fine-grain call-by-value lambda calculus~\cite{levy:fgcbpv}, and as such we split its terms into \emph{values}, terms which do not compute and can be bound to variables, and \emph{expressions}, terms representing computations which may trigger effects. Values and expression are respectively given by the following grammars.
\begin{align*}
  v ~&::=~ x \mid () \mid (v_1,v_2) \mid \fst v \mid \snd v \mid \tt \mid \ff \mid \lambda x. e\\
  e ~&::=~ \return v \mid \letin{x}{e_1}e_2 \mid e_1\seq e_2 \mid v_1(v_2) \mid \termif{v}{e_1}{e_2}
\end{align*}
As with the types, we have not included all possible constructions, but just enough to exhibit non-trivial behaviour and interaction with our quantum effects. For example, we do not include a case matching/pattern matching construction for products, instead relying on projections. We do however include a ``let'' binding, even though it can be emulated with \(\lambda\)-abstraction and application, in order to increase the legibility of our example programs. We also include \(e_1\mathop{;}e_2\), the sequential composition of expressions, as a simplification of the ``let'' binding for the case where \(e_2\) does not require the value produced by \(e_1\).

Our typing rules are split into two judgements: \(\Gamma \judgementvalue v : A\) for values and \(\Gamma \judgementproducer e : A\) for expressions. While they are standard, we include them in \cref{fig:toy-typing}, as we build on them throughout the following sections.

\begin{figure}
  \centering
  \begin{mathpar}
    \inferrule{x : X \in \Gamma}{\Gamma \judgementvalue x : X}\and
    \inferrule{ }{\Gamma \judgementvalue () : \unit}\and
    \inferrule{\Gamma \judgementvalue v_1 : X \and \Gamma \judgementvalue v_2 : Y}{\Gamma \judgementvalue (v_1, v_2) : X \times Y}\and
    \inferrule{\Gamma \judgementvalue v : X \times Y}{\Gamma \judgementvalue \fst v : X}\and
    \inferrule{\Gamma \judgementvalue v : X \times Y}{\Gamma \judgementvalue \snd v : Y}\and
    \inferrule{ }{\Gamma \judgementvalue \tt : \bool}\and
    \inferrule{ }{\Gamma \judgementvalue \ff : \bool}\and
    \inferrule{\Gamma, x : X \judgementproducer e : Y}{\Gamma \judgementvalue \lambda x. e : X \to Y}\and
    \inferrule{\Gamma \judgementvalue v : X}{\Gamma \judgementproducer \return v : X}\and
    \inferrule{\Gamma \judgementproducer e_1 : X \and \Gamma, x : X \judgementproducer e_2 : Y}{\Gamma \judgementproducer \letin{x}{e_1}e_2 : Y}\and
    \inferrule{\Gamma \judgementproducer e_1 : X \and \Gamma \judgementproducer e_2 : Y}{\Gamma \judgementproducer e_1 \seq e_2 : Y}\and
    \inferrule{\Gamma \judgementvalue v_1 : X \to Y \and \Gamma \judgementvalue v_2 : X}{\Gamma \judgementproducer v_1(v_2) : Y}\and
    \inferrule{\Gamma \judgementvalue v : \bool \and \Gamma \judgementproducer e_1 : X \and \Gamma \judgementproducer e_2 : X}{\Gamma \judgementproducer \termif{v}{e_1}{e_2} : X}
  \end{mathpar}
  \caption[Toy language typing]{Typing rules for the classical fragment of the toy language}
  \label{fig:toy-typing}
  \Description[Typing rules for the classical fragment of the toy language.]%
  {A total of 13 typing rules for the classical fragment of the toy language.}
\end{figure}

We give semantics to this language in a standard way. First, a semantics \(\sem X\) is given to each type \(X\), which denotes the possible closed values of that type. The semantics of an (open) value \(\Gamma \judgementvalue v : X\) will then simply be a function \(\sem v : \sem \Gamma \to \sem X\), where:
\[ \sem {x_1 : X_1, \dots, x_n : X_n} = \sem {X_1} \times \dots \times \sem {X_n}\]
and we let \(\rho_x \in \sem X\) be the component of \(x\) for \(x : X \in \Gamma\) and \(\rho \in \sem \Gamma\).

We could give semantics to expressions in the same way, but this would not extend to the effectful computations of the following sections. We instead suppose that we have an effect given by some strong monad \(T\) on a cartesian closed category \(\mathcal{C}\), and then the semantics of an expression \(\Gamma \judgementproducer e : X\) will be given by a function:
\[ \sem e : \sem \Gamma \to T \sem X\]
Essentially, an expression of type \(X\), produces a value of type \(X\) and effects governed by \(T\). Of course, one could instantiate \(T\) as the identity monad to obtain a semantics for the pure language above. Equipped with this, we can now give semantics to each type:
\[ \sem \unit = \{()\} \quad \sem {X \times Y} = \sem X \times \sem Y \quad \sem \bool = \{\tt, \ff\}\]
A (closed) value of type \(X \to Y\) then represents a function that takes a \emph{value} of type \(X\), and returns an \emph{expression} of type \(Y\), and so should have semantics:
\[ \sem {X \to Y} = \sem X \to T\sem Y\]
The semantics of values in context \(\Gamma\) is then standard:
\begin{align*}
  \sem x(\semctx) &= \semctx_x & \sem {(v_1, v_2)}(\semctx) &= (\sem {v_1}(\semctx), \sem {v_2}(\semctx))\\
  \sem {()}(\semctx) &= () & \sem {\fst v}(\semctx) &= \pi_1(\sem v(\semctx))\\
  \sem \tt(\semctx) &= \tt & \sem {\snd v}(\semctx) &= \pi_2(\sem v(\semctx))\\
  \sem \ff(\semctx) &= \ff & \sem {\lambda x. e}(\semctx)(\semctx_x) &= \sem e(\semctx, \semctx_x)
\end{align*}
where \(\pi_1\) and \(\pi_2\) are the projections from the product. The semantics of each expression is also standard and can be given in terms of the monad operations (where \(\eta\) is the unit, \(\tau\) is the strength, and extension of \(f\) is given by \(\kleisli{f}\)).
\begin{align*}
  \sem {\return v} &= \sem \Gamma \xrightarrow{\sem v} \sem X \xrightarrow{\eta_X} T \sem X\\
  \sem {\letin {x}{e_1}e_2} &= \sem \Gamma \xrightarrow{\langle \id, \sem {e_1} \rangle} \sem \Gamma \times T \sem X \xrightarrow{\tau_{\sem \Gamma, \sem X}} T(\sem \Gamma \times \sem X) \xrightarrow{\kleisli{\sem {e_2}}} T \sem Y\\
  \sem {e_1 \seq e_2} &=
                        \begin{aligned}[t]
                          \sem \Gamma \xrightarrow{\langle \id, \sem {e_1} \rangle} \sem \Gamma \times T \sem X \xrightarrow{\tau_{\sem \Gamma, \sem X}}{T(\sem \Gamma \times \sem X)} \xrightarrow{T\pi_1}{} &T \sem \Gamma\\
                          \xrightarrow{\kleisli{\sem {e_2}}}{} &T \sem Y
                        \end{aligned}
  \\
  \sem {v_1(v_2)} &= \sem \Gamma \xrightarrow{\langle \sem {v_1}, \sem {v_2} \rangle} (\sem X \to T \sem Y) \times \sem X \to T \sem Y\\
  \sem {\termif{v}{e_1}{e_2}} &= \sem \Gamma \xrightarrow{\langle \sem v, \id \rangle} \sem \bool \times \sem \Gamma \simeq \sem \Gamma + \sem \Gamma \xrightarrow{[\sem {e_1}, \sem {e_2}]} T \sem X
\end{align*}
where \(\langle f, g \rangle\) is the universal map into the product, and \([f, g]\) is the universal map out of the sum.

\subsection{A Classical Language with a Quantum Effect}
\label{sec:toy-cp}
We begin our investigation of quantum effects by augmenting our classical language with a side-effect for performing unitary quantum operations, allowing it to describe (parameterised) circuits. As is common for describing pure quantum computations, we reduce the problem of describing all possible unitaries we could want to describing a fixed \emph{gate set}. To describe which qubits each gate is applied to, we add a type \(\qid\) of \emph{qubit references}:
\[ X, Y ~::=~ \colorboxed[White]{DodgerBlue}{\qid} \mid \unit \mid X \times Y \mid \bool \mid X \rightarrow Y \]
Our quantum gates are then simply functions which take some number of qubit references as input (corresponding to the number of qubits the gate acts on) and return unit. To introduce terms of type \(\qid\), we add a set of constants \(l \in \qram\) to the language, which index a fixed quantum state. These constants have the following typing rules:
\begin{mathpar}
  \inferrule{q \in \qram}{\Gamma \judgementvalue q : \qid}\and
  \inferrule{\textsc{G} \in \{\Had, \SG, \X, \T\} }{\Gamma \judgementvalue \textsc{G} : \qid \to 1}\and
  \inferrule{ }{\Gamma \judgementvalue \CNOT : \qid \times \qid \to 1}\and
\end{mathpar}

We emphasise that gates do not compute without being applied to their qubit arguments, so all appear as values. When they are applied to a sequence of qubit references, these operations are understood to mutate the quantum state referenced by those variables as a side effect. In this presentation we have chosen to use the universal Clifford+T gate set, but this could be substituted by any reasonable alternative, including those parameterised by booleans (classically conditioned gates) or angles (e.g. \textsc{RZ} gates), though doing so would require the addition of angle types to the language, which we avoid for simplicity.

\begin{example}[Swap gate]
  Consider the following circuit implementing the swap gate:
  
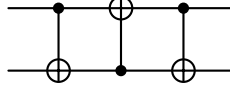
\begin{center}
    \begin{quantikz}
      & \ctrl{1} & \targ{} & \ctrl{1} & \\
       & \targ{} & \ctrl{-1} & \targ{} &
    \end{quantikz}
  \end{center}
  This can be naturally implemented as a function in this toy language:
  \[ \judgementvalue \lambda q. \CNOT((\fst q, \snd q)) \seq \CNOT((\snd q, \fst q)) \seq \CNOT((\fst q, \snd q)) : \qid \times \qid \to \unit\]
  Such a function could be applied in the same way as a primitive gate.
\end{example}

\subsubsection{A Unitary Effect}
\label{sec:unitary-effect}

To extend the semantics of \cref{sec:base-class-lang} to accommodate the new features of this section, we must instantiate the monad \(T\) describing the effect, and give an interpretation of each new type and term constructor. At this stage, all our effect can do is generate some unitary evolution as the program progresses. We will assume that that these unitaries act on the space \(\mathcal{H}\) consisting of exactly the qubits present in \(\qram\):
\[ \mathcal{H} = \bigotimes_{\qram} \mathbb{C}^2 \]
For a choice of qubit \(q \in \qram\), each single qubit quantum operation \(U \in \{\Had, \X, \SG, \T\}\) then induces a unitary \(U_q\) which applies the gate to the given qubit. Pre-empting the following sections, we instead associate to the quantum operation the map:
\[ \overline {U_q}(M) = U_q^\dagger M U_q \]
which is a complete positive unital map from \(\mathcal{B(H)}\) to itself. The interpretation of \(\CNOT\) raises an interesting question. If the two qubit locations passed to the operations are distinct, then we can proceed as before, letting \(\overline{\CNOT_{q_1, q_2}}\) be the appropriate map on \(\mathcal{B(H)}\), but what should we when the qubit references are the same, that is what do we assign \(\overline{\CNOT_{q, q}}\)? One option is to enforce that this can never occur at the type system level (as done in Guppy~\cite{kochGUPPYPythonicQuantumClassical2025}). In this paper, we can instead ``fail'' at runtime by letting:
\[ \overline{\CNOT_{q, q}} = 0_{\mathcal{B(H)}} \]
Such a map is not unital, yet is still subunital, and so we are able to assign a subunital operation to every quantum operation for every choice of arguments.

An expression \(\Gamma \judgementproducer e : X\) produces some value of type \(X\), and also applies some evolution to the quantum space, and so its semantics is the function:
\[ \sem e : \sem \Gamma \to \sem X \times \wstarcat(\mathcal{B(H)}, \mathcal{B(H)})\]
And so by comparing the with the denotation of expressions in \cref{sec:base-class-lang}, we can see that the monad \(T\) should be given by the functor \(\_ \times \wstarcat(\mathcal{B(H)}, \mathcal{B(H)})\), which is the writer/printing monad (see e.g. \cite[Example 5.1.14]{perroneStartingCategoryTheory2024}) with the monoid structure on \(\wstarcat(\mathcal{B(H)}, \mathcal{B(H)})\) being given by composition and the identity map.

At this stage, with the restricted form of our qubit references, we can simply set \(\sem \qid = \qram\), so that closed qubit references is given by one of the constants. Concretely, we have \(\sem q(\semctx) = q\). All that remains is define the semantics of each gate, which is given below:
\[  \sem {G}(\semctx)(q) = ((), \overline G_q) \text { for }G \in \{\Had, \X, \SG, \T\} \quad \sem {\CNOT}(\semctx)((q_1, q_2)) = ((), \overline \CNOT_{q_1, q_2})\]
where each of these is a function \(\sem \Gamma \to \qram^i \to \{()\} \times \wstarcat(\mathcal{B(H)}, \mathcal{B(H)})\) where \(i = 2\) for \(\CNOT\) and \(1\) otherwise, which we derive by expanding the semantics of types and the definition of the monad.

\begin{example}[Classically Controlled Unitaries]
  \label{ex:classically-controlled-x}
  A classically controlled \(\X\) gate can be given as the following function, where \(q \in \qram\):
  \[ \textsc{condX} \coloneq \lambda b. \termif{b}{\X(q)}{\return {()}}\]
  As a closed term, this function has a denotation of the form:
  \[ \sem{\textsc{condX}} : \bool \to \{()\} \times \wstarcat(\mathcal{B(H)}, \mathcal{B(H)})\]
  and a calculation shows that \(\tt\) is sent to \(((), \overline {X_q})\) and \(\ff\) is sent to \(((), \id)\).
\end{example}

\subsection{Instrumenting Effects Between a Classical and Quantum Computer}
\label{sec:toy-inst}

To achieve more interesting hybrid interaction in a quantum programming language, we must be able to perform quantum measurements and be able to interact with their classical output.
To do so, we extend our toy language with an additional function \(\meas\) that takes a \(\qid\) and returns a \(\bool\):

\begin{mathpar}
  \inferrule{ }{\Gamma \judgementvalue \meas : \qid \to \bool}
\end{mathpar}

Here, \(\meas\) implements measurement of a single qubit in the computational basis. Computational basis measurement has two possible outcomes, \( \ket{0} \) and \( \ket{1} \) to which we associate to the classical outcomes \( \{\ff,\tt\} = \bool\). Given that our previous semantics already utilised quantum channels, it may seem like we already have the tools necessary to model this measurement operation. Unfortunately, while the previous semantics can model the action of measurement on our quantum space, it does not model the classical outcome obtained. We recall the reset example from the introduction.

\begin{example}[Reset]
  The reset channel maps all states to \(\ket 0\). As in the introduction, this can be performed by first measuring a qubit, and then applying an \(\X\) gate conditioned on the measurement result. This is realised by the following program, where we let \(q \in \qram\) as in \cref{ex:classically-controlled-x}:
  \[ \textsc{reset} \coloneq \letin{b}{\meas(q)} \textsc{condX}(b) \]
Crucially, the outcome of the measurement is used as an argument to the classical conditioned \(\X\) gate, and so the semantics of this channel cannot be realised by simply composing channels.
\end{example}

To model this action we move from quantum channels to \emph{quantum instruments}, as introduced in \cref{sec:background-inst}. A quantum instrument assigns a quantum channel to each potential classical output, with the each component giving the evolution of the quantum state if we had postselected for that classical output, and the subnormalisation factor giving the probability of receiving that output. More concretely, for the computational basis measurement we have the instrument:
\[ \delta(\ff, \overline{\ketbra{0}}) + \delta(\tt, \overline{\ketbra{1}})\]
where \(\delta\) is the Dirac quantum instrument and we recall that for a linear map \(V\), we define the channel \(\overline V(M) = V^\dagger M V\). We note that the combination of the components is subunital as required:
\[(\overline{\ketbra{0}} + \overline{\ketbra{1}})(1) = \ketbra{0}1\ketbra{0} + \ketbra{1}1\ketbra{1} = \ketbra{0} + \ketbra{1} = 1\]
such an instrument is exactly the tool we need to give semantics to our language with measurement, as it captures both the quantum effect and the classical outcome of the operation. We will therefore define the semantics of this language by letting the monad \(T\) be given by \(Q\), where \(Q(X)\) is the set of quantum instruments on our fixed quantum space \(\mathcal{H}\) with classical outcomes \(X\). We then define:
\[ \sem \meas (\rho)(q) = \delta(\ff, \overline{\ketbra{0}_q}) + \delta(\tt, \overline{\ketbra{1}_q})\]
where as in the previous section, the subscripted \(q\) refers to map which acts as indicated on the given qubit, and acts as the identity on all other qubits. We must also be careful that our quantum operations from the previous section still have an interpretation in the new monad, but this is immediate from the following.

\begin{remark}
  A quantum instrument on the set \(\{()\}\) is precisely a quantum channel: every channel \(\Psi\) gives rise to an instrument \(\delta((), \Psi)\), and every instrument on \(\{()\}\) is of this form.
\end{remark}

\subsubsection{Composing Quantum Instruments}
All that remains to complete our semantics is to show that \(Q\) is actually a monad. Critically, to make sense of programs whose effects are governed by quantum instruments, we must define what it means to compose these instruments. Recalling \cref{fig:quant-inst-compose}, the required composition is the monadic composition depicted in \cref{fig:inst-comp-monad}, where an instrument \(\xi \in Q(X)\) is composed with a Kleisli map \(f : X \to Q(Y)\). We defer a full description of the monad to \cref{sec:quant-instr-are}, partially as it a stepping stone to the full quantum orchestra monad given in \cref{sec:continuous-orchestras}. Here we instead describe the action of each monad operation on Dirac quantum instruments, which is sufficient for building intuition.

The unit of the monad is already represented as a Dirac instrument, with \(\eta_X(x) = \delta(x, \id_{\mathcal{B(H)}})\). Intuitively, this instrument does not effect the quantum state, and could be drawn as:
\[
  \begin{quantikz}[row sep=tiny, column sep=small]
    \setwiretype{n} & \gate[2][1cm]{\delta(x, \id)} & \wire[l][1]["X"{above,pos=0.5}]{c} & \midstick[2, brackets=none]{\(\quad=\quad\)} \setwiretype{n} & &\inputD{x} & \wire[l][1]["X"{above,pos=0.5}]{c} \\
    \setwiretype{n} & \wire[l][1]["\mathcal{H}"{above,pos=0.5}]{q} & \wire[l][1]["\mathcal{H}"{above,pos=0.5}]{q} & & & \wire[l][1]["\mathcal{H}"{above,pos=0.7}]{q} & \wire[l][1]["\mathcal{H}"{above,pos=0.35}]{q}
  \end{quantikz}
\]
Now we consider the monadic composition of a Dirac instrument \(\delta(x, \Psi)\) and a Kleisli map \(f : X \to Q(Y)\), given by \(\kleisli f (\delta(x, \Psi))\). We can again depict this graphically:
\[
  \begin{quantikz}[row sep=tiny, column sep=small]
    \setwiretype{n} &\gate[2][1cm]{\delta(x, \Psi)} & \gate[2][1cm]{f} \wire[l][1]["X"{above,pos=0.5}]{c} & \wire[l][1]["Y"{above,pos=0.5}]{c} & \midstick[2, brackets=none]{\(\quad=\quad\)} \setwiretype{n} & & & \gate[2][1cm]{f(x)} & \wire[l][1]["Y"{above,pos=0.5}]{c} \\
    & \wire[l][1]["\mathcal{H}"{above,pos=0.5}]{q} & \wire[l][1]["\mathcal{H}"{above,pos=0.5}]{q} & \wire[l][1]["\mathcal{H}"{above,pos=0.5}]{q} & \setwiretype{n} &  & \gate{\Psi}\setwiretype{q} \wire[l][1]["\mathcal{H}"{above,pos=0.5}]{q} & \wire[l][1]["\mathcal{H}"{above,pos=0.5}]{q} & \wire[l][1]["\mathcal{H}"{above,pos=0.5}]{q}
  \end{quantikz}
\]
and we see that the \(y\) component of \(\kleisli f (\delta(x, \Psi))\) is given the \(y\) component of \(f(x)\) composed with \(\Psi\) (recalling we work in the Heisenberg picture). The full definition of the monad in \cref{sec:quant-instr-are} implies that this action can be extended by linearity to all finite quantum instruments, despite Dirac decompositions not being unique.

With the definition of the monad, and the semantics of the new constructor, we have a full description of the semantics of this iteration of the toy language. We can now return to the reset example and calculate its semantics. For simplicity, we will assume we work in the single qubit space, that is \(\qram = \{q\}\). We are then in the following situation:
\[
  \begin{quantikz}[row sep=tiny, column sep=small]
    \setwiretype{n} &\gate[2][1cm]{\textsc{reset}} & \wire[l][1]["\unit"{above,pos=0.5}]{c} & \midstick[2, brackets=none]{\(\quad=\quad\)}\setwiretype{n} &  & \gate[2][1cm]{\delta(\ff, \overline{\ketbra{0}})+\delta(\tt, \overline{\ketbra{1}})} & \wire[l][1]["\bool"{above,pos=0}]{c} &\gate[2][1cm]{\textsc{condX}} \setwiretype{c} & \wire[l][1]["\unit"{above,pos=0.5}]{c} \\
    & \wire[l][1]["\mathbb{C}^2"{above,pos=0.5}]{q} & \wire[l][1]["\mathbb{C}^2"{above,pos=0.5}]{q} &\setwiretype{n} & & \wire[l][1]["\mathbb{C}^2"{above,pos=0.5}]{q} & \wire[l][1]["\mathbb{C}^2"{above,pos=0}]{q} & \setwiretype{q} & \wire[l][1]["\mathbb{C}^2"{above,pos=0.5}]{q}
  \end{quantikz}
\]
We can now calculate the component of each side for \(() \in \sem{\unit}\):
\begin{align*}
  \sem {\textsc{reset}}
  &= \kleisli{\sem{\textsc{condX}}}(\sem{\meas})\\
  &= \kleisli{\sem{\textsc{condX}}}(\delta(\ff, \overline{\ketbra{0}}) + \delta(\tt, \overline{\ketbra{1}}))\\
  &= \kleisli{\sem{\textsc{condX}}}(\delta(\ff, \overline{\ketbra{0}})) + \kleisli{\sem{\textsc{condX}}}(\delta(\tt, \overline{\ketbra{1}}))\\
  &= \overline{\ketbra{0}} \circ \sem{\textsc{condX}}(\ff) + \overline{\ketbra{1}} \circ \sem{\textsc{condX}}(\tt)\\
  &= \overline{\ketbra{0}} \circ \id + \overline{\ketbra{1}} \circ \overline{X}\\
  &= \overline{\ketbra{0}} + \overline{\ket{0}\!\!\bra{1}}
\end{align*}
where we have omitted the arguments for the semantics of the empty context and for taking the appropriate component, as both are singleton sets. From this is it clear that our program for reset is semantically identical to the true reset channel in the Heisenberg picture we are working within.

\subsection{Allocating Qubits}
\label{sec:toy-alloc}

So far, our programs have all operated on a fixed quantum space governed by a set of constants \(\qram\). While this aligns closely with hardware, it forces the programmer to calculate how much memory is needed, or dually, to know how much memory is available on their target hardware. Ideally, the programmer should be alleviated of this concern, and be able to request new memory when needed through a new function \(\alloc : 1 \to \qid\).

To accommodate this in our semantics, the quantum space our instruments act on must be continually updated throughout the computation. Further, it should be possible for quantum instruments to have different input and output spaces. We achieve this by generalising our construction to a \emph{parameterised monad} (see \citet{atkeyParameterisedNotionsComputation2009}), \(Q(\csaa, \csab, X)\), where \(X\) is the set acted on by the monad, with \(\csaa\) and \(\csab\) being the input and output \wstar-algebras. The monad operations work similarly to before, with the constraint that the input and output spaces must align when appropriate. We include full details of the parameterised construction for finite instruments in \cref{sec:quant-instr-are}.

Now that we are able to paramaterise our effect, the natural question is where to get the input and output parameters from. To address the second of these questions, we adapt our typing to track the number of qubits allocated by an expression using a grading, reminiscent of graded type-and-effect systems \cite{katsumata_parametric_2014, gaboardi_combining_2016, katsumata_flexible_2022}. The function type is replaced with a graded set of function types:
\[ X \to_n Y \]
and the judgement for expressions is given by \(\Gamma \judgementproducer_n e : X\). The updated typing rules to accommodate these grades are standard, and are given in full in \cref{app:toy}. We highlight the rule for the if statement, and the new rule for \(\alloc\):
\begin{mathpar}
  \inferrule{\Gamma \judgementvalue v : \bool \and \Gamma \judgementproducer_n e_1 : X \and \Gamma \judgementproducer_n e_2 : X}{\Gamma \judgementproducer_n \termif{v}{e_1}{e_2} : X} \and
  \inferrule{ }{\Gamma \judgementvalue \alloc : \unit \to_1 \qid}
\end{mathpar}
The \(\alloc\) function allocates exactly one qubit and returns a references to this new qubit, which is tracked in its type. The rule for the if statement enforces that both expressions must allocate the same number of qubits, ensuring that the total allocation is static. This allows us to determine how much larger the output parameter should be than the input.

Determining the input space is complicated by the fact that an expression could be run in a variety of quantum spaces, sometimes within the same program. To remedy this, we define the semantics of a term in all possible quantum spaces, much like how the semantics of a term is given in semantic contexts. The semantics of an expression is then given by a function:
\[ \sem{e}_m \in \dcpo(\sem{\Gamma}_m, \orch(\mathcal{B}(\mathcal{H}^m), \mathcal{B}(\mathcal{H}^{m+n}), \sem{X}_{m + n}))\]
and is defined for every number of starting qubits \(m\), with the final space having \(m + n\) qubits. We have chosen here to also parametrise the semantics of types over natural numbers, allowing us to define:
\[ \sem \qid_m = \{0, \dots, m - 1\}\]
ensuring that qubit references are always ``in scope''. To define the semantics of ``let'', we now require an inclusion function \(\iota_{\Gamma, m, m'} : \sem \Gamma_m \to \sem \Gamma_{m'}\) for \(m' \geq m\). Doing is not entirely trivial, requiring the semantics of types to be given as:
\[ \sem{X \to_n Y}_m = \forall m' \geq m. \dcpo(\sem{X}_{m'}, \orch(\mathcal{B}(\mathcal{H}^{m'}), \mathcal{B}(\mathcal{H}^{m' + n}), \sem{Y}_{m' + n})) \]
ensuring that \(\sem {X \to_n Y}_m \subseteq \sem {X \to_n Y}_{m'}\) for \(m' \geq m\). Finally, we have:
\[ \sem{\alloc}_m(\semctx)(u) = \delta(m, \overline V) \quad \text{where } V(v) = v \otimes \ket{0}\]
recalling that \(\overline V (M) = V^\dagger M V\). We emphasise that the action of \(\alloc\) depends on the size of the input quantum space, both for the quantum channel and the returned qubit reference.
\begin{example}[Bell pair allocation]
  The following expression \(\judgementproducer_2 \textsc{GHZ} : \qid \times \qid\) allocates a (two-qubit) Bell pair:
  \begin{equation*}
    \textsc{Bell} \coloneq \letin{q_1}{\alloc(())}\{ \Had(q_1)\seq \letin{q_2}{\alloc(())}\{ \CNOT(q_1, q_2) \seq \return (q_1, q_2) \} \}
  \end{equation*}
  The returned qubits references depend on the size of the quantum space this program is executed in. In general we have:
  \[\sem {\textsc{Bell}}_m = \delta((m, m+1), \overline\phi) \in Q(\mathcal{B}(\mathcal{H}^m), \mathcal{B}(\mathcal{H}^{m+2}), \{0,\dots, m+1\}^2)\]
  where \(\phi(v) = v \otimes \frac{1}{\sqrt{2}} (\ket{00} + \ket{11})\). We note that the inclusion \(\iota_{\qid, m, m+1}\) is used to include the returned qubit reference from the first \(\alloc\) into the set \(\{0, \dots, m+1\}\).
\end{example}

This restricted form of static allocation through graded types is of course not the only form that allocation could take, though we choose the form as we believe it best showcases the parameterised nature of the monad. An alternative to our \(\alloc\) function could be a construction:
\[\inferrule{\Gamma, a: \qid \judgementproducer e : X}{\Gamma \judgementproducer \mathop{\mathsf{ancilla}} a \mathop{\mathsf{in}} e : X}\]
which allocates an ancilla qubit, executes the expression \(e\) in the larger space, and releases the qubit afterwards. Such behaviour could also be captured by our parameterised monad (though this would only ever need invocations of the monad where the input and output space are the same).

Finally, a more flexible allocation method could be modelled by fixing a single infinite dimensional \wstar-algebra instead of parametrising, utilising that the monad is defined for all \wstar-algebras, and not just the finite dimensional ones we have used so far. We believe that with the appropriate algebra, the monad could support the allocation operation above without the need to enforce that it is only used statically, allowing it to appear in conditional statements, and even loops (see \cref{sec:continuous-orchestras}).

\section{Modelling Diverging Quantum Programs}\label{sec:continuous-orchestras}
To model programs with both quantum effects and divergence, we extend our finite instrument monad of the previous sections to a ``Quantum Orchestra'' monad \(\orch\) on the category \(\dcpo\). We motivate this by investigating a \emph{repeat-until-success} circuit~\cite{adamRepeatUntilSuccessNondeterministicDecomposition2014}.

\begin{example}[Repeat-Until-Success]
  \label{ex:rus}
  Let \(U\) be some unitary we want to perform on \(\mathcal{H}\), and suppose we have a (closed) computation \(\judgementproducer C : \bool\) with a single classical boolean output such that:
  \begin{itemize}
  \item if the classical output is \(\ff\), then the quantum state of \( \mathcal{H} \) remains unchanged,
  \item if it is \(\tt\), then \(U\) is performed on the state.
  \end{itemize}
  A repeat-until-success circuit performs the evolution \(U\) by repeatedly running the computation \(C\) until a \(\tt\) outcome is observed. We can model this in our setting as follows: the semantics of \(C\) can be given by the following instrument, where \(p\) is a fixed probability of observing \(\tt\):
  \begin{align*}
    &\sem C : Q(\mathcal{B(H)}, \mathcal{B(H)}, \mathbf{2})\\
    &\sem C = \delta(\tt, p\overline U) + \delta(\ff, (1-p)\id_{\mathcal{B(H)}})
  \end{align*}
  We then define the repeat until success circuit recursively by:
  \[ \textsc{RUS} \coloneq \letin{b}{C} \termif{b}{\return {()}}{\textsc{RUS}}\]
  A standard method to solve this recursive equation is through a fixpoint:
  \[ \textsc{RUS} \coloneq \fix \lambda r.\letin{b}{C}\termif{b}{\return {()}}{r(())}\]
  where \(r : \unit \to \unit\). Using our instrument monad, we could calculate (up to the isomorphism):
  \[\sem{\lambda r.\letin{b}{C} \termif{b}{\return {()}}{r(())}} : Q(\mathcal{B(H)}, \mathcal{B(H)}, \unit) \to Q(\mathcal{B(H)}, \mathcal{B(H)}, \unit) \]
  but could not in general take a fixpoint of this function. By extending the monad to \(\dcpo\), assuming the function above were monotone we would be able to define:
  \[ \sem{\textsc{RUS}} = \fix {\sem{\lambda r.\letin{b}{C}\termif{b}{\return {()}}{r(())}}}\]
  where it could then be proved that this program has the desired properties.
\end{example}

An immediate consequence of modelling divergence is that the finiteness condition on our instruments is too restrictive: if we adapted the computation of \cref{ex:rus} to return a counter of how many iterations occurred, then the resulting instrument would clearly have non-zero components for each natural number. This complicates taking the summations used to define the monad multiplication.
A second observation is that an appropriate ordering on instruments cannot be given pointwise on components. Despite this ordering making the set of quantum instruments a dcpo, it is too coarse: the unit of the monad is not continuous in this ordering as \(x < y\) would imply that \(\eta(x)\) is incomparible to \(\eta(y)\).

A potential solution to both problems, inspired by the probabilistic powerdomain monad~\cite{jones:probabilisticpowerdomain}, is for instruments to have a component for each Scott-open subset (those that are upwards closed and inaccessible from below), as opposed to each element, subject to the conditions of being a \emph{valuation}. Under this interpretation, the multiplication must then be defined by an integral construction, and it must be shown that this integral construction is monotone. The techniques used in the classical (probabilistic) case do not neatly transfer to our quantum case, and we will compare these approaches in \cref{sec:discussion-and-future-work}.

Instead of using valuations, we instead extend further to a continuation monad of ``Quantum orchestras'', where instead of having a component for each open set of a dcpo \(X\), we have a component for each continuous function \(X \to \wstarcat(\csab, \csaa)\). Orchestras restrict to valuations, by considering the components associated to indicator functions. Intuitively, the data of an orchestra is that of a valuation and a definition of integral for that valuation. Crucially, this allows the order to be defined in a way that makes this bundled integration function continuous by definition.

We again define this monad as a \(\wstarcat^{\text{op}}\)-parameterised monad. We note that the resulting form of the monad is very close to the continuation passing style monad, from which we inherit much of the monad structure.

\begin{definition}[Parameterised Quantum Orchestra Monad]
  \label{def:quantum-orchestra}
  We define the \( \wstarcat^{\text{op}} \)-parameterised quantum orchestra monad \( \orch \) as follows:
  \begin{itemize}
  \item Let \( \csaa \), \( \csab \) be \wstar-algebras and \( X \) be a \( \dcpo \). Then elements \( \xi \in \orch(\csaa, \csab, X) \) are given by a collection:
    \[\{\xi_k \in \wstarcat(\csak, \csaa) \mid \csak \in \wstarcat, k \in \dcpo(X, \wstarcat(\csak, \csab)) \}\]
    satisfying:
    \begin{itemize}
    \item \textbf{Continuity:} For each \(\csak\) and \(\csae\), the function
      \[\lambda k. \xi_k : \dcpo(X, \wstarcat(\csak, \csab)) \to \wstarcat(\csak, \csaa)\]
      is Scott-continuous.
    \item \textbf{Subconvexity:} \(\xi\) has a restricted form of linearity. Given \(\theta, \rho \in \mathbb{R}^+\) with \(\theta + \rho \leq 1\) we have for all \( k, k' \in \mathbf{DCPO}(X,\wstarcat(\csak, \csab)) \):
      \[ \xi_{\lambda x. \theta k(x) + \rho k'(x)} = \theta \xi_k + \rho \xi_{k'} \]
    \item \textbf{Compositionality:} If \(\chi \in \wstarcat(\csak', \csak)\), then:
      \[ \xi_{\lambda x.k(x) \circ \chi} = \xi_k \circ \chi\]
      for all \(k \in \dcpo(X, \wstarcat(\csak, \csab))\).
    \end{itemize}
    We define the ordering pointwise: \(\xi \leq \xi'\) if \(\xi_k \leq \xi'_k\) for all \(k\), where the ordering on \(\wstarcat(\csak, \csab)\) is given by the standard L\"owner order. For a directed set \(\Delta \subseteq \orch(\csaa, \csab, X)\), we define its supremum by:
    \[ (\sup \Delta)_k = \sup \{ \xi_k \mid \xi \in \Delta \} \]
    for all \(\csak\), and \(k \in \dcpo(X, \wstarcat(\csak, \csab))\), using that \(\wstarcat\) is \(\dcpo\)-enriched~(see \cref{thm:w_star_is_dcpo}).
  \item The action of the functor \(\orch\) on \(\phi \in \wstarcat(\csaa, \csaa')\), \(\psi \in \wstarcat(\csab', \csab)\), and \(f \in \dcpo(X, Y)\) is given by:
    \[ \orch(\phi, \psi, f)(\xi)_k = \phi \circ \xi_{\lambda x.\psi \circ k(f(x))}\]
    where \(k \in \dcpo(Y, \wstarcat(\csak, \csab'))\) and \(\xi \in \orch(\csaa, \csab, X)\).
  \item The unit \(\eta_{\csaa X} : X \to \orch(\csaa, \csaa, X)\) is given by:
    \[\eta_{\csaa X}(x)_k = k(x)\]
    for \(x \in X\) and \(k \in \dcpo(X, \wstarcat(\csak, \csaa))\).
  \item The multiplication \(\mu_{\csaa, \csab, \csac, X} : \orch(\csaa, \csab, \orch(\csab, \csac, X)) \to \orch(\csaa, \csac, X)\) is given by:
    \[\mu_{\csaa, \csab, \csac, X}(\Xi)_k = \Xi_{\lambda \xi.\xi_k}\]
    for \(k \in \dcpo(X, \wstarcat(\csak, \csac \otimes \csae))\).
  \item The strength \(\tau_{X, \csaa, \csab, Y} : X \times \orch(\csaa, \csab, Y) \to \orch(\csaa, \csab, X \times Y)\) is defined by:
    \[\tau_{X, \csaa, \csab, Y}(x, \xi)_k = \xi_{\lambda y. k(x, y)}\]
    for \(k \in \dcpo(X \times Y, \wstarcat(\csak, \csab))\).
  \end{itemize}
\end{definition}

In \cref{sec:orch-monad}, we will prove that the definition above actually forms a monad, as we have claimed. We next define the \emph{Dirac} quantum orchestra, the analogue of the Dirac quantum instrument. As each of the quantum constructors in \cref{sec:toy-inst,sec:toy-alloc} were given in terms of Dirac instruments, this construction allows us to immediately transfer their definition to the orchestra monad.

\begin{definition}[Dirac quantum orchestra]
  Given \(x \in X\) and \(\Phi \in \wstarcat(\csab, \csaa)\), we can define the \emph{Dirac quantum orchestra} \(\delta(x, \Phi) \in \orch(\csaa, \csab, X)\) by:
  \[\delta(x,\Phi)_k = \Phi \circ k(x)\]
  for all \(k \in \dcpo(X, \wstarcat(\csak, \csab))\).
  To show this is indeed a quantum orchestra, we first note that continuity and subconvexity follow from suprema and subconvex combinations being defined pointwise on the space \(\dcpo(X, \wstarcat(\csak, \csab))\), and from the composition of morphisms in \(\wstarcat\) being continuous and linear. For compositionality, we have:
  \[ \delta(x, \Phi)_{\lambda x.k(x) \circ \chi} = \Phi \circ k(x) \circ \chi = \delta(x, \Phi)_k \circ \chi\]
  for \(\chi \in \wstarcat(\csak', \csak)\) and \(k \in \dcpo(X, \wstarcat(\csak, \csaa))\), as required.
\end{definition}

We pause to consider the action of each of the monad operations on Dirac quantum orchestras. Although it does not take an orchestra as an argument, it is immediate that the unit \(\eta_{\csaa X}(x)\) is given by the Dirac quantum orchestra \(\delta(x, \id_\csaa)\). Given \(x \in X\), \(\Phi \in \wstarcat(\csab, \csaa)\), \(\phi \in \wstarcat(\csaa, \csaa')\), \(\psi \in \wstarcat(\csab', \csab)\), and \(f \in \dcpo(X, Y)\), we have:
\begin{equation*}
  \orch(\phi, \psi, f)(\delta(x, \Phi))_k = \phi \circ \Phi \circ \psi \circ k(f(x)) = \delta(f(x), \phi \circ \Phi \circ \psi)_k
\end{equation*}
for all \(k \in \dcpo(Y, \wstarcat(\csak, \csab'))\). For the multiplication, consider \(\xi \in \orch(\csab, \csac, X)\) and \(\Phi \in \wstarcat(\csab, \csaa)\). Then for \(k \in \dcpo(X, \wstarcat(\csak, \csac))\) we have:
\begin{equation*}
  \mu_{\csaa, \csab, \csac, X}(\delta(\xi, \Phi))_k = \delta(\xi, \Phi)_{\lambda \xi'. \xi'_k} = \Phi \circ \xi_k = \orch(\Phi, \csac, X)(\xi)_k
\end{equation*}
We consider applying the strength to \(\delta(y, \Phi)\) where \(y \in Y\) and \(\Phi \in \wstarcat(\csab, \csaa)\). For \(x \in X\) and \(k \in \dcpo(X \times Y, \wstarcat(\csak, \csab))\):
\[ \tau_{X, \csaa, \csab, Y}(x, \delta(y, \chi))_k = \delta(y, \chi)_{\lambda y', k(x, y')} = (\chi \otimes \csae) \circ k(x, y) = \delta((x, y), \chi)_k \]
Lastly, we consider the extension operation for this monad. We gave the monad multiplication above to align with the known monad laws for parameterised monads, but can recover the extension of \(f: X \to \orch(\csab, \csac, Y)\) as:
\[ \kleisli f : \orch(\csaa, \csab, X) \to \orch(\csaa, \csac, Y) \qquad \kleisli f(\xi) = \mu_{\csaa, \csab, \csac, Y}(\orch(\csaa, \csab, f)(\xi)) \]
applying this to a Dirac quantum orchestra \(\delta(x, \Psi) \in \orch(\csaa, \csab, X)\), we get:
\[ \kleisli f(\delta(x, \Psi)) = \mu_{\csaa, \csab, \csac, Y}(\orch(\csaa, \csab, f)(\delta(x, \Psi))) = \mu_{\csaa, \csab, \csac, Y}(\delta(f(x), \Psi)) = \orch(\Psi, \csac, Y)(f(x))\]
for a continuation \(k \in \dcpo(Y, \wstar(\csak, \csac))\), we further calculate:
\[ \kleisli f(\delta(x, \Psi))_k = \orch(\Psi, \csac, Y)(f(x))_k = \Psi \circ f(x)_k\]
We note the action of this monad on Dirac orchestras is entirely analogous to the action of the finite quantum instrument monad on Dirac instruments, and therefore the examples and reasoning from \cref{sec:toy-inst} embed into this setting. Much of this reasoning extends linearly to \emph{finite quantum orchestras}, finite sums of Dirac orchestras, though characterising the order on finite orchestras is not trivial. We discuss finite orchestras in greater detail in \cref{sec:finite-quant-orch}.

We can finally now define the final version of our toy language. The only addition over the language of \cref{sec:toy-alloc} is the fix point operator, which has the following typing rule:
\begin{mathpar}
  \inferrule{\Gamma \judgementvalue f : (\unit \to_0 X) \to_0 X}{\Gamma \judgementproducer_0 \fix f : X}
\end{mathpar}
We may have expected to apply a fixed point to an endofunction \(f: X \to X\), but we note that \(\sem{X \to X} = \sem X \to \orch(\sem X)\) which is not suitable for taking a fixed point. If we had tried to remedy this by applying the extension to this function, the resulting endofunction would be strict, and hence its fixpoint would always be zero. In contrast:
\[ \sem{(\unit \to X) \to X} = (\unit \to \orch(X)) \to \orch(X) \simeq \orch(X) \to \orch(X)\]
and so is suitable for taking a fixpoint. We note that this form of the fixpoint operator is reminiscent of the construction in \citet[Chapter 4]{levyCallByPushValue2003}, but with thunking replaced by functions from the unit.

For this language, we enforce at the type level that the function \(f\) performs no allocations, so that we can statically know the output space of the computation. As the monad can operate on infinite \wstar-algebras, we believe this is not a fundamental limitation, but we leave allocating within a loop to future work. Recalling \cref{thm:kleene-fixed-point}, it is now routine to give semantics to this new constructor, which should take the form of a (continuous) function:
\[ \sem{\fix f}_m \in \dcpo(\sem{\Gamma}, \orch(\mathcal{B}(\mathcal{H}^m), \mathcal{B}(\mathcal{H}^m), \sem{X}_m))\]
and is given by the following composition:
\begin{align*}
    \sem{\fix f}_m = \sem{\Gamma}_m &\xrightarrow{\sem{f}_m} \dcpo((\unit \to \orch(\mathcal{B}(\mathcal{H}^m), \mathcal{B}(\mathcal{H}^m), \sem{X}_m)), \orch(\mathcal{B}(\mathcal{H}^m), \mathcal{B}(\mathcal{H}^m), \sem{X}_m))\\
                    &\simeq \dcpo(\orch(\mathcal{B}(\mathcal{H}^m), \mathcal{B}(\mathcal{H}^m), \sem{X}_m), \orch(\mathcal{B}(\mathcal{H}^m), \mathcal{B}(\mathcal{H}^m), \sem{X}_m))\\
  & \xrightarrow{\fix} \orch(\mathcal{B}(\mathcal{H}^m), \mathcal{B}(\mathcal{H}^m), \sem{X}_m)
\end{align*}
Crucially, the set \(\orch(\csaa, \csab, X)\) is a pointed dcpo for any \(\csaa\), \(\csab\), and \(X\), with the bottom element being given by the the orchestra which maps every continuation to the zero channel, allowing us to take a fixpoint. The semantics of all other terms can be defined analogously to before, but for completeness we include a full description in \cref{app:toy}.

We end this section by recalling the repeat-until-success program \(\textsc{RUS}\), and computing its meaning. Let us assume that \(\Gamma \judgementproducer_0 C : \bool\) and for a given \(\semctx \in \sem{\Gamma}_m\) we have:
\[ \sem{C}_m(\rho) = \delta(\tt, p\overline{U}) + \delta(\ff, (1 - p)\id_{\mathcal{B}(\mathcal{H})^m})\]
for some unitary \(U\) on \(\mathcal{H}^m\) and \(0 < p < 1\). Then, letting:
\[ f \coloneq \lambda r. \letin b C \termif{b}{\return ()}{r(())} \qquad \textsc{RUS} \coloneq \fix f \]
we have (again up to the iso \((\unit \to X) \simeq X\)) by calculation that:
\[ \sem f_m(\rho)(\xi) = \delta((), p\overline{U}) + (1-p)\xi \]
Now, it is clear that \((\sem{f}_m(\rho))^n(\bot) = \delta((), (1 - (1-p)^n)\overline{U})\), and so:
\[ \fix (\sem f_m (\rho)) = \sup\nolimits_n (\sem f_m(\rho))^n(\bot) = \sup\nolimits_n \delta((), (1-(1-p)^n)\overline{U}) = \delta((), \overline{U})\]
as desired.

\section{Quantum orchestras form a strong parameterised monad}%
\label{sec:orch-monad}

We prove that the quantum orchestra construction given in
\cref{def:quantum-orchestra} is well defined, leading to
\cref{thm:strong-monad} which states that \(\orch\) is a strong
\(\wstarcat^{\text{op}}\)-parameterised monad on \(\dcpo\) with unit
\(\eta\), multiplication \(\mu\), and strength \(\tau\). We begin by
showing that the action on objects produces dcpos.

\begin{lemma}
  \label{lem:orch-dcpo}
  \(\orch(\csaa, \csab, X)\) is a dcpo.
\end{lemma}
\begin{proof}
  As the ordering on \(\orch(\csaa, \csab, X)\) is pointwise on components, it is clear that it is a partial order. To show the suprema are well defined, fix a directed set \(\Delta \subseteq \orch(\csaa, \csab, X)\). We first show \(\sup \Delta \in \orch(\csaa, \csab, X)\):
  \begin{itemize}
  \item \textbf{Continuity}: Fix \(\csak\). Then the function:
    \[ \lambda k. (\sup \Delta)_k : \dcpo(X, \wstarcat(\csak, \csab)) \to \wstarcat(\csak, \csaa)\]
    is a supremum of functions \(\{\lambda k. \xi_k \mid \xi \in \Delta\}\), each of which is continuous by the continuity of \(\xi\). Hence \(\lambda k. (\sup \Delta)_k\) is a supremum of continuous functions, and so is continuous.
  \item \textbf{Subconvexity}: Let \(\theta, \rho \in \mathbb{R}^+\) with \(\theta + \rho \leq 1\), and suppose \(k, k' \in \dcpo(X, \wstarcat(\csak, \csab))\). Then:
    \begin{align*}
      (\sup \Delta)_{\lambda x.\theta k(x) + \rho k'(x)}
      &= \sup \{ \xi_{\lambda x. \theta k(x) + \rho k'(x)} \mid \xi \in \Delta\}\\
      &= \sup \{ \theta \xi_k + \rho \xi_{k'} \mid \xi \in \Delta\}\\
      &= \theta \sup\{ \xi_k \mid \xi \in \Delta\} + \rho \sup\{\xi_{k'} \mid \xi \in \Delta\}\\
      &= \theta (\sup \Delta)_k + \rho (\sup \Delta)_{k'}
    \end{align*}
    Where scalar multiplication by a positive number is continuous because scalar multiplication by a positive number is a poset isomorphism. Similarly, adding a fixed channel is a poset isomorphism, and so addition is continuous as it is separately continuous in both arguments.
  \item \textbf{Compositionality}: Suppose \(\chi \in \wstarcat(\csak', \csak)\). Then:
    \begin{align*}
      (\sup \Delta)_{\lambda x.k(x) \circ \chi}
      &= \sup\{\xi_{\lambda x.k(x)\circ \chi} \mid \xi \in \Delta\}\\
      &= \sup\{\xi_k \circ \chi \mid \xi \in \Delta\}\\
      &= \sup\{\xi_k \mid \xi \in \Delta\} \circ \chi &\text{by continuity of }\circ\\
      &= (\sup\Delta)_k \circ \chi
    \end{align*}
    for all \(k \in \dcpo(X, \wstarcat(\csak, \csab))\), using that composition is Scott-continuous \cite[Proposition~C.13]{cho:wstarisdcpo}.
  \end{itemize}
  Consequently, \(\sup \Delta \in \orch(\csaa, \csab, X)\).

  It remains to show that \(\sup \Delta\) actually is the supremum of \(\Delta\), as our notation has claimed. For any \(\xi' \in \Delta\), we have \(\xi'_k \leq \sup \{ \xi_k \mid \xi \in \Delta\} = (\sup \Delta)_k\), and so \(\xi' \leq \sup \Delta\). Instead, assume \(\xi'\) is an upper bound for \(\Delta\). Then for all \(k\) and \(\xi \in \Delta\), \(\xi'_k \geq \xi_k\) and so \(\xi'_k \geq \sup \{\xi_k \mid \xi \in \Delta\} = (\sup \Delta)_k\) and so \(\xi' \geq \sup \Delta\). Hence \(\sup \Delta\) is the supremum of \(\Delta\) and so \(\orch(\csaa, \csab, X) \in \dcpo\).
\end{proof}

\subsection{\texorpdfstring{\(\bm{\orch}\)}{Q} is a functor}
\label{sec:orch-funct}

The next two lemmas prove that \(\orch\) is a functor \(\wstarcat \times \wstarcat^{\text{op}} \times \dcpo \to \dcpo\). Before proving the equations for functoriality, we must first show that the action on morphisms produces a well-defined function, and that this function is continuous (and hence forms a dcpo morphism).

\begin{lemma}
  For \(\phi \in \wstarcat(\csaa, \csaa')\), \(\psi \in \wstarcat(\csab', \csab)\), and \(f \in \dcpo(X, Y)\), the function \(\orch(\phi, \psi, f)\) is well-defined and continuous.
\end{lemma}
\begin{proof}
  To show this function is well-defined, we need that \(\orch(\phi, \psi, f)(\xi) \in \orch(\csaa', \csab', Y)\) for \(\xi \in \orch(\csaa, \csab, X)\). We check each condition in turn:
  \begin{itemize}
  \item \textbf{Continuity}: Let \(K \subseteq \dcpo(Y, \wstarcat(\csak, \csab'))\) be directed. Then:
    \begin{align*}
      \orch(\phi, \psi, f)(\xi)_{\sup K}
      &= \phi \circ \xi_{\lambda x. \psi \circ \sup K(f(x))}\\
      &= \phi \circ \xi_{\lambda x. \sup \{\psi \circ k(f(x)) \mid k \in K\}}&\text{by continuity of }\circ\\
      &= \phi \circ \sup\{\xi_{\lambda x. \psi \circ k(f(x))} \mid k \in K\}&\text{by continuity of }\xi\\
      &= \sup\{\phi \circ \xi_{\lambda x. \psi \circ k(f(x))} \mid k \in K\}&\text{by continuity of }\circ
    \end{align*}
    for all \(k\), as required.
  \item \textbf{Subconvexity}: Let \(\theta + \rho \leq 1\) and \(k, k' \in \dcpo(X, \wstarcat(\csak, \csab'))\). Then we need:
    \begin{align*}
      \phi \circ \xi_{\lambda x. \psi \circ (\theta k(f(x)) + \rho k'(f(x)))}
      = &\phantom{{}+{}} \theta( \phi \circ \xi_{\lambda x. \psi \circ k (f(x))})\\
        &+ \rho (\phi \circ \xi_{\lambda x. \psi \circ k'(f(x))})
    \end{align*}
    but this follows from the linearity of \(\circ\) and subconvexity of \(\xi\).
  \item \textbf{Compositionality}: Let \(\chi \in \wstarcat(\csak', \csak)\). Then we need that:
    \[\phi \circ \xi_{\lambda x. \psi \circ k(f(x)) \circ \chi} = \phi \circ \xi_{\lambda x. \psi \circ k(f(x))} \circ \chi\]
    but this follows immediately from the compositionality of \(\xi\).
  \end{itemize}
  To show continuity, let \(\Delta\) be a directed set of orchestras and \(k \in \dcpo(Y, \wstarcat(\csak, \csaa'))\). Then:
  \begin{align*}
    \orch(\phi, \psi, f)(\sup \Delta)_k
    &= \phi \circ \sup \{ \xi_{\lambda x. \psi \circ k(f(x))} \mid \xi \in \Delta\}\\
    &= \sup \{ \phi \circ \xi_{\lambda x. \psi \circ k(f(x))} \mid \xi \in \Delta\} &\text{by continuity of }\circ\\
    &= {(\sup \{ \orch(\phi, \psi, f)(\xi) \mid \xi \in \Delta \})}_k
  \end{align*}
  and so \(\orch(\phi, \psi, f)\) is continuous.
\end{proof}

\begin{lemma}
  \(\orch\) is a functor from \(\wstarcat \times \wstarcat^{\text{op}} \times \dcpo \to \dcpo\).
\end{lemma}
\begin{proof}
  We first check that \(\orch\) sends the identity to the identity:
  \[ \orch(\id_\csaa, \id_\csab, \id_X)(\xi)_k = \id_\csaa \circ \xi_{\lambda x. \id_\csab \circ k(\id_X(x))} = \xi_k\]
  Suppose \(\phi \in \wstarcat(\csaa, \csab)\), \(\phi' \in \wstarcat(\csab, \csac)\), \(\psi \in \wstarcat(\csab', \csaa')\), \(\psi' \in \wstarcat(\csac', \csab')\), \(f \in \dcpo(X, Y)\), and \(g \in \dcpo(Y, Z)\). Then:
  \begin{equation*}
    \orch(\phi', \psi', g)(\orch(\phi, \psi, f)(\xi))_k
    = \phi' \circ \phi \circ \xi_{\lambda x. \psi \circ \psi' \circ k(g(f(x)))}
    = \orch(\phi' \circ \phi, \psi \circ \psi', g \circ f)(\xi)_k
  \end{equation*}
  for all \(\xi\) and \(k\), and so the functor is compatible with composition.
\end{proof}

\subsection{The unit is a natural transformation}
\label{sec:unit-natur-transf}

We next show that the unit \(\eta\) is a natural transformation, whose components are morphisms \( X \to \orch(\csaa, \csaa, X) \) in \(\dcpo\)

\begin{lemma}
  The unit \(\eta\) is well-defined, that is for \(\csaa \in \wstarcat\) and dcpo \(X\) that:
  \[\eta_{\csaa X}(x) \in \orch(\csaa,\csaa, X) \]
  for \(x \in X\), and that \(\eta_{\csaa X}\) is continuous.
\end{lemma}
\begin{proof}
  As \(\eta_{\csaa X}(x)\) is a Dirac quantum orchestra for each \(x\), \(\eta_{\csaa X}\) is a  well-defined function from \(X\) to \(\orch(\csaa, \csaa, X)\), so it remains to show it is continuous. Let \(D \subseteq X\) be directed, then:
  \[ \eta_{\csaa X}(\sup D)_k = k(\sup D) = \sup \{k(x) \mid x \in D \} = \sup \{\eta_{\csaa X}(x) \mid x \in D\}\]
  as required, where the second equality is from the continuity of \(k\).
\end{proof}

\begin{lemma}
  The unit \(\eta\) is a natural transformation.
\end{lemma}
\begin{proof}
  To show naturality of \(\eta\), we need naturality in \(X\) and dinaturality in \(\csaa\). For the first we require the following to commute:
    \[
    \begin{tikzcd}[column sep=4em]
      X & Y \\
      {\orch(\csaa, \csaa, X)} & {\orch(\csaa, \csaa, Y)}
      \arrow["f", from=1-1, to=1-2]
      \arrow["{\eta_{\csaa X}}"', from=1-1, to=2-1]
      \arrow["{\eta_{\csaa Y}}", from=1-2, to=2-2]
      \arrow["{\orch(\csaa, \csaa, f)}", from=2-1, to=2-2]
    \end{tikzcd}
  \]
  for \(f \in \dcpo(X, Y)\). We calculate:
  \[ \eta_{\csaa Y}(f(x))_k = k(f(x)) = \eta_{\csaa X}(x)_{\lambda x.k(f(x))} = Q(\csaa, \csaa, f)(\eta_{\csaa X}(x))_k\]
  for all \(x\) and \(k\). For \(\phi : \wstarcat(\csab, \csaa)\), the dinaturality condition is the commutativity of:
  \[
    \begin{tikzcd}[row sep = tiny]
      & {\orch(\csaa,\csaa, X)} & \\
      X && {\orch(\csaa, \csab, X)} \\
      & {Q(\csab, \csab, X)}
      \arrow["{\orch(\csaa, \phi, X)}", from=1-2, to=2-3]
      \arrow["{\eta_{\csaa X}}", from=2-1, to=1-2]
      \arrow["{\eta_{\csab X}}"', from=2-1, to=3-2]
      \arrow["{\orch(\phi, \csab, X)}"', from=3-2, to=2-3]
    \end{tikzcd}
  \]
  By again applying both sides to arbitrary \(x\) and \(k\) and expanding definitions, we get:
  \[ \orch(\csaa, \phi, X)(\eta_{\csaa X}(x))_k = \phi \circ k(x) = \orch(\phi, \csab, X)(\eta_{\csab X}(x))_k \]
  Hence, \(\eta\) is natural.
\end{proof}

\subsection{The multiplication is a natural transformation}
\label{sec:mult-natur-transf}

Similarly to the unit, we show that the multiplication is a well-defined natural transformation in two parts, first showing that each component is a well-defined dcpo morphism:
\[ \orch(\csaa, \csab, \orch(\csab, \csac, X)) \to \orch(\csaa, \csac, X) \]
and secondly showing that the required naturality squares hold.

\begin{lemma}
  The multiplication \(\mu\) is well-defined: for \(\csaa, \csab, \csac \in \wstarcat\) and dcpo \(X\), we have that:
  \[ \mu_{\csaa, \csab, \csac, X}(\Xi) \in \orch(\csaa, \csac, X)\]
  for \(\Xi \in \orch(\csaa, \csab, \orch(\csab, \csac, X))\), and the function \(\mu_{\csaa, \csab, \csac, X}\) is continuous.
\end{lemma}
\begin{proof}
  Suppose \(\Xi \in \orch(\csaa, \csab, \orch(\csab, \csac, X))\). Then we must show \(\mu_{\csaa, \csab, \csac, X}(\Xi) \in \orch(\csaa, \csac, X)\). For continuity, if \(K \subseteq \dcpo(X, \wstarcat(\csak, \csab))\) is directed, then:
  \begin{align*}
    \mu_{\csaa, \csab, \csac, X}(\Xi)_{\sup K}
    &= \Xi_{\lambda \xi. \xi_{\sup K}}\\
    &= \Xi_{\sup \{\lambda \xi. \xi_k \mid k \in K\}}&\text{by continuity of }\xi\\
    &= \sup \{ \Xi_{\lambda \xi. \xi_k} \mid k \in K\}&\text{by continuity of }\Xi\\
    &= \sup \{ \mu_{\csaa, \csab, \csac, X}(\Xi)_k \mid k \in K\}
  \end{align*}
  as required. Subconvexity holds similarly. For compositionality, if \(\chi \in \wstarcat(\csak', \csak)\), then:
  \[\mu_{\csaa, \csab, \csac, X}(\Xi)_{\lambda x. k(x) \circ \chi} = \Xi_{\lambda \xi.\xi_{\lambda x. k(x) \circ \chi}} = \Xi_{\lambda \xi. \xi_k \circ \chi} = \Xi_{\lambda \xi. \xi_k} \circ \chi = \mu_{\csaa, \csab, \csac, X}(\Xi)_k \circ \chi\]
  as required, using the compositionality of each \(\xi\) and \(\Xi\). Hence, \(\mu\) is a well-defined function. In order to show that it is a continuous function we let \(\Delta \subseteq \orch(\csaa, \csab, \orch(\csab, \csac, X))\) be directed. Then:
  \begin{align*}
    \mu_{\csaa, \csab, \csac, X}(\sup \Delta)_k
    = (\sup \Delta)_{\lambda \xi. \xi_k}
    = \sup \{\Xi_{\lambda \xi. \xi_k} \mid \Xi \in \Delta\}
    = (\sup \{\mu_{\csaa, \csab, \csac, X}(\Xi) \mid \Xi \in \Delta\})_k
  \end{align*}
  for any \(k\), completing the proof.
\end{proof}

\begin{lemma}
  The multiplication \(\mu\) is a natural transformation.
\end{lemma}

\begin{proof}
  To show \(\mu\) is natural we need dinaturality in \(\csab\) and naturality in \(\csaa\), \(\csac\), and \(X\). We tackle the naturality first, for which we need tho following to commute:
  \[
    \begin{tikzcd}[column sep=7em]
      {\orch(\csaa, \csab, \orch(\csab, \csac, X))} & {\orch(\csaa', \csab, \orch(\csab, \csac', Y))} \\
      {\orch(\csaa, \csac, X)} & {\orch(\csaa', \csac', Y)}
      \arrow["{\orch(\phi, \csab, \orch(\csab, \psi, f))}", from=1-1, to=1-2]
      \arrow["{\mu_{\csaa, \csab, \csac, X}}"', from=1-1, to=2-1]
      \arrow["{\mu_{\csaa', \csab, \csac', Y}}", from=1-2, to=2-2]
      \arrow["{\orch(\phi, \psi, f)}", from=2-1, to=2-2]
    \end{tikzcd}
  \]
  for \(\phi : \wstarcat(\csaa, \csaa')\), \(\psi \in \wstarcat(\csac', \csac)\), and \(f : X \to Y\). For each \(\Xi\) and \(k\), we have:
  \begin{align*}
    \mu_{\csaa', \csab, \csac', Y} (\orch(\phi, \csab, \orch(\csab, \psi, f)) (\Xi))_k
    = \phi \circ \Xi_{\lambda \xi. \xi_{\lambda x. \psi \circ k(x)}}
    = \orch(\phi, \psi, f)(\mu_{\csaa, \csab, \csac, X}(\Xi))_k
  \end{align*}
  simply by expanding definitions on each side.
  For dinaturality, let \(\phi \in \wstarcat(\csab', \csab)\). Then the following diagram must commute:
  \[
    \begin{tikzcd}[row sep=small]
      & {\orch(\csaa, \csab',\orch(\csab',\csac, X))} & \\
      {\orch(\csaa, \csab, \orch(\csab', \csac, X))} && {\orch(\csaa, \csac, X)} \\
      & {\orch(\csaa, \csab, \orch(\csab, \csac, X))}
      \arrow["{\mu_{\csaa, \csab', \csac, X}}", from=1-2, to=2-3]
      \arrow["{\orch(\csaa, \phi, \orch(\csab', \csac, X))}", from=2-1, to=1-2]
      \arrow["{\orch(\csaa, \csab, \orch(\phi, \csac, X))}"', from=2-1, to=3-2]
      \arrow["{\mu_{\csaa, \csab, \csac, X}}"', from=3-2, to=2-3]
    \end{tikzcd}
  \]
  Then again by expanding definitions we get:
  \begin{align*}
    \mu_{\csaa, \csab, \csac, X}(\orch(\csaa, \csab, \orch(\phi, \csac, X))(\Xi))_k
    = \Xi_{\lambda \xi. \phi \circ \xi_k}
    = \mu_{\csaa, \csab', \csac, X}(\orch(\csaa, \phi, \orch(\csab', \csac, X))(\Xi))_k
  \end{align*}
   for each \(\Xi\) and \(k\) and so \(\mu\) is natural as required.
\end{proof}

\subsection{\texorpdfstring{\(\bm{\orch}\)}{Q} is a monad}
\label{sec:boldsymb-orch-monad}

We can now show that \(\orch\) does indeed form a monad.

\begin{theorem}
  \label{thm:is-a-monad}
  \((\orch, \eta, \mu)\) is a \(\wstarcat^{\text{op}}\)-parameterised monad on \(\dcpo\).
\end{theorem}

\begin{proof}
  We check each monad law in turn. The first unitality law states that:
  \[\mu_{\csaa, \csab, \csab, X} \circ \orch(\csaa, \csab, \eta_{\csaa X}) = \id_{\orch(\csaa, \csab, X)}\]
  Taking \(\xi \in \orch(\csaa, \csab, X)\) and \(k \in \dcpo(X, \wstarcat(\csak, \csab))\):
  \begin{align*}
    \mu_{\csaa, \csab, \csab, X}(\orch(\csaa, \csab, \eta_{\csaa X})(\xi))_k
    = \orch(\csaa, \csab, \eta_{\csaa X})(\xi))_{\lambda \xi'. \xi'_k}
    = \xi_{\lambda x. \eta_{\csaa X}(x)_k}
    = \xi_k
  \end{align*}
  The second unitality law is \(\mu_{\csaa, \csaa, \csab, X} \circ \eta_{\csaa \orch(\csaa, \csab, X)} = \id_{\orch(\csaa, \csab, X)}\). Taking the same \(\xi\) and \(k\) we have:
  \begin{align*}
    \mu_{\csaa, \csaa, \csab, X}(\eta_{\csaa \orch(\csaa, \csab, X)}(\xi))_k
    = \eta_{\csaa \orch(\csaa, \csab X)}(\xi)_{\lambda \xi'. \xi'_k}
    = \xi_k
  \end{align*}
  The associativity monad law is the commutativity of the following diagram:
  \[
    \begin{tikzcd}[column sep = 8em]
      {\orch(\csaa, \csab, \orch(\csab, \csac, \orch(\csac, \csad, X)))} & {\orch(\csaa, \csac, \orch(\csac, \csad, X))} \\
      {\orch(\csaa, \csab, \orch(\csab, \csad, X))} & {\orch(\csaa, \csad, X)}
      \arrow["{\mu_{\csaa, \csab, \csac, \orch(\csac, \csad, X)}}", from=1-1, to=1-2]
      \arrow["{\orch(\csaa, \csab, \mu_{\csab, \csac, \csad, X})}"', from=1-1, to=2-1]
      \arrow["{\mu_{\csaa, \csac, \csad, X}}", from=1-2, to=2-2]
      \arrow["{\mu_{\csaa, \csab, \csad, X}}", from=2-1, to=2-2]
    \end{tikzcd}
  \]
  By expanding definitions on both sides we have:
    \begin{align*}
    \mu_{\csaa, \csab, \csad, X}(\orch(\csaa, \csab, \mu_{\csab, \csac, \csad, X})(\aleph))
    = \aleph_{\lambda \Xi. \Xi_{\lambda \xi. \xi_k}}
    = \mu_{\csaa, \csac, \csad, X}(\mu_{\csaa, \csab, \csac, \orch(\csac, \csad, X)}(\aleph))_k
  \end{align*}
  for \(\aleph \in \orch(\csaa, \csab, \orch(\csab, \csac, \orch(\csac, \csad, X)))\) and \(k \in \dcpo(X, \wstarcat(\csak, \csad))\), and so the square commutes. With these three laws proven, we have shown that \(\orch\) is a monad.
\end{proof}

\subsection{Strength of the monad}
\label{sec:strength-monad}

Finally, we show that the monad we have constructed is strong. Much like \cref{sec:unit-natur-transf,sec:mult-natur-transf}, we must first show that \(\tau\) is well-defined before showing that it is a natural transformation.

\begin{lemma}\label{lemma:orchestra-strength}
  The strength \(\tau\) is well-defined, that is:
  \[ \tau_{X, \csaa, \csab, Y} \in \dcpo(X \times \orch(\csaa, \csab, Y), \orch(\csaa, \csab, X \times Y))\]
  for dcpos \(X\) and \(Y\), and \(\csaa, \csab \in \wstarcat\).
\end{lemma}
\begin{proof}
  We first show that \(\tau_{X, \csaa, \csab, Y}(x, \xi) \in \orch(\csaa, \csab, X \times Y)\) for \(x \in X\) and \(\xi \in \orch(\csaa, \csab, Y)\). Let \(K \subseteq \dcpo(X \times Y, \wstarcat(\csak, \csab))\) be directed. Then:
  \[\tau_{X, \csaa, \csab, Y}(\xi)_{\sup K} = \xi_{\lambda y. \sup K(x, y)} = \sup \{ \xi_{\lambda y. k(x, y)} \mid k \in K\} = \sup \{ \tau_{X, \csaa, \csab, Y}(\xi)_k \mid k \in K\}\]
  Subconvexity holds similarly: given \(\theta, \rho\) such that \(\theta + \rho \leq 1\) and \(k, k'\), we have:
  \begin{align*}
    \tau_{X, \csaa, \csab, Y}(x, \xi)_{\lambda z. \theta k(z) + \rho k'(z)}
    &= \xi_{\lambda y. \theta k(x, y) + \rho k'(x, y)}\\
    &= \theta \xi_{\lambda y.k(x, y)} + \rho \xi_{\lambda y. k'(x, y)}&\text{by subconvexity of }\xi\\
    &= \lambda \tau_{X, \csaa, \csab, Y}(x, \xi)_k + \rho \tau_{X, \csaa, \csab, Y}(x, \xi)_{k'}
  \end{align*}
  For compositionality, we have:
  \begin{equation*}
    \tau_{X, \csaa, \csab, Y}(x, \xi)_{\lambda z. k(z) \circ \chi} = \xi_{\lambda y. k(x, y) \circ \chi} = \xi_{\lambda y. k(x, y)} \circ \chi
    = \tau_{X, \csaa, \csab, Y}(x, \xi)_k \circ \chi
  \end{equation*}
  for \(\chi \in \wstarcat(\csak', \csak)\) and \(k \in \dcpo(X, \wstarcat(\csak, \csab))\), using the compositionality of \(\xi\). Therefore, \(\tau_{X, \csaa, \csab, Y}(\xi) \in \orch(\csaa, \csab, X \times Y)\) as required.

  To show \(\tau_{X, \csaa, \csab, Y} \in \dcpo(X \times \orch(\csaa, \csab, Y), \orch(\csaa, \csab, X \times Y))\), it remains to show that it is continuous, for which it is sufficient to show continuity in each argument separately. First let \(D \subseteq X\) be directed. Then for fixed \(\xi\) and \(k\):
  \begin{align*}
    \tau_{X, \csaa, \csab, Y}(\sup D, \xi)_k
    &= \xi_{\lambda y. k(\sup D, y)}\\
    &= \xi_{\sup \{\lambda y. k(d, y) \mid d \in D\}}&\text{by continuity of }k\\
    &= \sup \{ \xi_{\lambda y. k(d, y)} \mid d \in D\}&\text{by continuity of }\xi\\
    &= \sup \{ \tau_{X , \csaa, \csab, Y}(d, \xi) \mid d \in D\}_k
  \end{align*}
  Now let \(\Delta \subseteq \orch(\csaa, \csab, Y)\) be directed. Then for fixed \(x\) and \(k\):
  \begin{align*}
    \tau_{X, \csaa, \csab, Y}(x, \sup \Delta)_k = (\sup \Delta)_{\lambda y. k(x, y)} = \{\xi_{\lambda y. k(x, y)} \mid \xi \in \Delta\}
    &= \sup\{\tau_{X ,\csaa, \csab, Y}(x, \xi) \mid \xi \in \Delta\}_k
  \end{align*}
  by the definition of the supremum on \(\orch(\csaa, \csab, Y)\). Hence, \(\tau\) is well-defined.
\end{proof}

\begin{lemma}
  \(\tau\) is a natural transformation.
\end{lemma}
\begin{proof}
  Let \(f : X \to X'\), \(g : Y \to Y'\), \(\phi : \wstarcat(\csaa, \csaa')\), and \(\psi : \wstarcat (\csab', \csab)\). Then the following must commute:
  \[
    \begin{tikzcd}
      {X \times \orch(\csaa, \csab, Y)} & {X' \times \orch(\csaa', \csab', Y')} \\
      {\orch(\csaa, \csab, X \times Y)} & {\orch(\csaa', \csab', X' \times Y')}
      \arrow["{f \times \orch(\phi, \psi, g)}", far, from=1-1, to=1-2]
      \arrow["{\tau_{X, \csaa, \csab, Y}}"', from=1-1, to=2-1]
      \arrow["{\tau_{X', \csaa', \csab', Y'}}", from=1-2, to=2-2]
      \arrow["{\orch(\phi, \psi, f \times g)}", far, from=2-1, to=2-2]
    \end{tikzcd}
  \]
  For \(x \in X\), \(\xi \in \orch(\csaa, \csab, Y)\), and \(k \in \dcpo(X' \times Y', \wstarcat(\csak, \csab'))\), we get:
  \begin{align*}
    \tau_{X', \csaa', \csab', Y'}((f \times \orch(\phi, \psi, g))(x, \xi))_k
    &= \phi \circ \xi_{\lambda y. \psi \otimes \csae \circ k(f(x), g(y))}\\
    &= \orch(\phi, \psi, f \times g)(\tau_{X, \csaa, \csab, Y}(x, \xi))_k
  \end{align*}
  by expanding definitions and so \(\tau\) is a natural transformation.
\end{proof}

\begin{theorem}
  \label{thm:strong-monad}
  \((\orch, \eta, \mu, \tau)\) is a \(\wstarcat^{\text{op}}\)-parameterised strong monad on \(\dcpo\).
\end{theorem}
\begin{proof}
  From \cref{thm:is-a-monad}, we have that \((\orch, \eta, \mu)\) is a monad. It remains to prove the appropriate monad laws for the strength.
  Letting \(1 = \{*\}\) and \(L_X : 1 \times X \to X\) be the left unitor isomorphism, we need:
  \[ \orch(\csaa, \csab, L_X) \circ \tau_{1, \csaa, \csab, X} = L_{\orch(\csaa, \csab, X)}\]
  For \(\xi \in \orch(\csaa, \csab, X)\) and \(k \in \dcpo(X, \wstarcat(\csak, \csab))\), we have:
  \begin{align*}
    \orch(\csaa, \csab, L_X)(\tau_{1, \csaa, \csab, X}(*, \xi))_k
    = \tau_{1, \csaa, \csab, X}(*, \xi)_{\lambda z. k(L_X(z))}
    = \xi_{\lambda x. k(L_X(*, x))}
    = \xi_k
    = L_{\orch(\csaa, \csab, X)}(*, \xi_k)
  \end{align*}
  For compatibility with the associator \(\alpha_{X, Y, Z} : (X \times Y) \times Z \to X \times (Y \times Z)\), the following must commute:
  \[
    \begin{tikzcd}
      {(X \times Y) \times \orch(\csaa, \csab, Z)} && {\orch(\csaa, \csab, (X \times Y) \times Z)} \\
      {X \times (Y \times \orch(\csaa, \csab, Z))} & {X \times \orch(\csaa, \csab, Y \times Z)} & {\orch(\csaa, \csab, X \times (Y \times Z))}
      \arrow["{\tau_{X \times Y, \csaa, \csab, Z}}", from=1-1, to=1-3]
      \arrow["{\alpha_{X, Y, \orch(\csaa, \csab, Z)}}"', from=1-1, to=2-1]
      \arrow["{\orch(\csaa, \csab, \alpha_{X,Y,Z})}", from=1-3, to=2-3]
      \arrow["{\id_X \times \tau_{Y, \csaa, \csab, Z}}", far, from=2-1, to=2-2]
      \arrow["{\tau_{X, \csaa, \csab, Y \times Z}}", far, from=2-2, to=2-3]
    \end{tikzcd}
  \]
  For \(x \in X\), \(y \in Y\), \(\xi \in \orch(\csaa, \csab, Z)\), and \(k \in \dcpo(X \times (Y \times Z), \wstarcat(\csak, \csab))\), we have:
  \begin{align*}
    \tau_{X, \csaa, \csab, Y \times Z}((\id_X \times \tau_{Y, \csaa, \csab, Z})(\alpha_{X, Y, \orch(\csaa, \csab, Z)}((x, y), \xi)))_k
    &= \xi_{\lambda z. k(x, (y, z))}\\
    &=\orch(\csaa, \csab, \alpha_{X, Y, Z})(\tau_{X \times Y, \csaa, \csab, Z}((x, y), \xi))_k
  \end{align*}
  by unfolding definitions.
  For compatibility with the unit we need:
  \[ \tau_{X, \csaa, \csaa, Y} \circ (\id_X \times \eta_{\csaa Y}) = \eta_{\csaa X \times Y}\]
  Given \(x \in X\), \(y \in Y\), and \(k \in \dcpo(X \times Y, \wstarcat(\csak, \csaa))\):
  \begin{align*}
    \tau_{X, \csaa, \csaa, Y}(x, \eta_{\csaa Y}(y))_k = \eta_{\csaa Y}(y)_{\lambda y. k(x, y)} = k(x, y) = \eta_{\csaa X\times Y}(x, y)_k
  \end{align*}
  Finally for compatibility with \(\mu\) we need the following to commute:
  \[
    \begin{tikzcd}
      {X \times \orch(\csaa, \csab, \orch(\csab, \csac, Y))} & {\orch(\csaa, \csab, X \times \orch(\csab, \csac, Y))} & {\orch(\csaa, \csab, \orch(\csab, \csac, X \times Y))} \\
      {X \times \orch(\csaa, \csac, Y)} && {\orch(\csaa, \csac, X \times Y)}
      \arrow["{\tau_{X, \csaa, \csab, \orch(\csab, \csac, Y)}}", far, from=1-1, to=1-2]
      \arrow["{\id_X \times \mu_{\csaa, \csab, \csac, Y}}", from=1-1, to=2-1]
      \arrow["{\orch(\csaa, \csab, \tau_{X, \csab, \csac, Y})}", far, from=1-2, to=1-3]
      \arrow["{\mu_{\csaa, \csab, \csac, X \times Y}}", from=1-3, to=2-3]
      \arrow["{\tau_{X, \csaa, \csac, Y}}", from=2-1, to=2-3]
    \end{tikzcd}
  \]
  So taking \(x \in X\), \(\Xi \in \orch(\csaa, \csab, \orch(\csab, \csac, Y))\), and \(k \in \dcpo(X \times Y, \wstarcat(\csak, \csac))\):
  \begin{align*}
    \mu_{\csaa, \csab, \csac, X \times Y}(\orch(\csaa, \csab, \tau_{X, \csab, \csac, Y})(\tau_{X, \csaa, \csab, \orch(\csab, \csac, Y)}(x, \Xi)))_k
    &= \Xi_{\lambda \xi. \xi_{\lambda y. k(x, y)}}\\
    &= \tau_{X, \csaa, \csac, Y}(x, \mu_{\csaa, \csab, \csac, Y}(\Xi))_k
  \end{align*}
  again by simply unfolding definitions on both sides. Hence, all of the monad laws hold and \((\orch, \eta, \mu, \tau)\) is a \(\wstarcat^{\text{op}}\)-parameterised strong monad.
\end{proof}

\section{Discussion and future work}%
\label{sec:discussion-and-future-work}

In this paper, we presented the \emph{quantum orchestra} monad, which captures quantum side effects by describing them as the most general form of physical evolution of quantum systems, subunital completely positive maps on \wstar-algebras. As the monad acts on the category \(\dcpo\) of directed complete partially ordered sets, it naturally enables the interpretation of hybrid classical-quantum programs, including those relying on unbounded recursion.

The quantum orchestra monad heavily draws inspiration from \emph{continuation passing-style} monads, in agreement with other monad constructions based on the quantum instrument formalism \cite{booth2026composingquantuminstruments,fritz2026quantuminstrumentmonad}.
This seems to indicate that quantum effects of programs are most naturally described in CPS, rather than in direct style.
We leave the question whether a quantum orchestra-like monad with direct-style flavour is possible for future investigation.

\paragraph{Order structure of quantum orchestras and \(\wstarcat(\csab,\csaa)\)-valued valuations}
As hinted at in \cref{sec:continuous-orchestras}, one main advantage of defining the quantum orchestra monad in CPS is the possibility of defining the order of quantum orchestras pointwise for all continuations.
Valuations that take values in some \(\wstarcat(\csab,\csaa)\) can be seen as a direct-style alternative to quantum orchestras, and in fact orchestras restrict to valuations, by considering the components associated to indicator functions.
It remains unclear how the pointwise order of these valuations however relates to the order of quantum orchestras, as the former only ``tests'' on indicator functions as opposed to all continuations, and thus has the potential to be insufficiently coarse to render integration continuous.
Interestingly, it is known that these two orders coincide in the case of scalar-valued valuations, a consequence of the max-flow min-cut theorem \cite{jones:probabilisticpowerdomain}.
While the additional lattice-structure of commutative \wstar-algebras (their self-adjoint parts form in fact Dedekind complete Riesz spaces \cite{Dodds_1974}) makes it believable that an analogous correspondence might hold more generally for valuations and quantum orchestras over commutative \wstar-algebras, thus allowing for continuous integration along valuations, we leave this problem, as well as the more general non-commutative case for future research.

\paragraph{Relation to the probabilistic powerdomain monad}
Given any \(X \in \dcpo\), we observe that all orchestras \( \xi \in \orch(\mathbb{C}, \mathbb{C}, X) \) restrict to \([0,1]\)-valued valuations on indicator functions.
Compositionality of quantum orchestras guarantees that this mapping is injective, and the pointwise nature of the order of quantum orchestras implies that this mapping is continuous. The structure of such quantum orchestras is thus entirely determined by the underlying valuations.
Vice versa, any \([0,1]\)-valued valuation on \(X\) yields a unique continuous integral for continuous scalar-valued functions on \(X\) which indeed gives rise a full, consistent orchestra.
Perhaps surprisingly, keeping in mind that quantum orchestras were entirely defined in CPS, it turns out that the thus found relation between the simplest quantum orchestras and probabilistic powerdomains is in fact an \emph{isomorphism of monads}.
The full proof of this claim can be found in \cref{app:comparison-prob-powerdomain}.
\begin{theorem}
	\( \orch(\mathbb{C}, \mathbb{C}, -) \) and the probabilistic powerdomain monad \( \valu \) are equivalent.
\end{theorem}
Note that the commutativity of the probabilistic powerdomain monad, which is equivalent to a Fubini property of the valuations integral, has been an open problem since its discovery \cite{carette2023central}, and thus this question also affects \( \orch(\mathbb{C}, \mathbb{C}, -) \).
Beyond this special case, \( \orch(\csaa, \csaa, -) \) is not commutative for any \wstar-algebra \(\csaa\) which is not equal to \(\mathbb{C}\), which reflects the fact that the composition of morphisms in \(\wstarcat(\csaa,\csaa)\) is non-commutative, even if \(\csaa\) is abelian.
This perspective indicates that the multiplication of the quantum orchestra monad implements a form of \emph{non-commutative integral}.

\paragraph{Denotational semantics for existing languages}
In this work, we have seen how the new quantum orchestra monad could be used to give semantics for a prototype hybrid classical-quantum language, exhibiting the monad's versatility.
This makes it a natural and flexible tool to further develop the semantics of existing languages in the QRAM model.
In this context, it would be interesting to explore the possibility of accompanying orchestra-based denotational semantics with a compatible operational semantics.
We leave this question for future work.

\begin{acks}
  AR was funded by \grantsponsor{Rubberduq}{UKRI}{https://www.ukri.org/} grant \grantnum{Rubberduq}{EP/X025551/1}: Rubber DUQ: Flexible Dynamic Universal Quantum programming.
  DL acknowledges support from the Quantum Advantage Pathfinder research program within the UK's National Quantum Computing Center.
  RIB was supported by the \grantsponsor{EPSRC-DeQS}{Engineering and Physical Sciences Research Council}{https://www.ukri.org/councils/epsrc/} grant reference \grantnum{EPSRC-DeQS}{EP/Z002230/1}: (De)constructing quantum software (DeQS).
\end{acks}

\bibliographystyle{ACM-Reference-Format}
\bibliography{bibliography}

\appendix
\crefalias{section}{appendix}

\section{Toy Language - Full Description}
\label{app:toy}

The final typing rules are given below. We first begin with the rules for values.
\begin{mathpar}
  \inferrule{x : X \in \Gamma}{\Gamma \judgementvalue x : X}\and
  \inferrule{ }{\Gamma \judgementvalue () : \unit}\and
  \inferrule{\Gamma \judgementvalue v_1 : X \and \Gamma \judgementvalue v_2 : Y}{\Gamma \judgementvalue (v_1, v_2) : X \times Y}\and
  \inferrule{\Gamma \judgementvalue v : X \times Y}{\Gamma \judgementvalue \fst v : X}\and
  \inferrule{\Gamma \judgementvalue v : X \times Y}{\Gamma \judgementvalue \snd v : Y}\and
  \inferrule{ }{\Gamma \judgementvalue \tt : \bool}\and
  \inferrule{ }{\Gamma \judgementvalue \ff : \bool}\and
  \inferrule{\Gamma, x : X \judgementproducer_n e : Y}{\Gamma \judgementvalue \lambda x. e : X \to_n Y}\and
  \inferrule{\textsc{G} \in \{\Had, \SG, \X, \T\} }{\Gamma \judgementvalue \textsc{G} : \qid \to_0 \unit}\and
  \inferrule{ }{\Gamma \judgementvalue \CNOT : \qid \times \qid \to_0 \unit}\and
  \inferrule{ }{\Gamma \judgementvalue \meas : \qid \to_0 \bool}\and
  \inferrule{ }{\Gamma \judgementvalue \alloc : \unit \to \qid}
\end{mathpar}
And then we have the rules for expressions.
\begin{mathpar}
  \inferrule{\Gamma \judgementvalue v : X}{\Gamma \judgementproducer_0 \return v : X}\and
  \inferrule{\Gamma \judgementproducer_{n_1} e_1 : X \and \Gamma, x : X \judgementproducer_{n_2} e_2 : Y}{\Gamma \judgementproducer_{n_1 + n_2} \letin{x}{e_1}e_2 : Y}\and
  \inferrule{\Gamma \judgementproducer_{n_1} e_1 : X \and \Gamma \judgementproducer_{n_2} e_2 : Y}{\Gamma \judgementproducer_{n_1 + n_2} e_1 \seq e_2 : Y}\and
  \inferrule{\Gamma \judgementvalue v_1 : X \to_n Y \and \Gamma \judgementvalue v_2 : X}{\Gamma \judgementproducer_n v_1(v_2) : Y}\and
  \inferrule{\Gamma \judgementvalue v : \bool \and \Gamma \judgementproducer_n e_1 : X \and \Gamma \judgementproducer_n e_2 : X}{\Gamma \judgementproducer_n \termif{v}{e_1}{e_2} : X}\and
  \inferrule{\Gamma \judgementvalue f : (1 \to_0 X) \to_0 X}{\Gamma \judgementproducer_0 \fix f : X}
\end{mathpar}
The semantics for types are parameterised over the number of qubits available and are given by:
\begin{mathpar}
  \sem{\unit}_m = \{()\} \and
  \sem{X \times Y}_m = \sem{X}_m \times \sem{Y}_m \and
  \sem{\bool}_m = \{\tt, \ff\} \and
  \sem{X \to_n Y}_m = \forall m' \geq m. \dcpo(\sem{X}_{m'}, \orch(\mathcal{B}(\mathcal{H}^{m'}), \mathcal{B}(\mathcal{H}^{m' + n}), \sem{Y}_{m' + n})) \and
  \sem{\qid}_m = \{0, \dots, m - 1\}
\end{mathpar}
where we let \(\mathcal{H} = \mathbb{C}^2\) and the semantics of contexts is given by:
\[ \sem{x_1 : X_1, \dots, x_n : X_n}_m = \sem{X_1}_m \times \dots \times \sem{X_n}_m \]
We note that the definition of the function type is quantified over the number of qubits. This trick allows us to define a function:
\[ \iota_{X,m,m'} : \sem{X}_m \to \sem{X}_{m'}\]
which can be extended from types \(X\) to contexts \(\Gamma\).

The semantics for a value \(\Gamma \judgementvalue v : X\) is a (continuous) function:
\[ {v}_m \in \dcpo(\sem{\Gamma}_m, \sem{X}_m) \]
We define it for each constructor as follows:
\begin{align*}
  \sem{x}_m(\semctx) &= \semctx_x & \sem{(v_1, v_2)}_m(\semctx) &= (\sem {v_1}_m(\semctx), \sem {v_2}_m(\semctx))\\
  \sem {()}_m(\semctx) &= () & \sem {\fst v}_m(\semctx) &= \pi_1(\sem{v}_m(\semctx))\\
  \sem{\tt}_m(\semctx) &= \tt & \sem{\snd v}_m(\semctx) &= \pi_2(\sem{v}_m(\semctx))\\
  \sem{\ff}_m(\semctx) &= \ff & \sem{\lambda x. e}_m(\semctx)(\semctx_x) &= \sem{e}_{m'}(\iota_{\Gamma, m, m'}(\semctx), \semctx_x)\\
  \sem{\Had}_m(\semctx)(q) &= \delta((), {\overline \Had_q}) & \sem{\T}_m(\semctx)(q) &= \delta((), {\overline \T_q})\\
  \sem{\SG}_m(\semctx)(q) &= \delta((), {\overline \SG_q}) & \sem{\CNOT}_m(\semctx)((q_1, q_2)) &= \delta((), \overline \CNOT_{q_1, q_2})\\
  \sem{\X}_m(\semctx)(q) &= \delta((), {\overline \X_q}) & \sem{\meas}_m(\semctx)(q) &= \delta(\ff, \overline{\ketbra{0}}_q) + \delta(\tt, \overline{\ketbra{1}}_q)\\
  \sem{\alloc}_m(\semctx)(u) &= \delta(m, \overline V) &\text{where } V(v) &= v \otimes \ket{0}
\end{align*}
where \(\semctx_x\) refers to the \(x^{\text{th}}\) projection from the semantics of the context, and we recall that for a linear map \(f\) we define:
\[ \overline f = \rho \mapsto f^\dagger \circ \rho \circ f\]
and then use subscripts to refer to the qubits of \(\mathcal{H}^m\) that this channel should target.

Expressions \(\Gamma \judgementproducer_n e : X\) have semantics given by (continuous) functions:
\[ \sem{e}_m \in \dcpo(\sem{\Gamma}_m, \orch(\mathcal{B}(\mathcal{H}^m), \mathcal{B}(\mathcal{H}^{m+n}), \sem{X}_{m + n}))\]
The semantics for each constructor of expressions are given by the following composites:
\begin{align*}
  \sem{\return v}_m ={} &\sem{\Gamma}_m \xrightarrow{\sem{v}_m} \sem{X}_m \xrightarrow{\eta_{\mathcal{B}(\mathcal{H}^m) X}} \orch(\mathcal{B}(\mathcal{H}^m), \mathcal{B}(\mathcal{H}^m), \sem{X}_m)\\
  \sem{\letin {x}{e_1}e_2}_m ={} &\sem{\Gamma}_m\\
                        &\Big\downarrow\langle \iota_{\Gamma,m,m+n_1}, \sem{e_1}_m \rangle\\
                        &\sem{\Gamma}_{m+n_1} \times \orch(\mathcal{B}(\mathcal{H}^m), \mathcal{B}(\mathcal{H}^{m + n_1}), \sem{X}_{m + n_1})\\
                        &\Big\downarrow\tau_{\sem{\Gamma}_{m + n_1}, \mathcal{B}(\mathcal{H}^m), \mathcal{B}(\mathcal{H}^{m + n_1}), \sem{X}_{m + n_1}}\\
                        &\orch(\mathcal{B}(\mathcal{H}^m), \mathcal{B}(\mathcal{H}^{m + n_1}), \sem{\Gamma}_{m + n_1} \times \sem{X}_{m + n_1})\\
                        &\Big\downarrow \kleisli{\sem{e_2}_{m + n_1}}\\
                        &\orch(\mathcal{B}(\mathcal{H}^m), \mathcal{B}(\mathcal{H}^{m + n_1 + n_2}), \sem{Y}_{m + n_1 + n_2})\\
\intertext{where \(\Gamma \judgementproducer_{n_1} e_1 : X\) and \(\Gamma, x : X \judgementproducer_{n_2} e_2 : Y\)}
  \sem {e_1 \seq e_2}_m ={} &\sem{\Gamma}_m \\
                        &\Big\downarrow \langle \iota_{\Gamma, m, m+n_1}, \sem{e_1}_m \rangle\\
                        &\sem{\Gamma}_{m+n_1} \times \orch(\mathcal{B}(\mathcal{H}^m), \mathcal{B}(\mathcal{H}^{m + n_1}), \sem{X}_{m + n_1})\\
                        &\Big\downarrow\tau_{\sem{\Gamma}_{m + n_1}, \mathcal{B}(\mathcal{H}^m), \mathcal{B}(\mathcal{H}^{m + n_1}), \sem{X}_{m + n_1}}\\
                        &\orch(\mathcal{B}(\mathcal{H}^m), \mathcal{B}(\mathcal{H}^{m + n_1}), \sem{\Gamma}_{m + n_1} \times \sem{X}_{m + n_1})\\
                        &\Big\downarrow\orch(\mathcal{B}(\mathcal{H}^m), \mathcal{B}(\mathcal{H}^{m + n_1}), \pi_1)\\
                        &\orch(\mathcal{B}(\mathcal{H}^m), \mathcal{B}(\mathcal{H}^{m + n_1}), \sem{\Gamma}_{m + n_1})\\
                        &\Big\downarrow\kleisli{\sem{e_2}_{m + n_1}}\\
                        &\orch(\mathcal{B}(\mathcal{H}^m), \mathcal{B}(\mathcal{H}^{m + n_1 + n_2}), \sem{Y}_{m + n_1 + n_2})\\
\intertext{where \(\Gamma \judgementproducer_{n_1} e_1 : X\) and \(\Gamma, x : X \judgementproducer_{n_2} e_2 : Y\)}
  \sem {v_1(v_2)}_m ={} &\sem{\Gamma}_m \\
                        &\Big\downarrow\langle \sem{v_1}_m, \sem{v_2}_m \rangle\\
                        &(\dcpo(\sem{X}_m, \orch(\mathcal{B}(\mathcal{H}^m), \mathcal{B}(\mathcal{H}^{m + n_1}), \sem{Y}_{m + n_1}))) \times \sem{X}_m\\
                        &\Big\downarrow \mathsf{eval}\\
                        &\orch(\mathcal{B}(\mathcal{H}^m), \mathcal{B}(\mathcal{H}^{m + n}), \sem{Y}_{m + n})\\
  \intertext{where \(\Gamma \judgementvalue v_1 : X \to_n Y\)}
  \sem {\termif{v}{e_1}{e_2}}_m ={} &\sem{\Gamma}_m\\
                        &\Big\downarrow \langle \sem{v}_m, \id \rangle\\
                        &\sem{\bool}_m \times \sem{\Gamma}_m\\
                        &\Big\downarrow {}\simeq{}\\
                        &\sem{\Gamma}_m + \sem{\Gamma}_m\\
                        &\Big\downarrow [\sem{e_1}_m, \sem{e_2}_m]\\
                        &\orch(\mathcal{B}(\mathcal{H}^m), \mathcal{B}(\mathcal{H}^{m + n}), \sem{X}_{m + n})\\
  \sem{\fix f}_m ={} &\sem{\Gamma}_m \\
                        &\Big\downarrow{\sem{f}_m}\\
                        &\dcpo((\unit \to \orch(\mathcal{B}(\mathcal{H}^m), \mathcal{B}(\mathcal{H}^m), \sem{X}_m)), \orch(\mathcal{B}(\mathcal{H}^m), \mathcal{B}(\mathcal{H}^m), \sem{X}_m))\\
                        &\Big\downarrow {}\simeq{}\\
                        &\dcpo(\orch(\mathcal{B}(\mathcal{H}^m), \mathcal{B}(\mathcal{H}^m), \sem{X}_m), \orch(\mathcal{B}(\mathcal{H}^m), \mathcal{B}(\mathcal{H}^m), \sem{X}_m))\\
                        &\Big\downarrow \fix\\
                        &\orch(\mathcal{B}(\mathcal{H}^m), \mathcal{B}(\mathcal{H}^m), \sem{X}_m)
\end{align*}
where \(\langle f, g \rangle\) is the universal map into the product, and \([f, g]\) is the universal map out of the sum.

\section{Finite Quantum Instruments are a Strong Parametrised Monad}
\label{sec:quant-instr-are}

In this section we show that the finite quantum instruments used in \cref{sec:toy-inst,sec:toy-alloc} actually form a monad.

\begin{definition}[Parameterised Quantum Instrument Monad]
  We define the \(\wstarcat^{\text{op}}\)-parameterised quantum instrument monad \(Q\) as follows.
  \begin{itemize}
  \item Let \(\csaa\) and \(\csab\) be \wstar-algebras and \(X\) be a set. Then elements \(\xi \in Q(\csaa, \csab, X)\) are given by a collection \(\{\xi_x \in \wstarcat(\csab, \csaa) \mid x \in X\}\) such that:
    \begin{itemize}
    \item \(\Supp(\xi) \coloneq \{x \in X \mid \xi_x \neq 0 \}\) is finite.
    \item \(\sum_{x \in X} \xi_x\) is subunital.
    \end{itemize}
  \item To show that \(Q\) is a functor \(\wstarcat \times \wstarcat^{\text{op}} \times \Set \to \Set\), we let \(\phi : \wstarcat(\csaa, \csab)\), \(\psi : \wstarcat(\csab', \csaa')\), and \(f : \Set(X, Y)\), and define \(Q(\phi, \psi, f) : Q(\csaa, \csaa', X) \to Q(\csab, \csab', Y)\) by:
    \[ Q(\phi, \psi, f)(\xi)_y = \sum_{x \in f^{-1}(y)} \phi \circ \xi_x \circ \psi\]
    for \(\xi \in Q(\csaa, \csaa', X)\) and \(y \in Y\).
  \item We define the unit \(\eta_{\csaa X} : X \to Q(\csaa, \csaa, X)\) by:
    \[ \eta_{\csaa X}(x) = \delta(x, \id_{\csaa}) \]
    for \(x \in X\), recalling that \(\delta\) is defined as the \emph{Dirac quantum instrument}.
  \item Let the multiplication \(\mu_{\csaa, \csab, \csac, X} : Q(\csaa, \csab, Q(\csab, \csac, X)) \to Q(\csaa, \csac, X)\) be defined by:
    \[ \mu_{\csaa, \csab, \csac, X}(\Xi)_x = \sum_{\xi \in Q(\csab, \csac, X)} \Xi_\xi \circ \xi_x \]
    for \(x \in X\).
  \item For completeness, we explicitly give the strength \(\tau_{X, \csaa, \csab, Y} : X \times Q(\csaa, \csab, Y) \to Q(\csaa, \csab, X \times Y)\) by:
    \[ \tau_{X, \csaa, \csab, Y}(x, \xi)_{(x',y)} = \delta(x, \xi_y)_{x'} \]
    for each \(x, x' \in X\) and \(y \in Y\).
  \end{itemize}
\end{definition}

We use the multiplication instead of the Kleisli extension above to align more closely with the laws of a parameterised monad. The extension of a map \(f : X \to Q(\csab, \csac, Y)\) is given by:
\begin{align*}
  f_\csaa^\# &: Q(\csaa, \csab, X) \to Q(\csaa, \csac, Y)\\
  f_\csaa^\# &= \mu_{\csaa, \csab, \csac, Y} \circ Q(\csaa, \csab, f)
    \intertext{expanding out these definitions, we get:}
             f_\csaa^\#(\xi)_y &= \mu_{\csaa, \csab, \csac, Y}(Q(\csaa, \csab, f)(\xi))_y\\
            &= \sum_{\chi \in Q(\csab, \csac, Y)} \left( Q(\csaa, \csab, f)(\xi)_\chi \right) \circ \chi_y\\
            &= \sum_{\chi \in Q(\csab, \csac, Y)} \left( \sum_{x \in f^{-1}(\chi)} \xi_x \right) \circ \chi_y\\
            &= \sum_{\substack{\chi \in Q(\csab, \csac, Y)\\x \in f^{-1}(\chi)}} \xi_x \circ \chi_y\\
            &= \sum_{x \in X} \xi_x \circ f(x)_y
\end{align*}
for \(\xi \in Q(\csaa, \csab, X)\) and \(y \in Y\), aligning with the extension given in the previous section.

\begin{lemma}
  \label{lem:dirac-well-defined}
  If \(x \in X\) and \(\Phi \in \wstarcat(\csab, \csaa)\), then the Dirac quantum instrument \(\delta(x, \Phi) \in Q(\csaa, \csab, X)\) is well-defined. Further, we have:
  \[ Q(\phi, \psi, f)(\delta(x, \Phi)) = \delta(f(x), \phi \circ \Phi \circ \psi)\]
  for \(f : X \to Y\), \(\phi \in \wstarcat(\csaa, \csaa')\), and \(\psi \in \wstarcat(\csab', \csab)\).
\end{lemma}
\begin{proof}
  The support of \(\delta(x, \Phi)\) is \(\{x\}\) and so is finite. The sum of all the components is \(\Phi\), which is subunital by definition. Hence \(\delta(x, \Phi) \in Q(\csaa, \csab, X)\). Let \(f\), \(\phi\), and \(\psi\) be as in the statement of the lemma. Then:
  \begin{align*}
    Q(\phi, \psi, f)(\delta(x, \Phi))_y
    &= \sum_{x' \in f^{-1}(y)} \phi \circ \delta(x, \Phi)_{x'} \circ \psi \\
    &= \begin{cases}
      \phi \circ \Phi \circ \psi&\text{if }x \in f^{-1}(y)\\
      0&\text{otherwise}
    \end{cases}\\
    &=
      \begin{cases}
        \phi \circ \Phi \circ \psi &\text{if }f(x) = y\\
        0&\text{otherwise}
      \end{cases}\\
    &= \delta(f(x), \phi \circ \Phi \circ \psi)_y
  \end{align*}
  for all \(y\), and so the stated equality holds.
\end{proof}

\begin{lemma}
  \(Q\) is well-defined and a functor from \(\wstarcat \times \wstarcat^{op} \times \Set \to \Set\).
\end{lemma}
\begin{proof}
  Firstly,
  \[ Q(\id_\csaa, \id_\csab, \id_X)(\xi)_x = \sum_{x' \in \id_X^{-1}(x)} \id_\csaa \circ \xi_x \circ \id_\csab = \xi_x \]
  and so \(Q(\id_\csaa, \id_\csab, \id_X) = \id_{Q(\csaa, \csab, X)}\). Supposing \(\phi : \wstarcat(\csaa, \csab)\), \(\psi : \wstarcat(\csab', \csaa')\), and \(f : X \to Y\), we must show that \(Q(\phi, \psi, f)\) is well defined. Specifically, for \(\xi \in Q(\csaa, \csaa', X)\), we have that:
  \begin{align*}
    \Supp(Q(\phi, \psi, f)(\xi))
    &= \left\{y \in Y \mid \sum\nolimits_{x \in f^{-1}(y)} \phi \circ \xi_x \circ \psi \neq 0 \right\}\\
    &\subseteq \left\{y \in Y \mid \sum\nolimits_{x \in f^{-1}(y)} \xi_x \neq 0\right\}\\
    &\subseteq f(\Supp(\xi))
  \end{align*}
  and so the support is finite. Then:
  \begin{align*}
    \sum_{y \in Y} Q(\phi,\psi,f)(\xi)_y
    &= \sum_{\substack{y \in Y\\x \in f^{-1}(y)}} \phi \circ \xi_x \circ \psi \\
    &= \sum_{x \in X} \phi \circ \xi_x \circ \psi \\
    &= \phi \circ \left( \sum_{x \in X} \xi_i  \right) \circ \psi
  \end{align*}
  which is a composition of subunital maps and so is subunital. Hence, \(Q(\phi, \psi, f)(\xi) \in Q(\csab, \csab', Y)\) and so \(Q(\phi, \psi, f)\) is well-defined.

  If we further suppose that we have \(\phi': \wstarcat(\csab, \csac)\), \(\psi' : \wstarcat(\csac', \csab')\), and \(g : Y \to Z\), then:
  \begin{align*}
    Q(\phi', \psi', g) (Q(\phi,\psi,f) (\xi))_z
    &= \sum_{y \in g^{-1}(z)} \phi' \circ Q(\phi,\psi,f)(\xi)_y \circ \psi'\\
    &= \sum_{y \in g^{-1}(z)} \phi' \circ \left( \sum_{x \in f^{-1}(y)} \phi \circ \xi_x \circ \psi \right) \circ \psi'\\
    &= \sum_{\substack{y \in g^{-1}(z)\\x\in f^{-1}(y)}} \phi' \circ \phi \circ \xi_x \circ \psi \circ \psi' \\
    &= \sum_{x \in (g \circ f)^{-1}(z)} \phi' \circ \phi \circ \xi_x \circ \psi \circ \psi'\\
    &= Q(\phi' \circ \phi, \psi \circ \psi', g \circ f)(\xi)_z
  \end{align*}
  and so \(Q\) is a functor.
\end{proof}

\begin{lemma}
  The unit \(\eta\) is a natural transformation.
\end{lemma}
\begin{proof}
  As \(\eta\) is a Dirac quantum instrument it is well-defined by \cref{lem:dirac-well-defined}.

  To show naturality of \( \eta \), we need naturality in \(X\) and dinaturality in \(\csaa\). The first requires that the following square commutes for all \(\csaa\):
  \[
    \begin{tikzcd}[column sep=4em]
      X & Y \\
      {Q(\csaa, \csaa, X)} & {Q(\csaa, \csaa, Y)}
      \arrow["f", from=1-1, to=1-2]
      \arrow["{\eta_{\csaa X}}"', from=1-1, to=2-1]
      \arrow["{\eta_{\csaa Y}}", from=1-2, to=2-2]
      \arrow["{Q(\csaa, \csaa, f)}", from=2-1, to=2-2]
    \end{tikzcd}
  \]
  By \cref{lem:dirac-well-defined}, it holds that
  \[ Q(\csaa, \csaa, f)(\eta_{\csaa X}(x)) = Q(\csaa, \csaa, f)(\delta(x, \id_{\csaa})) = \delta(f(x), \id_{\csaa}) = \eta_{\csaa Y}(f(x)). \]
  Showing dinaturality in \(\csaa\) amounts to proving the following diagram commutes for \(\phi : \wstarcat(\csab, \csaa)\):
  \[
    \begin{tikzcd}[row sep=small]
      & {Q(\csaa,\csaa, X)} & \\
      X && {Q(\csaa, \csab, X)} \\
      & {Q(\csab, \csab, X)}
      \arrow["{Q(\csaa, \phi, X)}", from=1-2, to=2-3]
      \arrow["{\eta_{\csaa X}}", from=2-1, to=1-2]
      \arrow["{\eta_{\csab X}}"', from=2-1, to=3-2]
      \arrow["{Q(\phi, \csab, X)}"', from=3-2, to=2-3]
    \end{tikzcd}
  \]
  Letting \(x \in X\) and again using \cref{lem:dirac-well-defined}:
  \[ Q(\csaa, \phi, X)(\eta_{\csaa X}(x)) = Q(\csaa, \phi, X)(\delta(x, \id_{\csaa})) = \delta(x, \phi) = Q(\phi, \csab, X)(\delta(x, \id_{\csab})) = \eta_{\csab X}(x), \]
  and so \(\eta\) is natural as required.
\end{proof}

\begin{lemma}
  The multiplication \(\mu\) is a natural transformation.
\end{lemma}
\begin{proof}
  For an instrument \(\Xi \in Q(\csaa, \csab, Q(\csab, \csac, X))\) and \(x \in X\), we have that \(\mu_{\csaa, \csab, \csac}(\Xi)_x \neq 0\) implies that there is \(\xi \in \Supp(\Xi)\) with \(x \in \Supp(\xi)\). Therefore,
\[ \Supp \left( \mu_{\csaa, \csab, \csac}(\Xi) \right) \subseteq \bigcup_{\xi \in \Supp(\Xi)} \Supp(\xi), \]
with the right-hand side being a finite union of finite sets. We then note that
  \begin{align*}
    \sum_{x \in X} \mu_{\csaa, \csab, \csac, X}(\Xi)_x
    &= \sum_{\substack{x \in X\\\xi \in Q(\csab, \csac, X)}} \Xi_\xi \circ \xi_x
    = \sum_{\xi \in Q(\csab, \csac, X)} \Xi_\xi \circ \Bigl( \sum_{x \in X} \xi_x \Bigr), \\
    \intertext{and hence,}
    \sum_{x \in X} \mu_{\csaa, \csab, \csac, X}(\Xi)_x (1_\csac)
    &= \sum_{\xi \in Q(\csab, \csac, X)} \Xi_\xi \Bigl( \sum_{x \in X} \xi_x(1_\csac) \Bigr)
    \leq \sum_{\xi \in Q(\csab, \csac, X)} \Xi_\xi(1_\csab)
    \leq 1_\csaa,
  \end{align*}
  where the second equality holds as \(\sum_{x \in X} \xi_x\) is subunital for each \(\xi\) and the third holds as the sum of the components of \(\Xi\) is subunital.

  We now require naturality in \(\csaa, \csac\), and \(X\), with dinaturality in \(\csab\). We treat each case separately:
  \begin{itemize}
  \item Naturality in \(X\): We must show the following square commutes:
    \[
      \begin{tikzcd}[column sep=7em]
        {Q(\csaa, \csab, Q(\csab, \csac, X))} & {Q(\csaa, \csab, Q(\csab, \csac, Y))} \\
        {Q(\csaa, \csac, X)} & {Q(\csaa, \csac, Y)}
        \arrow["{Q(\csaa, \csab, Q(\csab, \csac, f))}", from=1-1, to=1-2]
        \arrow["{\mu_{\csaa, \csab, \csac, X}}"', from=1-1, to=2-1]
        \arrow["{\mu_{\csaa, \csab, \csac, Y}}", from=1-2, to=2-2]
        \arrow["{Q(\csaa, \csac, f)}", from=2-1, to=2-2]
      \end{tikzcd}
    \]
    Take \(\Xi \in Q(\csaa, \csab, Q(\csab, \csac, X))\). Then for \(y \in Y\) we have:
    \begin{align*}
      \mu_{\csaa, \csab, \csac, Y}(Q(\csaa, \csab, Q(\csab, \csac, f))(\Xi))_y
      &= \sum_{\xi \in Q(\csab, \csac, Y)} Q(\csaa, \csab, Q(\csab, \csac, f))(\Xi)_\phi \circ \xi_y\\
      &= \sum_{\substack{\xi \in Q(\csab, \csac, Y)\\ \xi' \in Q(\csab, \csac, f)^{-1}(\xi)}} \Xi_{\xi'} \circ \xi'_y\\
      &= \sum_{\xi' \in Q(\csab, \csac, X)} \Xi_{\xi'} \circ Q(\csab, \csac, f)(\xi')_y\\
      &= \sum_{\substack{x \in f^{-1}(y) \\ \xi' \in Q(\csab, \csac, X)}} \Xi_{\xi'} \circ \xi'_x \\
      &= \sum_{x \in f^{-1}(y)} \mu_{\csaa, \csab, \csac, X}(\Xi)_x\\
      &= Q(\csaa, \csac, f)(\mu_{\csaa, \csab, \csac, X}(\Xi))_y
    \end{align*}
    and so the square commutes.
  \item Naturality in \(\csaa\) and \(\csac\): Suppose \(\phi \in \wstarcat(\csaa, \csaa')\) and \(\psi \in \wstarcat(\csac', \csac)\), then the following square must commute:
    \[
      \begin{tikzcd}[column sep=7em]
	{Q(\csaa, \csab, Q(\csab, \csac, X))} & {Q(\csaa', \csab, Q(\csab, \csac', X))} \\
	{Q(\csaa, \csac, X)} & {Q(\csaa', \csac', X)}
	\arrow["{Q(\phi, \csab, Q(\csab, \psi, X))}", from=1-1, to=1-2]
	\arrow["{\mu_{\csaa, \csab, \csac, X}}"{description}, from=1-1, to=2-1]
	\arrow["{\mu_{\csaa', \csab, \csac', X}}"{description}, from=1-2, to=2-2]
	\arrow["{Q(\phi, \psi, X)}", from=2-1, to=2-2]
      \end{tikzcd}
    \]
    Again taking \(\Xi \in Q(\csaa, \csab, Q(\csab, \csac, X))\) and \(x \in X\) we have:
    \begin{align*}
      \mu_{\csaa', \csab, \csac', X} (Q(\phi, \csab, Q(\csab, \psi, X))(\Xi))
      &= \sum_{\xi' \in Q(\csab, \csac', X)} Q(\phi, \csab, Q(\csab, \psi, X))(\Xi)_{\xi'} \circ \xi'_x\\
      &= \sum_{\xi' \in Q(\csab, \csac', X)} \Bigl(\sum_{\xi \in Q(\csab, \psi, X)^{-1}(\xi')} \phi \circ \Xi_\xi \Bigl) \circ \xi'_x \\
      &= \sum_{\substack{\xi' \in Q(\csab, \csac', X)\\\xi \in Q(\csab, \psi, X)^{-1}(\xi')}} \phi \circ \Xi_\xi \circ \xi'_x\\
      &= \sum_{\xi \in Q(\csab, \csac, X)} \phi \circ \Xi_\xi \circ \xi_x \circ \psi\\
      &= \phi \circ \Bigl( \sum_{\xi \in Q(\csab, \csac, X)} \Xi_\xi \circ \xi_x \Bigr) \circ \psi\\
      &= \phi \circ \mu_{\csaa, \csab, \csac, X}(\Xi)_x \circ \psi\\
      &= Q(\phi, \psi, X) (\mu_{\csaa, \csab, \csac, X}(\Xi))_x
    \end{align*}
  \item Dinaturality in \(\csab\). Suppose \(\phi \in \wstarcat(\csab', \csab)\). Dinaturality is the commutativity of the following square:
    \[
      \begin{tikzcd}
        & {Q(\csaa, \csab',Q(\csab',\csac, X))} & \\
        {Q(\csaa, \csab, Q(\csab', \csac, X))} && {Q(\csaa, \csac, X)} \\
        & {Q(\csaa, \csab, Q(\csab, \csac, X))}
        \arrow["{\mu_{\csaa, \csab', \csac, X}}", from=1-2, to=2-3]
        \arrow["{Q(\csaa, \phi, Q(\csab', \csac, X))}", from=2-1, to=1-2]
        \arrow["{Q(\csaa, \csab, Q(\phi, \csac, X))}"', from=2-1, to=3-2]
        \arrow["{\mu_{\csaa, \csab, \csac, X}}"', from=3-2, to=2-3]
      \end{tikzcd}
    \]
    With \(\Xi \in Q(\csaa, \csab', Q(\csab', \csac, X))\) and \(x \in X\):
    \begin{align*}
      \mu_{\csaa, \csab, \csac, X}(Q(\csaa, \csab, Q(\phi, \csac, X))(\Xi))_x
      &= \sum_{\xi \in Q(\csab, \csac, X)} Q(\csaa, \csab, Q(\phi, \csac, X))(\Xi)_\xi \circ \xi_x\\
      &= \sum_{\xi \in Q(\csab, \csac, X)} \Bigl( \sum_{\xi' \in Q(\phi, \csac, X)^{-1}(\xi)} \Xi_{\xi'} \Bigr) \circ \xi_x\\
      &= \sum_{\substack{\xi \in Q(\csab, \csac, X)\\\xi' \in Q(\phi, \csac, X)^{-1}(\xi)}} \Xi_{\xi'} \circ \xi_x\\
      &= \sum_{\xi' \in Q(\csab', \csac, X)} \Xi_{\xi'} \circ \phi \circ \xi'_x\\
      &= \sum_{\xi' \in Q(\csab', \csac, X)} Q(\csaa, \phi, Q(\csab', \csac, X))(\Xi)_{\xi'} \circ \xi'_x\\
      &= \mu_{\csaa, \csab',\csac, X}(Q(\csaa, \phi, Q(\csab', \csac, X))(\Xi))_x
    \end{align*}
  \end{itemize}
  Hence \(\mu\) is natural.
\end{proof}

\begin{lemma}
  The strength \(\tau\) is a natural transformation.
\end{lemma}
\begin{proof}
  We first note that \(\Supp(\tau_{X,\csaa,\csab,Y}(x, \xi)) = \{(x, y) \mid y \in \Supp(\xi)\}\), and so is finite. The sum of components of \(\tau_{X,\csaa,\csab,Y}(x, \xi)\) is equal to the sum of components of \(\xi\) so is subunital. Hence, \(\tau_{X, \csaa, \csab, Y}(x, \xi)\) is well-defined.

  Now suppose \(\phi \in \wstarcat(\csaa, \csaa')\), \(\psi \in \wstarcat(\csab', \csab)\), and \(g \in \Set(Y, Y')\). Then:
  \begin{align*}
    Q(\phi, \psi, \id_X \times g) (\tau_{X, \csaa, \csab, Y'}(x, \xi))_{(x', y)}
    &= \sum_{y' \in g^{-1}(y)} \phi \circ \tau_{X, \csaa, \csab, Y'}(x, \xi)_{(x', y')} \circ \psi\\
    &= \sum_{y' \in g^{-1}(y)} \phi \circ \delta(x, \xi_{y'})_{x'} \circ \psi\\
    &= \delta(x, \sum_{y' \in g^{-1}(y)} \phi \circ \xi_{y'} \circ \psi)_{x'}\\
    &= \delta(x, Q(\phi, \psi, g)(\xi)_{y})_{x'}\\
    &= \tau_{X, \csaa', \csab', Y}(x, Q(\phi, \psi, g)(\xi))_{(x', y)}
  \end{align*}
  and for \(f \in \Set(X, X')\):
  \begin{align*}
    Q(\csaa, \csab, f \times \id_Y)(\tau_{X, \csaa, \csab, Y}(x, \xi))_{(x', y)}
    &= \sum_{x'' \in f^{-1}(x')} \tau_{X, \csaa, \csab, Y}(x, \xi)_{(x'', y)} \\
    &= \sum_{x'' \in f^{-1}(x')} \delta(x, \xi_y)_{x''}\\
    &= Q(\csaa, \csab, f)(\delta(x, \xi_y))_{x'}\\
    &= \delta(f(x), \xi_y)_{x'}\\
    &=\tau_{X', \csaa, \csab, Y}(f(x), \xi)_{(x', y)}
  \end{align*}
  where the second to last line is by \cref{lem:dirac-well-defined} and so \(\tau\) is a natural transformation.
\end{proof}

\begin{theorem}
  \((Q, \eta, \mu, \tau)\) is a \(\wstarcat^{\text{op}}\)-parameterised monad.
\end{theorem}
\begin{proof}
  We begin with the monad laws for unitality. We must have \(\mu_{\csaa, \csab, \csab, X} \circ Q(\csaa, \csab, \eta_{\csaa X}) = \id_{Q(\csaa, \csab, X)}\). Taking \(\xi \in Q(\csaa, \csab, X)\) and \(x \in X\):
  \begin{align*}
    \mu_{\csaa, \csab, \csab, X}(Q(\csaa, \csab, \eta_{\csaa X})(\xi))_x
    &= \sum_{\xi' \in Q(\csab, \csab, X)} Q(\csaa, \csab, \eta_{\csaa X})(\xi)_{\xi'} \circ \xi'_x\\
    &= \sum_{\substack{\xi' \in Q(\csab, \csab, X)\\x' \in \eta_{\csaa X}^{-1}(\xi')}} \xi_{x'} \circ \xi'_x\\
    &= \sum_{x' \in X} \xi_{x'} \circ \eta_{\csaa X}(x')_x\\
    &= \xi_x
  \end{align*}
  Similarly we must have \(\mu_{\csaa, \csaa, \csab, X} \circ \eta_{\csaa Q(\csaa, \csab, X)} = \id_{Q(\csaa, \csab, X)}\). Again for all \(\xi\) and \(x\) we have:
  \begin{align*}
    \mu_{\csaa, \csaa, \csab, X}(\eta_{\csaa Q(\csaa, \csab, X)}(\xi))_x
    &= \sum_{\xi' \in Q(\csaa, \csab, X)} \eta_{\csaa Q(\csaa, \csab, X)}(\xi)_{\xi'} \circ \xi'_x\\
    &= \sum_{\xi' \in Q(\csab, \csab, X)}
      \begin{cases}
        \id_{\csaa} \circ \xi'_x&\text{if }\xi = \xi'\\
        0 &\text{otherwise}
      \end{cases}\\
    &= \xi_x
  \end{align*}

For the associativity monad law, we need the following diagram to commute:
\[
  \begin{tikzcd}[column sep = 8em]
    {Q(\csaa, \csab, Q(\csab, \csac, Q(\csac, \csad, X)))} & {Q(\csaa, \csac, Q(\csac, \csad, X))} \\
    {Q(\csaa, \csab, Q(\csab, \csad, X))} & {Q(\csaa, \csad, X)}
    \arrow["{\mu_{\csaa, \csab, \csac, Q(\csac, \csad, X)}}", from=1-1, to=1-2]
    \arrow["{Q(\csaa, \csab, \mu_{\csab, \csac, \csad, X})}"', from=1-1, to=2-1]
    \arrow["{\mu_{\csaa, \csac, \csad, X}}", from=1-2, to=2-2]
    \arrow["{\mu_{\csaa, \csab, \csad, X}}", from=2-1, to=2-2]
  \end{tikzcd}
\]
Take \(\aleph \in Q(\csaa, \csab, Q(\csab, \csac, Q(\csac, \csad, X)))\) and \(x \in X\). Then evaluating the bottom left of the square gives:
\begin{align*}
  \mu_{\csaa, \csab, \csad, X}(Q(\csaa, \csab, \mu_{\csab, \csac, \csad, X})(\aleph))_x
  &= \sum_{\xi \in Q(\csab, \csad, X)} Q(\csaa, \csab, \mu_{\csab, \csac, \csad, X})(\aleph)_\xi \circ \xi_x\\
  &= \sum_{\xi \in Q(\csab, \csad, X)} \sum_{\Xi \in \mu_{\csab, \csac, \csad, X}^{-1}(\xi)} \aleph_\Xi \circ \xi_x\\
  &= \sum_{\substack{\xi \in Q(\csab, \csad, X)\\\Xi \in \mu_{\csab, \csac, \csad, X}^{-1}(\xi)}} \aleph_\Xi \circ \xi_x\\
  &= \sum_{\Xi \in Q(\csab, \csac, Q(\csac, \csad, X))} \aleph_\Xi \circ \mu_{\csab, \csac, \csad, X}(\Xi)_x\\
  &= \sum_{\Xi \in Q(\csab, \csac, Q(\csac, \csad, X))} \aleph_\Xi \circ \Bigl( \sum_{\xi \in Q(X)} \Xi_\xi \circ \xi_x \Bigr)\\
  &= \sum_{\substack{\Xi \in Q(\csab, \csac, Q(\csac, \csad, X)) \\ \xi \in Q(\csac, \csad, X)}} \aleph_\Xi \circ \Xi_\xi \circ \xi_x\\
  \intertext{and evaluating the top right gives:}
  \mu_{\csaa, \csac, \csad, X}(\mu_{\csaa, \csab, \csac, Q(\csac, \csad, X)}(\aleph))_x &= \sum_{\xi \in Q(\csac, \csad, X)} \mu_{\csaa, \csab, \csac, Q(\csac, \csad, X)}(\aleph)_\xi \circ \xi_x\\
  &= \sum_{\xi \in Q(\csac, \csad, X)} \Bigl( \sum_{\Xi \in Q(\csaa, \csab, Q(\csab, \csac, X))} \aleph_\Xi \circ \Xi_\xi \Bigr) \circ \xi_x\\
  &= \sum_{\substack{\Xi \in Q(\csab, \csac, Q(\csac, \csad, X)) \\ \xi \in Q(\csac, \csad, X)}} \aleph_\Xi \circ \Xi_\xi \circ \xi_x
\end{align*}
and so the square commutes as required.

For strength, we first show that it interacts well with \(1 = \{*\}\):
  \[ Q(\csaa, \csab, \lambda_X) \circ \tau_{1, \csaa, \csab, X} = \lambda_{Q(\csaa, \csab, X)} \]
  Where \(\lambda_X : 1 \times X \to X\) is the left unitor isomorphism. For \(\xi \in Q(\csaa, \csab, X)\) and \(x \in X\):
  \begin{align*}
    Q(\csaa, \csab, \lambda_X)(\tau_{1, \csaa, \csab, X}(*, \xi))_x &= \sum_{x \in \lambda_X^{-1}(x)} \tau_{1, \csaa, \csab, X}(*, \xi)_x\\
                                                                    &= \tau_{1, \csaa, \csab, X}(*, \xi)_{(*, x)}\\
                                                                    &= \xi_x \\
                                                                    &= \lambda_{Q(\csaa, \csab, X)}(*, \xi)_x
  \end{align*}

  For compatibility with the associator isomorphism \(\alpha_{X,Y,Z} : (X \times Y) \times Z \to X \times (Y \times Z)\), we need the following to commute:
  \[
    \begin{tikzcd}
      {(X \times Y)  \times Q(\csaa, \csab, Z)} && {Q(\csaa, \csab, (X \times Y) \times Z)} \\
      {X \times (Y \times Q(\csaa, \csab, Z))} & {X \times Q(\csaa, \csab, Y \times Z)} & {Q(\csaa, \csab, X \times (Y \times Z))}
      \arrow["{\tau_{X \times Y, \csaa, \csab, Z}}", from=1-1, to=1-3]
      \arrow["{\alpha_{X, Y, Q(\csaa, \csab, Z)}}"', from=1-1, to=2-1]
      \arrow["{Q(\csaa, \csab, \alpha_{X,Y,Z})}", from=1-3, to=2-3]
      \arrow["{\id_X \times \tau_{Y, \csaa, \csab, Z}}", far, from=2-1, to=2-2]
      \arrow["{\tau_{X, \csaa, \csab, Y \times Z}}", far, from=2-2, to=2-3]
    \end{tikzcd}
  \]
  With \(\xi \in Q(\csaa, \csab, Z)\), \(x, x' \in X\), \(y, y' \in Y\), and \(z \in Z\):
  \begin{align*}
    &\phantom{{}={}}\tau_{X, \csaa, \csab, Y \times Z} (\id_X \times \tau_{Y, \csaa, \csab, Z} (\alpha_{X, Y, Q(\csaa, \csab, Z)}((x', y'), \xi)))_{(x,(y,z))} \\
    &= \tau_{X, \csaa, \csab, Y \times Z} (x', \tau_{Y, \csaa, \csab, Z}(y', \xi))_{(x, (y, z))}\\
    &=
      \begin{cases}
        \tau_{Y, \csaa, \csab, Z}(y', \xi)_{(y, z)}&\text{if }x = x'\\
        0 &\text{otherwise}
      \end{cases}\\
    &=
      \begin{cases}
        \xi_z &\text{if }x = x'\text{ and }y = y'\\
        0&\text{otherwise}
      \end{cases}\\
    &= \tau_{X \times Y, \csaa, \csab, Z}((x', y'), \xi)_{((x, y), z)}\\
    &= Q(\csaa, \csab, \alpha_{X, Y, Z})(\tau_{X \times Y, \csaa, \csab, Z}((x', y'), \xi))_{(x, (y, z))}
  \end{align*}
  For compatibility with the unit \(\eta\):
  \[ \tau_{X, \csaa, \csaa, Y} \circ (\id_X \times \eta_{\csaa Y}) = \eta_{\csaa X \times Y} \]
  Taking \(x, x' \in X\), and \(y, y' \in Y\):
  \begin{align*}
    \tau_{X, \csaa, \csaa, Y}(x, \eta_{\csaa Y}(y))_{(x', y')} &=
                                                                         \begin{cases}
                                                                           \eta_{\csaa, \csab, Y}(y)_{y'}&\text{if }x = x'\\
                                                                           0&\text{otherwise}
                                                                         \end{cases}\\
    &=
      \begin{cases}
        \id_{\csaa} &\text{if }x = x'\text{ and }y = y'\\
        0&\text{otherwise}
      \end{cases}\\
    &= \eta_{\csaa X \times Y}(x, y)_{(x', y')}
  \end{align*}
  Finally for compatibility with \(\mu\) the following must commute:
  \[
    \begin{tikzcd}
      {X \times Q(\csaa, \csab, Q(\csab, \csac, Y))} & {Q(\csaa, \csab, X \times Q(\csab, \csac, Y))} & {Q(\csaa, \csab, Q(\csab, \csac, X \times Y))} \\
      {X \times Q(\csaa, \csac, Y)} && {Q(\csaa, \csac, X \times Y)}
      \arrow["{\tau_{X, \csaa, \csab, Q(\csab, \csac, Y)}}", far, from=1-1, to=1-2]
      \arrow["{\id_X \times \mu_{\csaa, \csab, \csac, Y}}", from=1-1, to=2-1]
      \arrow["{Q(\csaa, \csab, \tau_{X, \csab, \csac, Y})}", far, from=1-2, to=1-3]
      \arrow["{\mu_{\csaa, \csab, \csac, X \times Y}}", from=1-3, to=2-3]
      \arrow["{\tau_{X, \csaa, \csac, Y}}", from=2-1, to=2-3]
    \end{tikzcd}
  \]
  With \(x', x \in X\), \(y \in Y\), and \(\Xi \in Q(\csaa, \csab, Q(\csab, \csac, Y))\):
  \begin{align*}
    &\phantom{{}={}}\mu_{\csaa, \csab, \csac, X \times Y}(Q(\csaa, \csab, \tau_{X, \csab, \csac, Y})(\tau_{X, \csaa, \csab, Q(\csab, \csac, Y)}(x', \Xi)))_{(x, y)}\\
    &= \sum_{\xi \in Q(\csab, \csac, X \times Y)} Q(\csaa, \csab, \tau_{X, \csab, \csac, Y})(\tau_{X, \csaa, \csab, Q(\csab, \csac, Y)}(x', \Xi))_\xi \circ \xi_{(x, y)}\\
    &= \sum_{\substack{\xi \in Q(\csab, \csac, X \times Y)\\(x'', \xi') \in \tau_{X, \csab, \csac, Y}^{-1}(\xi)}} \tau_{X, \csaa, \csab, Q(\csab, \csac, Y)}(x', \Xi)_{(x'', \xi')} \circ \xi_{(x, y)}\\
    &= \sum_{\substack{x'' \in X\\ \xi' \in Q(\csab, \csac, Y)}} \tau_{X, \csaa, \csab, Q(\csab, \csac, Y)}(x', \Xi)_{(x'', \xi')} \circ \tau_{X, \csab, \csac, Y}(x'', \xi')_{(x, y)}\\
    &= \sum_{\xi' \in Q(\csab, \csac, Y)} \tau_{X, \csaa, \csab, Q(\csab, \csac, Y)}(x', \Xi)_{(x, \xi')} \circ \xi'_y\\
    &=
      \begin{cases}
        \sum_{\xi' \in Q(\csab, \csac, Y)} \Xi_{\xi'} \circ \xi'_y&\text{if }x = x'\\
        0 &\text{otherwise}
      \end{cases}\\
    &=
      \begin{cases}
        \mu_{\csaa, \csab, \csac, Y}(\Xi)_y&\text{if }x = x'\\
        0 &\text{otherwise}
      \end{cases}\\
    &= \tau_{X, \csaa, \csac, Y}(x', \mu_{\csaa, \csab, \csac, Y}(\Xi))_{(x,y)}
  \end{align*}
  Therefore, \((Q, \eta, \mu, \tau)\) is a parameterised strong monad.
\end{proof}

\section{Finite Quantum Orchestras}
\label{sec:finite-quant-orch}

The collection of quantum orchestras introduced above is large, and it can be unwieldy to perform calculations or define operations on it, motivating the search for well-behaved subsets of this structure. We have already seen the Dirac quantum orchestra, which corresponds to a quantum operations which always returns the same classical value. A natural extension is to consider the quantum orchestras who have a finite number of possible outputs. We will see that these correspond closely to quantum instruments.

\begin{definition}[Finite Quantum Orchestra]
  \label{def:fin-quantum-orchestra}
  Recall that \(\delta(x, \phi)\) is a Dirac quantum orchestra, and let addition of quantum orchestras be given pointwise (when this is well defined). Then, a quantum orchestra is \emph{finite} if it takes the form:
  \[ \sum_i \delta(x_i, \phi_i)\]
  where \(x_0, \dots, x_n \in X\) and \(\phi_0, \dots, \phi_n \in \wstarcat(\csab, \csaa)\) such that:
  \[ \phi_0(1) + \dots + \phi_n(1) \leq 1\]
  that is \(\sum\phi_i\) is subunital. We write \(\orch_{\mathsf{f}}(\csaa, \csab, X)\) for the subset of \(\orch(\csaa, \csab, X)\) consiting of finite quantum orchestras.
\end{definition}

The condition that \(\sum_i\phi_i\) is subunital is clear necessary for being a quantum orchestra. Consider \(k\) which maps everything to the identity channel. Then:
\[ \sum_i \delta(x_i, \phi_i)_k = \sum_i \phi_i \circ k(x_i) = \sum_i \phi_i\]
and so the right hand side must be subunital. It in fact turns out that this condition is sufficient to be a quantum orchestra, meaning all finite quantum orchestras take this form.

\begin{proposition}
  Let \(x_0, \dots, x_n\) and \(\phi_0, \dots, \phi_n\) be as in \cref{def:fin-quantum-orchestra}. Then,
  \[ \sum_i \delta(x_i, \phi_i) \in \orch_{\mathsf{f}}(\csaa, \csab, X)\]
  and so the finite quantum orchestras are exactly the sums of this form.
\end{proposition}
\begin{proof}
  We must first show that the sum sends \(k \in \dcpo(X, \wstarcat(\csak, \csab))\) to an element of \(\wstarcat(\csak, \csaa)\). Applying it to the unit \(1\) we get:
  \begin{align*}
    \sum_i \delta(x_i, \phi_i)_k(1)
    &= \left(\sum_i \phi_i \circ k(x_i)\right)(1)\\
    &= \sum_i \phi_i (k(x_i)(1))\\
    &\leq \sum_i \phi_i(1)&\text{as each }k(x_i)\text{ is subunital}\\
    &\leq 1&\text{by assumption on the }\phi_i
  \end{align*}
  Continuity is clear from the continuity of the Dirac quantum orchestra and the continuity of addition. For subconvexity, let \(\lambda, \rho \geq 0\) with \(\lambda + \rho \leq 1\). Then for all \(k\) and \(k'\):
  \[ \sum_i \delta(x_i, \phi_i)_{\lambda k + \rho k'} = \sum_i \lambda\delta(x_i, \phi_i)_k + \rho \delta(x_i, \phi_i)_{k'} = \lambda \sum_i \delta(x_i, \phi_i)_k + \rho \sum_i \delta(x_i, \phi_i)_{k'} \]
  Lastly, for compositionality, suppose \(\chi \in \wstarcat(\csak', \csak)\). Then:
  \[\sum_i \delta(x_i, \phi_i)_{k(\_) \circ \chi} = \sum_i \delta(x_i, \phi_i) \circ \chi = \left( \sum_i \delta(x_i, \phi_i) \right) \circ \chi \]
  and so the sum is a quantum orchestra.
\end{proof}

We now investigate the action of the monad operations on finite quantum instruments. We first consider the action of the functor. Let
\[ \sum_i \delta(x_i, \phi_i)\]
be a finite quantum orchestra. Suppose \(\phi \in \wstarcat(\csaa, \csaa')\), \(\psi \in \wstarcat(\csab', \csab)\) and \(f : X \to Y\). Then:
\begin{align*}
  \orch(\phi, \psi, f)\left( \sum_i \delta(x_i, \phi_i) \right)_k
  &= \phi \circ \sum_i \delta(x_i, \phi_i)_{\psi \circ k(f(\_))}\\
  &= \sum_i \phi \circ \phi_i \circ \psi \circ k(f(x_i))\\
  &= \sum_i \delta(f(x_i), \phi \circ \phi_i \circ \psi)_k
\end{align*}
and so \(\orch(\phi, \psi, f)\) preserves finiteness.

By definition, the unit of the monad always produces a finite quantum orchestra. The action of the strength of the monad is also clear. Suppose \(x \in X\) and that:
\[ \sum_i \delta(y_i, \phi_i)\]
is a finite quantum orchestras in \(\orch(\csaa, \csab, Y)\). Then:
\[ \tau_{X, \csaa, \csab, Y}\left(x, \sum_i \delta(y_i, \phi_i)\right)_k = \sum_i \delta(y_i, \phi_i)_{k(x, \_)} = \sum_i \phi_i \circ k(x, y_i) = \sum_i \delta(x, y_i)_k\]

We now move our attention to the multiplication. Let \(I\) be finite and \(\xi_i\) be a finite quantum orchestras for \(i \in I\), each given by:
\[ \xi_i = \sum_{j \in J_i} \delta(x_{ij}, \psi_{ij}) \]
where \(x \in X\) and \(\psi_{ij} \in \wstarcat(\csac, \csab)\). Further suppose:
\[ \sum_{i \in I} \delta(\xi_i, \phi_i)\]
is a finite quantum orchestra. Then we have:
\begin{align*}
  \mu_{\csaa, \csab, \csac, X}\left( \sum_i \delta(\xi_i, \phi_i) \right)_k
  &= \sum_i \delta(\xi_i, \phi_i)_{\ev_k}\\
  &= \sum_i \phi_i \circ \ev_k(\xi_i)\\
  &= \sum_i \phi_i \circ \sum_{j \in J_i} \delta(x_{ij}, \psi_{ij})_k\\
  &= \sum_i \sum_{j \in J_i} \phi_i \circ  \psi_{ij} \circ k(x_{ij})
\end{align*}
which produces a finite quantum orchestra (as the disjoint union of the \(J_i\) is finite).

Although we know the form that finite quantum orchestras take, it can be hard to utilise this to define operations on them, as their decompositions may not be unique. However, functions can be defined this way if they satisfy certain conditions.

\begin{proposition}
  \label{prop:ap-functor}
  Let \(F : \wstarcat \to \wstarcat\) be a functor such that the maps on morphisms are additive. Then, there are functions \(F^\dagger_{\csaa, \csab, X} : \orch_{\mathsf{f}}(\csaa, \csab, X) \to \orch_{\mathsf{f}}(F(\csaa), F(\csab), X)\) given by:
  \[ \sum_i \delta(x_i, \phi_i) \mapsto \sum_i \delta(x_i, F(\phi_i))\]
  These maps form a natural transformation.
\end{proposition}
\begin{proof}
  Suppose \(\xi \in \orch_{\mathsf{f}}(\csaa, \csab, X)\) and that there are decompositions:
  \[ \sum_{i \in I} \delta(x_i, \phi_i) = \xi = \sum_{j \in J} \delta(x_j', \phi_j')\]
  Then we must show that:
  \[ \sum_{i \in I} \delta(x_i, F(\phi_i)) = \sum_{j \in J} \delta(x_j', F(\phi_j'))\]
  Now let \(S\) be the (finite) set \(\{x_i \mid i \in I\} \cup \{x_j' \mid j \in J\}\). Then:
  \begin{align*}
    \sum_{i \in I} \delta(x_i, F(\phi_i))
    &= \sum_{x \in S} \sum_{\substack{i \in I\\x_i = x}} \delta(x, F(\phi_i))\\
    &= \sum_{x \in S} \delta\bigg(x, \sum_{\substack{i \in I\\x_i = x} }F(\phi_i)\bigg)\\
    &= \sum_{x \in S} \delta\bigg(x, F\bigg(\sum_{\substack{i \in I\\x_i = x}}\phi_i\bigg)\bigg)
  \end{align*}
  Similarly we have:
  \[ \sum_{j \in J} \delta(x_j', F(\phi_j')) = \sum_{x \in S} \delta\bigg(x, F\bigg(\sum_{\substack{j \in J\\x_j'=x}}\phi_j'\bigg)\bigg)\]
  and so it is sufficient to show that for all \(x \in S\) that:
  \[ \sum_{\substack{i \in I\\x_i = x}}\phi_i = \sum_{\substack{j \in J\\x_j'=x}}\phi_j' \]
  Suppose for contradiction that this doesn't hold, and let \(x \in S\) be maximal such that this equality doesn't hold. Then consider the set:
  \[ O = \bigcap_{\substack{s \in S\\x \not\leq s}} s_{\not\leq} \quad \text{where} \quad s_{\not\leq} = \{y \in X \mid y \not\leq s\}\]
  \(O\) is a Scott open subset of \(X\): it is a finite intersection of sets of the form \(s_{\not\leq}\) which are all upwards closed, as if \(y \not\leq s\) and \(y' \geq y\) then certainly \(y' \not \leq s\), and inaccessible from below, as if \(D\) is directed with \(D \cap s_{\not\leq}\) then \(s\) is an upper bound for \(D\) and so \(\sup D \leq s\), so is not in \(s_{\not\leq}\). Further \(y \in O\) if and only if \(x \leq y\): if \(x \not \leq y\) then \(y_{\not\leq} \supseteq O\) and so \(y \not\in O\). By construction \(x \in O\) as \(x \in s_{\not\leq}\) when \(x \not\leq s\), and so \(y \in O\) for \(y \geq x\) as \(O\) is upwards closed.

  Then let \(k = \mathbbm{1}_{O}\), the indicator function on the set \(O\), which is continuous as \(O\) is open. Then:
  \[\xi_k = \sum_i \delta(x_i, \phi_i)_k = \sum_i \phi_i \circ k(x_i) = \sum_{\substack{i \in I\\x_i \geq x}} \phi_i\]
  and similarly:
  \[\xi_k = \sum_{\substack{j \in J\\x_j' \geq x}} \phi_j'\]
  However, by maximality of \(x\) we have that:
  \[\sum_{\substack{i \in I\\x_i > x}} \phi_i = \sum_{\substack{j \in J\\x_j' \geq x}} \phi_j'\]
  and so:
  \[\sum_{\substack{i \in I\\x_i = x}} \phi_i = \sum_{\substack{j \in J\\x_j' = x}} \phi_j'\]
  contradicting the construction of \(x\).

  To show \(F^\dagger\) is a natural transformation we need the following to commute:
  \[
    \begin{tikzcd}
      {\orch_{\mathsf{f}}(\csaa, \csab, X)} & {\orch_{\mathsf{f}}(F(\csaa), F(\csab), X)} \\
      {\orch_{\mathsf{f}}(\csaa', \csab', Y)} & {\orch_{\mathsf{f}}(F(\csaa'), F(\csab'), Y)}
      \arrow["{F^\dagger_{\csaa, \csab, X}}", from=1-1, to=1-2]
      \arrow["{\orch_{\mathsf{f}}(\phi, \psi,f)}"', from=1-1, to=2-1]
      \arrow["{\orch_{\mathsf{f}}(F(\phi),F(\psi), f)}", from=1-2, to=2-2]
      \arrow["{F^\dagger_{\csaa', \csab', Y}}", from=2-1, to=2-2]
    \end{tikzcd}
  \]
  For \(\phi \in \wstarcat(\csaa, \csaa')\), \(\psi \in \wstarcat(\csab', \csab)\), \(f : \dcpo(X, Y)\). Taking some decomposition:
  \[ \sum_i \delta(x_i, \chi_i) \in \orch_{\mathsf{f}}(\csaa, \csab, X)\]
  we get:
  \begin{align*}
    \orch_{\mathsf{f}}(F(\phi), F(\psi), f) \bigg(F^\dagger_{\csaa, \csab, X}\bigg(\sum_i \delta(x_i, \chi_i)\bigg)\bigg)
    &= \orch_{\mathsf{f}}(F(\phi), F(\psi), f) \bigg( \sum_i \delta(x_i, F(\chi_i)) \bigg)\\
    &= \sum_i \delta(f(x_i), F(\phi) \circ F(\chi_i) \circ F(\psi))\\
    &= \sum_i \delta(f(x_i), F(\phi \circ \chi_i \circ \psi))\\
    &= F^\dagger_{\csaa', \csab', Y}\bigg(\sum_i \delta(f(x_i), \phi \circ \chi_i \circ \psi)\bigg)\\
    &= F^\dagger_{\csaa', \csab', Y}\bigg(\orch_{\mathsf{f}}(\phi, \psi, f)\bigg( \sum_i \delta(x_i, \chi_i) \bigg) \bigg)
  \end{align*}
  And so \(F^\dagger\) is natural.
\end{proof}

Hence, all the operations of the monad preserve the finiteness of quantum orchestras. Of course, the set of finite quantum orchestras on a (non-finite) dcpo is not itself a dcpo, and so \(\orch_{\mathsf{f}}\) it not itself a monad. We can however study the order structure on this set.
\begin{lemma}
  \label{lem:dirac-order}
  If \(x \leq y\) and \(\Phi \leq \Psi\), then \(\delta(x, \Phi) \leq \delta(y, \Psi)\). Conversely, if \(\delta(x, \Phi) \leq \delta(y, \Psi)\) then \(\Phi \leq \Psi\) and if \(\Phi \neq 0\) then \(x \leq y\).
\end{lemma}

\begin{proof}
  Suppose \(x \leq y\) and \(\Phi \leq \Psi\). Then for all \(k\):
  \[ \delta(x, \Phi)_k = \Phi \circ k(x) \leq \Psi \circ k(y) = \delta(y, \Psi)\]
  as composition of channels is monotone. Conversely assume \(\delta(x, \Phi) \leq \delta(y, \Psi)\). Then let \(k\) be the continuation such that \(k(z) = \id\) for all \(z\). Then:
  \[ \Phi = \Phi \circ k(x) = \delta(x, \Phi) \leq \delta(y, \Psi) = \Psi \circ k(y) = \Psi\]
  Now consider \(k\) such that:
  \[ k(z) =
    \begin{cases*}
      \id&if \(z \not\leq y\)\\
      0&otherwise
    \end{cases*}
  \]
  Then \(\delta(y, \Psi)_k = 0\), hence \(\delta(x, \Phi) = 0\) and so either \(\Phi = 0\) or \(x \leq y\).
\end{proof}

Despite the relative simplicity of the order structure on Dirac quantum orchestras, the story for finite quantum orchestras is far less trivial. We conjecture that, for finite \(\xi\) and \(\xi'\) with \(\xi \leq \xi'\), one can realise this inequality as a finite sequence of inequalities from \cref{lem:dirac-order} on a single Dirac orchestra. A corresponding theorem is for classical valuations can be proved using the max-flow min-cut theorem (see \cite{jones:probabilisticpowerdomain}), but this theorem depends on the lattice structure of the real numbers, and does not hold for quantum channels. Despite the proof not generalising, we have found no counterexample to this statement, and it is unclear if it holds.
For a detailed discussion of the link between orchestras and valuations, see \cref{app:comparison-prob-powerdomain}.

\section{Comparison to the Probabilistic Powerdomain Monad}%
\label{app:comparison-prob-powerdomain}

The goal of this section is to clarify the relation of the quantum orchestra monad with the probabilistic powerdomain monad.
To this end, recall that the probabilistic powerdomain monad assigns to a dcpo \(X\) the set of valuations on \(X\), denoted \(\valu(X)\).
\begin{definition}[Scalar-valued valuations \cite{jones:probabilisticpowerdomain}]
	Let \(X\) be a topological space, and \(O(X)\) its complete lattice of open subsets. A function \( v : O(X) \to [0,1] \) is a \emph{(scalar-valued) valuation} if and only if it is
	\begin{itemize}
		\item monotone, \emph{i.e.}, \( v(U) \leq v(V) \) for all \( U \subseteq V \),
		\item strict, \emph{i.e.}, \( v(\emptyset) = 0 \),
		\item modular, \emph{i.e.}, \( v(U) + v(V) = v(U \cup V) + v(U \cap V) \),
		\item continuous, \emph{i.e.}, for any directed set \( \{ U_\lambda \}_{\lambda \in \Lambda} \) of open sets, it holds that
			\[ v \left( \bigcup_\lambda U_\lambda \right) = \sup_\lambda v \left( U_\lambda \right). \]
	\end{itemize}
\end{definition}
\begin{theorem}[\cite{jones:probabilisticpowerdomain}]
	Let \(X\) be a dcpo. Then, \(\valu(X)\) is the set of (scalar-valued) valuations on \(X\) with the Scott topology, ordered by:
	\[
		v \leq w \quad :\Leftrightarrow \quad v(U) \leq w(U) \;\forall U \in O(X).
	\]
	Also, \(\valu(X)\) is a pointed dcpo with directed suprema defined pointwise and the zero valuation as its least element.
	Finally, \(\valu\) is a monad on \(\dcpo\), called the \emph{probabilistic powerdomain} monad, with the unit \(\epsilon_X : X \to \valu(X)\) given by
	\[
		\epsilon_X (x)(U) =
		\begin{cases}
      		1, 	& \text{if } x \in U, \\
      		0, 	& \text{otherwise}.
    		\end{cases}
	\]
	and the Kleisli extension \( \kleisli{f} : \valu(X) \to \valu(Y)\) for continuous functions \( f : X \to \valu(Y) \) given by
	\[
		\kleisli{f}(v)(V) = \int_{x\in X} f(x)(V) dv.
	\]
\end{theorem}
Note that the notion of valuation can be generalised almost verbatim from scalar-valued to \(\wstarcat(\csab, \csaa)\)-valued maps as follows.
\begin{definition}[\(\wstarcat(\csab, \csaa)\)-valued valuations]
	Let \(\csaa,\csab \in \wstarcat\), \(X\) be a topological space, and \(O(X)\) its complete lattice of open subsets. A function \( v : O(X) \to \wstarcat(\csab, \csaa) \) is called a valuation if and only if it is
	\begin{itemize}
		\item monotone, \emph{i.e.}, \( v(U) \leq v(V) \) for all \( U \subseteq V \),
		\item strict, \emph{i.e.}, \( v(\emptyset) = 0 \),
		\item modular, \emph{i.e.}, \( v(U) + v(V) = v(U \cup V) + v(U \cap V) \),
		\item continuous, \emph{i.e.}, for any directed set \( \{ U_\lambda \}_{\lambda \in \Lambda} \) of open sets, it holds that
			\[ v \left( \bigcup_\lambda U_\lambda \right) = \sup_\lambda v \left( U_\lambda \right). \]
	\end{itemize}
\end{definition}
By identifying \(\wstarcat(\mathbb{C}, \mathbb{C})\) with the real interval \([0,1]\), \(\wstarcat(\mathbb{C}, \mathbb{C})\)-valued valuations are in fact scalar-valued valuations.

\subsection{Mapping orchestras to valuations}%
\label{sec:mapping-orch-to-val}

Let now \(X \in \dcpo\). We note that any quantum orchestra \(\xi \in \orch(\csaa, \csab, X)\) restricts to a \(\wstarcat(\csab, \csaa)\)-valued valuation \(v^\xi\) by taking
  \[ v^\xi (U) = \xi_{\mathbbm{1}_U}, \]
where \(\mathbbm{1}_U\) is the indicator function \(X \to \wstarcat(\csab, \csab)\):
  \[
    \mathbbm{1}_U(x) =
    \begin{cases}
      \id_{\csab},	&\text{if }x \in U, \\
      0, 			&\text{otherwise}.
    \end{cases}
  \]
The continuity of \(\mathbbm{1}_U\) corresponds directly to \(U\) being a Scott-open subset of \(X\).
By the subconvexity property of \(\xi\), we obtain
  \[ v^\xi(\emptyset) = \xi_{\mathbbm{1}_\emptyset} = \xi_{0} = 0 \]
and
\begin{align*}
    v^\xi(U \cup V) + v^\xi(U \cap V) &= \xi_{\mathbbm{1}_{U \cup V}} + \xi_{\mathbbm{1}_{U \cap V}} = 2 \left(\frac{1}{2} \xi_{\mathbbm{1}_{U \cup V}} + \frac{1}{2}\xi_{\mathbbm{1}_{U \cap V}}\right) = 2 \xi_{\frac{1}{2}\mathbbm{1}_{U \cup V} + \frac{1}{2}\mathbbm{1}_{U \cap V}} = 2 \xi_{\frac{1}{2}\mathbbm{1}_U + \frac{1}{2}\mathbbm{1}_V} \\
     								 &= 2 \left(\frac{1}{2} \xi_{\mathbbm{1}_U} + \frac{1}{2}\xi_{\mathbbm{1}_V}\right) = \xi_{\mathbbm{1}_U} + \xi_{\mathbbm{1}_V} = v^\xi(U) + v^\xi(V).
\end{align*}
By monotonicity we immediately have for \(U \subseteq V\) that
  \[ v^\xi(U) = \xi_{\mathbbm{1}_U} \leq \xi_{\mathbbm{1}_V} = v^\xi(V), \]
and by continuity we get that if \( \{ U_\lambda \}_{\lambda \in \Lambda} \) is a directed set of open subsets, then it also follows that
\begin{align*}
    \sup_\lambda v^\xi(U_\lambda) &= \sup_\lambda \xi_{\mathbbm{1}_{U_\lambda}} = \xi_{\sup_\lambda \mathbbm{1}_{U_\lambda}} = \xi_{\mathbbm{1}_{\bigcup_\lambda U_\lambda}} = v^\xi \left( \bigcup_\lambda U_\lambda \right),
\end{align*}
and so \(v^\xi\) is a valuation.
The pointwise order or quantum orchestras immediately implies for orchestras \(\xi,\zeta \in \orch(\csaa, \csab, X)\) that
\begin{align*}
	\xi \leq \zeta \quad\Rightarrow\quad \forall U \in O(X): \xi_{\mathbbm{1}_U} \leq \zeta_{\mathbbm{1}_U} \quad\Rightarrow\quad v^\xi \leq v^\zeta,
\end{align*}
and thus the mapping \( \xi \mapsto v^\xi \) is a monotone mapping from orchestras to valuations.

While we have obtained a monotone mapping from orchestras to valuations in this way, it is not immediately clear whether this mapping is injective in the general case.
The difficulty in showing injectivity lies in the absence of an existing non-commutative integration theory for \(\wstarcat(\csab, \csaa)\)-valued valuations.
We leave this question for future exploration, and subsequently focus on the scalar-valued case, in which a well-developed integration theory exists \cite{jones:probabilisticpowerdomain}.

\subsection{The scalar-valued case}%
\label{sec:scalar-valued-valuations}

Now consider the (non-parameterised) monad obtained as \(\orch(\mathbb{C},\mathbb{C},-)\). By the arguments from \cref{sec:mapping-orch-to-val}, every orchestra \(\xi \in \orch(\mathbb{C},\mathbb{C},-)\) gives rise to a scalar-valued valuation \(v^\xi : O(X) \to [0,1] \), where we identified the real interval \([0,1]\) with \(\wstarcat(\mathbb{C},\mathbb{C})\).
By use of the integration theory for scalar-valued valuations \cite{jones:probabilisticpowerdomain}, any such valuation gives rise to a continuous integral of any continuous function \(f : X \to [0,1] \):
\[
	f \mapsto \int f dv^\xi.
\]
By the definition of \(v^\xi\), this integral needs to coincide with the orchestra on indicator functions:
\[
	\int \mathbbm{1}_U dv^\xi = v^\xi (U) = \xi_{\mathbbm{1}_U},
\]
and by linearity also on simple functions, \emph{i.e.}, linear combinations of indicator functions. As both the integral and the orchestra are continuous, they preserve suprema, and the equality extends to all continuous functions, as all continuous functions can be represented as suprema of simple functions.
We thus record: for all continuous functions \(f : X \to [0,1] \), it holds that:
\[
	\int f dv^\xi = \xi_{f}.
\]
We now proceed to show that in the scalar-valued case, the mapping from orchestras to valuations is injective.
To this end, let \(\xi,\zeta \in \orch(\mathbb{C},\mathbb{C},-)\) with \(\xi \neq \zeta\). There then exists a \wstar-algebra \(\csak\) and a continuation \(k \in \dcpo(X,\wstarcat(\csak,\mathbb{C}))\) such that \(\xi_k \neq \zeta_k\), where \(\xi_k, \zeta_k \in \wstarcat(\csak, \mathbb{C})\).
There thus exists some \(m \in \wstarcat(\mathbb{C},\csak)\) such that \(\xi_k \circ m \neq \zeta_k \circ m\).
Define the continuous function
\[
	g : X \to [0,1] ,\quad x \mapsto k(x) \circ m.
\]
Using the compositionality of orchestras, we conclude that
\[
	\int g dv^\xi = \xi_g = \xi_k \circ m \neq \zeta_k \circ m = \zeta_g = \int g dv^\zeta,
\]
and thus \(\xi^v \neq \zeta^v\). The mapping \( \xi \mapsto v^\xi \) is hence injective.

Finally, let \(\Delta \subseteq \orch(\mathbb{C}, \mathbb{C}, X)\) be directed. Using that scalar-valued valuations form a dcpo \cite{jones:probabilisticpowerdomain}, it then follows for all \(U \in O(X)\) that
\begin{align*}
	\sup \left\{ v^\xi \mid \xi \in \Delta \right\} (U) &= \sup \left\{ v^\xi(U) \mid \xi \in \Delta \right\} = \sup \left\{ \xi_{\mathbbm{1}_U} \mid \xi \in \Delta \right\} \\
												  &= \left[ \sup \left\{ \xi \mid \xi \in \Delta \right\} \right]_{\mathbbm{1}_U} = v^{\left[ \sup \left\{ \xi \mid \xi \in \Delta \right\} \right]} (U).
\end{align*}
The mapping \( \xi \mapsto v^\xi \) is thus also preserving suprema and hence continuous.
We summarise these results in the following proposition.
\begin{proposition}
	For any \( X \in \dcpo\), there exists an injective, continuous map \( \alpha_X : \orch(\mathbb{C},\mathbb{C},X) \to \valu(X) ,\, \xi \mapsto v^\xi \).
\end{proposition}

\subsection{Mapping scalar-valued valuations to orchestras}%
\label{sec:mapping-valuations-to-orchestras}

Having established that there is an injective, continuous mapping from \(\orch(\mathbb{C}, \mathbb{C}, X)\) to the dcpo of valuations \(O(X) \to [0,1]\), we now ask the question whether this mapping is invertible.

To this end, let \(v : O(X) \to [0,1]\) be a valuation, let \(\csak \in \wstarcat\), and consider a continuation \(k \in \dcpo(X,\wstarcat(\csak,\mathbb{C}))\).
As a first step towards constructing an integral of \(k\) along \(v\), define
\[
	\xi^v_k : \csak^+ \to \mathbb{R}^+ ,\quad \kappa \mapsto \int k(\kappa) dv,
\]
using that \(k(\kappa)\) is always a positive real number for all \(\kappa \in \csak^+\).
The positive cone \(\csak^+\) linearly spans the full algebra \(\csak\) as every element can be decomposed into a linear combination of four positive elements.
Linearly extending to all of \(\csak\), we obtain a map \( \xi^v_k : \csak \to \mathbb{C} \).
The positivity of \( \xi^v_k \) is immediate, complete positivity is equivalent to positivity for functionals, and normality of \( \xi^v_k \) follows from the fact that normality on \wstar-algebras is equivalent to the preservation of suprema \cite{cho:wstarisdcpo} and the fact that both \(k\) and integration along valuations are Scott-continuous. We thus showed that \( \xi^v_k \in \wstarcat(\csak,\mathbb{C}) \).
As a candidate orchestra, we have obtained \( \xi^v = \{ \xi^v_k \in \wstarcat(\csak,\mathbb{C}) \mid \csak \in \wstarcat, k \in \dcpo(X,\wstarcat(\csak,\mathbb{C})) \} \).

It remains to show that \(\xi^v\) is continuous as a mapping \( k \mapsto \xi^v_k \), subconvex, and satisfies compositionality.
Subconvexity and compositionality are direct consequences of the linear construction of the candidate orchestra.
To prove continuity, let \( K \subseteq \dcpo(X,\wstarcat(\csak,\mathbb{C}))\) be directed. It then follows for \( \kappa \in \csak^+ \) that
\begin{align*}
	\xi^v_{\sup K} (\kappa) &= \int [\sup K](\kappa) dv = \int \sup \left\{ k(\kappa)
\mid k \in K \right\} dv = \sup \left\{ \left. \int k(\kappa) dv \,\right|\, k\in K \right\} \\
	                        &= \sup \left\{ \left. \xi^v_k (\kappa) \,\right|\, k\in K \right\} = \left[ \sup \left\{ \left. \xi^v_k \,\right|\, k\in K \right\} \right] (\kappa),
\end{align*}
and thus by linearity that \( \xi^v_{\sup K} = \sup \left\{ \xi^v_k \mid k\in K \right\} \).
We conclude that \( \xi^v \) is a quantum orchestra.
We have thus constructed a map \( \beta_X : \valu(X) \to \orch(\mathbb{C},\mathbb{C},X) ,\, v \mapsto \xi^v \).

By construction of \( \beta_X \), it holds that \( v = \alpha_X [ \beta_X [v] ] \) for all \( v \in \valu(X) \).
This implies that \( \alpha_X \) from \cref{sec:scalar-valued-valuations} must have been surjective, showing that it is in fact bijective, and \( \beta_X \) its inverse.

For the final goal to establish an isomorphism between scalar-valued orchestras and valuations on \(\dcpo\) it remains to show that this inverse is Scott-continuous.

In order to show monotonicity of the inverse, start with two comparable valuations \( v, w : O(X) \to [0,1] \) with \( v \leq w \).
Let \( \xi^v, \xi^w \) be the two arising orchestras constructed as above.
For any \(\csak \in \wstarcat\), any continuation \(k \in \dcpo(X,\wstarcat(\csak,\mathbb{C}))\), and \( \kappa \in \csak^+ \), it holds that
\begin{align*}
	\xi^v_k (\kappa) = \int k(\kappa) dv \leq \int k(\kappa) dw = \xi^w_k (\kappa),
\end{align*}
and thus \( \xi^v_k \leq \xi^w_k \) in the Löwner order for all \(k\), implying that \( \xi^v \leq \xi^w \).
The map \( \beta_X \) is thus monotone.

To prove the preservation of suprema, let \( K \subseteq \valu(X) \) be a directed set of valuations. It then holds for any \(\csak \in \wstarcat\), any continuation \(k \in \dcpo(X,\wstarcat(\csak,\mathbb{C}))\), and for any \( \kappa \in \csak^+ \) that
\begin{align*}
	\left[ \xi^{\sup K} \right]_k (\kappa)
	&= \int k(\kappa) d[\sup K]
	= \sup \left\{\left. \int k(\kappa) dv \,\right|\, v \in K \right\} \\
	&= \sup \left\{\left. \xi^v_k (\kappa) \,\right|\, v \in K \right\}
	= \left[ \sup \left\{\left. \xi^v_k  \,\right|\, v \in K \right\} \right] (\kappa)
	= \left[ \sup \left\{\left. \xi^v  \,\right|\, v \in K \right\} \right]_k (\kappa),
\end{align*}
using the continuity of the integral in the valuation.
It follows that \( \xi^{\sup K} = \left[ \sup \left\{\left. \xi^v  \,\right|\, v \in K \right\} \right] \), and thus \( \beta_X \) is continuous.
This fact implies the next proposition.
\begin{proposition}
  For any \( X \in \dcpo\), \( \alpha_X \) is a \(\dcpo\)-isomorphism \( \orch(\mathbb{C},\mathbb{C},X) \to \valu(X) \).
\end{proposition}

\subsection{Monad equivalence}
\label{sec:mapp-scal-valu}

We now aim to prove that \(\orch(\mathbb{C}, \mathbb{C},-)\) and \(\valu\) are equivalent. We firstly show that \( \alpha_X \) forms in fact a natural isomorphism between the functors \(\orch(\mathbb{C},\mathbb{C},-)\) and \(\valu\). To this end, we must show that the following diagram commutes for all morphisms \( f \in \dcpo(X,Y) \).
\[
  \begin{tikzcd}
    {\orch(\mathbb{C}, \mathbb{C}, X)} & {\valu(X)} \\
    {\orch(\mathbb{C}, \mathbb{C}, Y)} & {\valu(Y)}
    \arrow["{\alpha_X}", from=1-1, to=1-2]
    \arrow["{\orch(\mathbb{C}, \mathbb{C}, f)}"', from=1-1, to=2-1]
    \arrow["{\valu(f)}", from=1-2, to=2-2]
    \arrow["{\alpha_Y}", from=2-1, to=2-2]
  \end{tikzcd}
\]
Let now \( \xi \in \orch(\mathbb{C},\mathbb{C},X) \) and \( V \in O(Y) \). Then, by definition, it holds that
\begin{align*}
	\left[ \alpha_Y \circ \orch(\mathbb{C},\mathbb{C}, f)(\xi) \right] (V)
	&= \orch(\mathbb{C},\mathbb{C}, f)(\xi)_{\mathbbm{1}_V}
	= \xi_{\lambda x.\mathbbm{1}_V(f(x))}
	= \xi_{\mathbbm{1}_{f^{-1}(V)}} \\
	&= \alpha_X (\xi) (f^{-1}(V))
	= \left[ \mathcal{V}(f) \circ \alpha_X (\xi) \right] (V),
\end{align*}
which shows that the diagram indeed commutes, proving the naturality of \( \alpha_X \) and the following claim.
\begin{proposition}
	\( \orch(\mathbb{C},\mathbb{C},-) \) and \( \valu \) are functorially equivalent on \( \dcpo \).
\end{proposition}
We proceed with proving the compatibility of the natural transformation \( \alpha_X \) with the units of the two monads.
To this end, we need to show that the diagram
\[
  \begin{tikzcd}
    X & {\orch(\mathbb{C}, \mathbb{C}, X)} \\
    & {\valu(X)}
    \arrow["{\eta_{\mathbb{C} X}}", from=1-1, to=1-2]
    \arrow["{\epsilon_X}"', from=1-1, to=2-2]
    \arrow["{\alpha_X}", from=1-2, to=2-2]
  \end{tikzcd}
\]
commutes.
Let \(x \in X\) for some \(X \in \dcpo\), and \( U \in O(X) \). Then it holds that
\begin{align*}
  \left[ \alpha_X \circ \eta_{\mathbb{C} X} (x) \right] (U)
  = \left[ \eta_{\mathbb{C} X} (x) \right]_{\mathbbm{1}_U}
  = \mathbbm{1}_U (x)
  = \epsilon_X(x)(U),
\end{align*}
and thus \(\alpha_X\) is compatible with the units.

To further show that \( \alpha_X \) is compatible with the multiplications of the two monads, we need to prove the commutativity of the following diagram:
\[
  \begin{tikzcd}
    &&& {\orch(\mathbb{C}, \mathbb{C}, \valu(X))} & \\
    {\orch(\mathbb{C}, \mathbb{C}, \orch(\mathbb{C}, \mathbb{C}, X))} && {\orch(\mathbb{C}, \mathbb{C}, X)} & {\valu(X)} & {\valu(\valu(X))} \\
    &&& {\valu(\orch(\mathbb{C},\mathbb{C},X))}
    \arrow["{\alpha_{\valu(X)}}", from=1-4, to=2-5]
    \arrow["{\orch(\mathbb{C},\mathbb{C},\alpha_X)}", from=2-1, to=1-4]
    \arrow["{\mu_x}"', from=2-1, to=2-3]
    \arrow["{\alpha_{\orch(\mathbb{C},\mathbb{C},X)}}"', from=2-1, to=3-4]
    \arrow["{\alpha_X}"', from=2-3, to=2-4]
    \arrow["{\nu_X}", from=2-5, to=2-4]
    \arrow["{\valu(\alpha_X)}"', from=3-4, to=2-5]
  \end{tikzcd}
\]
To this end, let \(\chi \in \orch(\mathbb{C},\mathbb{C},\orch(\mathbb{C},\mathbb{C},X))\) and \(U \in O(X)\). It then holds
\begin{align*}
  \left[ \alpha_X \circ \mu_X (\chi) \right] (U)
  = \left[ \mu_X (\chi) \right]_{\mathbbm{1}_U}
  = \chi_{\lambda \xi . \xi_{\mathbbm{1}_U}},
\end{align*}
but also
\begin{align*}
  \left[ \nu_X \circ \alpha_{\valu(X)} \circ \orch(\mathbb{C}, \mathbb{C}, \alpha_X)(\chi) \right] (U)
  &= \int_{\mu \in \valu(X)} \mu(U) d\left( \alpha_{\valu(X)} \circ \orch(\mathbb{C},\mathbb{C},\alpha_X)(\chi) \right) \\
  &= \left[ \orch(\mathbb{C},\mathbb{C},\alpha_X)(\chi) \right]_{\lambda \mu . \mu(U)}
    = \chi_{\lambda \xi . \xi \left[ \lambda \mu . \mu(U) \right] ( \alpha_X(\xi) )} \\
  &= \chi_{\lambda \xi . \left[ \alpha_X (\xi) \right] (U)}
    = \chi_{\lambda \xi . \xi_{\mathbbm{1}_U}}.
\end{align*}
The equality of these two expressions, together with the naturality of \( \alpha_X \), shows that above diagram commutes, and that the natural isomorphism \(\alpha_X\) is compatible with the multiplication.
All in all, we arrive at the main result of this section.
\begin{theorem}
  The monad \(\orch(\mathbb{C},\mathbb{C},-)\) is isomorphic to the probabilistic powerdomain monad \(\valu\).
\end{theorem}
\end{document}